\documentclass{CSML}
\pdfoutput=1

\usepackage{lastpage}
\newcommand{\dOi}{14(1:14)2018}
\lmcsheading{}{1--\pageref{LastPage}}{}{}%
{Nov.~24, 2010}{Feb.~06, 2018}{}


\usepackage{hyperref}
\hypersetup{hidelinks}
\usepackage{amssymb,latexsym,amsmath,bbm,stmaryrd}

\newcommand{\nada}{\varnothing}  
\newcommand{\sse}{\subseteq}
\newcommand{\ouriff}{\text{ iff }}

\newcommand{\Ran}{\mathsf{Ran}}

\newcommand{\funP}{\mathsf{P}}

\newcommand{\isdef}{\mathrel{:=}}
\newcommand{\rst}[1]{{\upharpoonright}{#1}}
\newcommand{\sz}[1]{|#1|}
\newcommand{\wt}[1]{\mathsf{w}(#1)}
\newcommand{\last}{\mathsf{last}}

\newcommand{\ul}[1]{\underline{#1}}

\newcommand{\Prop}{\ensuremath{\mathtt{X}}}        
\newcommand{\Propvar}{\ensuremath{\mathtt{PROP}}}        

\newcommand{\MSO}{\ensuremath{\mathtt{MSO}}} 
\newcommand{\ML}{\ensuremath{\mathtt{ML}}} 
\newcommand{\muML}{\ensuremath{\mu\ML}}    

\newcommand{\Latt}{\mathtt{Latt}}          
\newcommand{\LitC}{\mathtt{CL}}          

\newcommand{\MLone}{\ensuremath{\mathtt{1ML}}} 
\newcommand{\DMLone}{\ensuremath{\mathtt{1DML}}} 

\newcommand{\FV}[1]{\mathit{FV}(#1)}
\newcommand{\BV}[1]{\mathit{BV}(#1)}
\newcommand{\Sfor}{\mathit{Sfor}}

\newcommand{\sforeq}{\trianglelefteqslant}
\newcommand{\Act}{\mathit{Act}}

\newcommand{\bv}{\bigvee}
\newcommand{\bw}{\bigwedge}

\newcommand{\bwsmall}{\textstyle{\bw}}

\newcommand{\ybullet}{\land}

\newcommand{\dia}{\Diamond}
\newcommand{\nb}{\nabla}
\newcommand{\hs}{\heartsuit}

\renewcommand{\phi}{\varphi} 
\newcommand{\isbnf}{\mathrel{::=}}
\newcommand{\divbnf}{\mid}

\newcommand{\mathstr}[1]{\mathbb{#1}}

\newcommand{\bbA}{\mathstr{A}}
\newcommand{\bbB}{\mathstr{B}}

\newcommand{\bbD}{\mathstr{D}}
\newcommand{\bbG}{\mathbb{G}}
\newcommand{\bbP}{\mathstr{P}}
\newcommand{\bbS}{\mathstr{S}}
\newcommand{\bbT}{\mathstr{T}}

\newcommand{\bis}{\mathrel{\mathchoice%
{\raisebox{.3ex}{$\,
  \underline{\makebox[.7em]{$\leftrightarrow$}}\,$}}%
{\raisebox{.3ex}{$\,
  \underline{\makebox[.7em]{$\leftrightarrow$}}\,$}}%
{\raisebox{.2ex}{$\,
  \underline{\makebox[.5em]{\scriptsize$\leftrightarrow$}}\,$}}%
{\raisebox{.2ex}{$\,
  \underline{\makebox[.5em]{\scriptsize$\leftrightarrow$}}\,$}}}}

\newcommand{\sat}{\Vdash}
\newcommand{\satone}{\sat^{1}}
\newcommand{\mng}[1]{[\![ #1 ]\!]}

\newcommand{\mop}[1]{\langle #1 \rangle}

\newcommand{\Aut}{\mathit{Aut}}
\newcommand{\IAut}{\mathit{IAut}}

\newcommand{\init}[1]{\langle#1\rangle}
\newcommand{\ai}{a_{I}}
\newcommand{\act}{\lhd}
\newcommand{\biact}{\bowtie}
\newcommand{\blw}{\sqsubset}
\newcommand{\blweq}{\sqsubseteq}

\newcommand{\tr}{\mathtt{tr}}

\newcommand{\AG}{\mathcal{A}}
\newcommand{\EG}{\mathcal{E}}

\newcommand{\eloi}{\exists}
\newcommand{\abel}{\forall}

\newcommand{\Win}{\mathsf{Win}}
\newcommand{\PM}[1]{\mathrm{PM}_{#1}}

\newcommand{\Si}{\Sigma}

\newcommand{\Th}{\Theta}

\newcommand{\Om}{\Omega}
\newcommand{\al}{\alpha}
\newcommand{\be}{\beta}
\newcommand{\de}{\delta}
\newcommand{\ga}{\gamma}

\newcommand{\ka}{\kappa}

\newcommand{\si}{\sigma}
\newcommand{\om}{\omega}

\newcommand{\Frag}{F}
\newcommand{\pX}{\mathcal{X}}

\newtheorem{claim2}{\sc Claim}
\newenvironment{claimyv}{\begin{claim2}\rm}{\end{claim2}\rm}
\newenvironment{claimfirstyv}{\setcounter{claim2}{0}
               \begin{claim2}\rm}{\end{claim2}\rm}

\newenvironment{pfclaim}{\begin{trivlist}\item[]{\sc Proof of
Claim}}{\hfill {\mbox{$\blacktriangleleft$}}\end{trivlist}}

\begin{document}

\title[Some model theory for the modal $\mu$-calculus]{Some model theory for the modal $\mu$-calculus:
\\syntactic characterisations of semantic properties}
\titlecomment{The research of both authors has been made possible by  
   VICI grant 639.073.501 of
   the Netherlands Organization for Scientific Research (NWO)}

\author[G.~Fontaine]{Ga\"elle Fontaine}	
\address{Institute for Logic, Language and Computation, Universiteit van Amsterdam}	
\email{gaelle2l@gmail.com}  
\email{y.venema@uva.nl}  

\author[Y.~Venema]{Yde Venema}	
\address{\vskip-6pt}	



\begin{abstract}
This paper contributes to the theory of the modal $\mu$-calculus by proving
some model-theoretic results.
More in particular, we discuss a number of semantic properties pertaining
to formulas of the modal $\mu$-calculus.
For each of these properties we provide a corresponding syntactic fragment,
in the sense that a $\mu$-formula $\xi$ has the given property iff it is
equivalent to a formula $\xi'$ in the corresponding fragment.
Since this formula $\xi'$ will always be effectively obtainable from $\xi$, as a
corollary, for each of the properties under discussion, we prove that it is
decidable in elementary time whether a given $\mu$-calculus formula
has the property or not.

The properties that we study all concern the way in which the meaning of a 
formula $\xi$ in a model depends on the meaning of a single, fixed proposition
letter $p$.
For example, consider a formula $\xi$ which is monotone in $p$; such a formula 
a formula $\xi$ is called \emph{continuous} (respectively, \emph{fully 
additive}), if in addition it satisfies the property that, if $\xi$ is true at
a state $s$ then there is a finite set (respectively, a singleton set) $U$ such
that $\xi$ remains true at $s$ if we restrict the interpretation of $p$ to the 
set $U$.
Each of the properties that we consider is, in a similar way, associated with 
one of the following special kinds of subset of a tree model: singletons, finite 
sets, finitely branching subtrees, noetherian subtrees (i.e., without infinite
paths), and branches.

Our proofs for these characterization results will be automata-theoretic in
nature; we will see that the effectively defined maps on formulas are in fact
induced by rather simple transformations on modal automata.
Thus our results can also be seen as a contribution to the model theory of 
modal automata.
\end{abstract}

\maketitle


\section{Introduction}
\label{sec:intro}

This paper is inspired by the model-theoretic tradition in logic of linking 
semantic properties of formulas to syntactic restrictions on their shape.
Such correspondences abound in the model theory of classical (propositional 
or first-order) logic~\cite{chan:mode73}.
Well-known preservation results are the {\L}os-Tarski theorem stating that 
the models of a formula $\phi$ are closed under taking submodels iff $\phi$
is equivalent to a universal formula, or Lyndon's theorem stating that a
formula $\phi$ is monotone with respect to the interpretation of a relation
symbol $R$ iff $\phi$ is equivalent to a formula in which all occurrences 
of $R$ are positive.
In the last example, the semantic property is monotonicity, and the
syntactic restriction is positivity.

Our aim here is to establish such correspondences in the setting of the 
(modal) $\mu$-calculus $\muML$, the extension of modal logic with least
and greatest fixpoint operators~\cite{brad:moda06}.
Since its introduction by Kozen in the 1980s~\cite{koze:resu83}, the modal
$\mu$-calculus has found increasing recognition as an important and natural
formalism for specifying properties of processes.
The main reason for this is that $\muML$, just like basic modal logic,
strikes a very favorable balance between expressiveness and computational 
feasibility.
In particular, it was proved by Janin \& Walukiewicz~\cite{jani:expr96}
that $\muML$ is expressively complete for those monadic second-order 
properties that are bisimulation invariant, so that most, if not all,
interesting properties of processes can be specified in the language.
On the other hand, despite this large expressive power, the computational
complexity of the satisfiability problem for $\muML$ can be solved in 
exponential time~\cite{emer:comp88}, which is basically the same as for 
any extension of modal logic with fixpoint connectives.
Other attractive features of the modal $\mu$-calculus include a semantics 
that can both be presented in a compositional, algebraic format and in 
intuitive, game-theoretic terms;
a natural axiomatization, formulated by Kozen~\cite{koze:resu83} and proven
to be complete for the semantics, partially by Kozen himself and fully
by Walukiewicz~\cite{walu:comp00};
a tight link with automata theory, established by Janin \&
Walukiewicz~\cite{jani:auto95} and Wilke~\cite{wilk:alte01} via the introduction 
of modal automata, of non-deterministic respectively alternating type, as 
automata-theoretic equivalents to $\muML$-formulas;
and a certain internal expressive balance, witnessed by the property of
uniform interpolation, proved by D'Agostino \& 
Hollenberg~\cite{dago:logi00}.

With all these positive results, the modal $\mu$-calculus has become the
canonical modal process logic, and it seems worth while to develop its 
model theory in full detail.
Some results are known: in particular, preservation results, similar to the
{\L}os-Tarski and Lyndon theorems, have been shown for
the $\mu$-calculus by  D'Agostino \& Hollenberg~\cite{dago:logi00}.
However, in the intended semantics of $\muML$, where models represent
computational processes, and accessibility relations, bisimulations, and 
trees play an important role, there are some specific properties of interest
that have not been studied in classical model theory.
Important examples that we will study here include the properties of 
\emph{full} and \emph{complete additivity} and of \emph{continuity} with respect 
to some fixed propositional variable $p$.
\medskip

\noindent\textbf{Full and complete additivity}
To define the properties of \emph{full} and \emph{complete additivity}, recall 
that in each Kripke model $\bbS = (S,R,V)$ we may formalize the dependence of 
the meaning of a $\mu$-formula $\xi$ in $\bbS$ on a fixed proposition letter 
$p$ as a map
\[
\xi_{p}^{\bbS}:  \funP S \to  \funP S,
\]
defined by $\xi_{p}^{\bbS}(X) \isdef \{ s \in S \mid \bbS[p \mapsto X] \sat \xi 
\}$, where $\bbS[p \mapsto X] = (S,R,V[p \mapsto X]$ is the model obtained from 
$\bbS$ by modifying $V$ so that $V(p) = X$.
Then a formula $\xi$ is fully additive in $p$ if for each model $\bbS$, 
the operation $\xi_{p}^{\bbS}$ distributes over arbitrary unions:
\[
\xi_{p}^{\bbS}\Big(\bigcup \mathcal{X} \Big) = 
\bigcup \Big\{ \xi_{p}^{\bbS}(X) \mid X \in \mathcal{X} \Big\},
\]
for any collection $\mathcal{X}$ of subsets of $S$, whereas we say that $\xi$ 
is completely additive if $\xi_{p}^{\bbS}$ distributes over arbitrary 
\emph{non-empty} unions.
Clearly, the difference between the two notions is that for full additivity 
we require the map $\xi_{p}^{\bbS}$ to be \emph{normal}, that is, 
$\xi_{p}^{\bbS}(\bigcup\nada) = \bigcup \nada$, or equivalently, 
$\xi_{p}^{\bbS}(\nada) = \nada$.

It is not difficult to see that full additivity is equivalent to requiring that 
for all $X \sse S$,
\[
\xi_{p}^{\bbS}\left(X\right) = 
\bigcup \big\{ \xi_{p}^{\bbS}(\{x\}) \mid x\in X \big\},
\]
or, putting it yet differently, to the following constraint, for every pointed
Kripke model $(\bbS,s)$:
\[
\bbS,s \sat \xi \ouriff
\bbS[p\mapsto \{x\}],s \sat \xi, \mbox{ for some } x \in V(p).
\]
This is a very natural property in the context of modal logic, since for 
any set $S$, there is a 1--1 correspondence between the fully additive
functions on $\funP S$ and the binary (accessibility) relations on $S$, the
relation associated with the map $f$ given as $R_{f} := \{ (s,s') \mid s 
\in f(\{s'\}) \}$.
Related to this, a more specific reason for studying full additivity is given
by its pivotal role in the characterization of the fragments of 
first- and monadic second-order logic of formulas that are \emph{safe for 
bisimulations} (for a detailed discussion of this notion and its relation to 
full additivity we refer to section~\ref{sec:ca}).
Syntactic characterizations of the formulas that are fully additive in 
a given proposition letter $p$, were obtained by van Benthem~\cite{bent:expl96}
in the setting of basic (i.e., fixpoint-free) modal logic and by 
Hollenberg~\cite{holl:logi98} in the setting of the modal $\mu$-calculus.
As an alternative to Hollenberg's result, we shall give a different syntactic
fragment characterizing full additivity.
More precisely, we will prove that a $\muML$-formula is fully additive
in $p$ if it is equivalent to a formula in the fragment $\muML^{A}_{\{p\}}$,
where we define, for a set $P$ of propositional variables, the fragment 
$\muML^{A}_{P}$ of the modal $\mu$-calculus by the following grammar:
\[
\phi \isbnf p 
\divbnf \phi\lor\phi \divbnf \phi\land\psi 
\divbnf \dia\phi
\divbnf \mu x. \phi'.
\]
Here $p$ belongs to $P$, $\psi$ is a $P$-free formula (i.e., without free 
occurrences of proposition letters in $P$), and the formula $\phi'$ belongs to
the fragment $\muML^{A}_{P \cup \{ x \}}$.
For the property of complete additivity we obtain a very similar 
characterization.
\medskip

\noindent\textbf{Continuity}
Another important property featuring in this study is that of 
\emph{continuity}.
We shall call a formula $\xi$ \emph{continuous} in a proposition letter $p$
if for all $\bbS$, and all $X \sse S$,
\[
\xi_{p}^{\bbS}\big(X\big) = 
\bigcup \Big\{ \xi_{p}^{\bbS}(F) \mid F \sse_{\om} X \Big\}.
\]
That is, where the meaning of $\xi_{p}(X)$ in the case of $\xi$ being fully
additive in $p$ depends on the singleton subsets of $X$, here the meaning of
$\xi_{p}(X)$ depends on the \emph{finite} subsets of $X$.

What explains both the name and the importance of this property, is its
equivalence to \emph{Scott continuity}: for any model $\bbS$, the map 
$\xi_{p}^{\bbS}$ is continuous (in our sense) iff it is continuous with
respect to the Scott topology on the powerset algebra. 
Scott continuity is of key importance in many areas of theoretical computer
sciences where ordered structures play a role, such as domain theory (see, 
e.g.,~\cite{abra:doma94}). 
Another motivation concerns the relation between continuity and another 
property of computational interest, \emph{constructivity}. 
A monotone formula $\xi$ is constructive in a proposition letter $p$ if for 
each model $\bbS$, the least fixpoint of the map $\xi_{p}^{\bbS}$ is reached
in at most $\om$ approximation steps. 
Locally, this means that a state satisfies a least fixpoint formula if it 
satisfies one of its finite approximants. 
While the exact relation between the two properties is not clearly
understood, it is well-known that continuity strictly implies constructivity,
and we believe that in a sense continuity can be considered as the most 
natural property to approximate constructivity syntactically.
A full discussion of the notion of contnuity and its relation with
constructivity can be found in section~\ref{sec:cont}.

As one of the main results of this paper we will show that a $\mu$-formula
$\xi$ is continuous in $p$ iff it is equivalent to a formula $\phi$ in the
syntactic fragment $\muML^{C}_{\{p\}}$, where we define, for any set $P$ of 
propositional variables, the set $\muML^{C}_{P}$ by a mutual induction based 
on the following grammar:
\[
\phi \isbnf 
   p \divbnf \psi 
   \divbnf \phi\lor\phi \divbnf \phi\land\phi 
   \divbnf \dia \phi 
   \divbnf \mu x. \phi',
\]
where $\psi$ is a $P$-free $\mu$-formula, and $\phi'$ belongs to the fragment
$\muML^{A}_{P \cup \{ x \}}$.
A first presentation of this result was given by the first author 
in~\cite{font:cont08}; here we will give an alternative and more insightful 
proof of this result.
\medskip

\noindent\textbf{Aim}
The general purpose of this paper is to prove, in a uniform framework, 
syntactic characterization results corresponding to a number of semantic 
properties, including full and complete addivity, and continuity.
What these properties have in common 
is that they all concern the dependence of the truth of the formula at stake, 
on a single proposition letter ---
it will be convenient to fix this letter from now on as `$p$'.
More precisely, given a Kripke model $\bbS= (S,R,V)$ (which in some
cases we require to be a tree), let $\pX_{\bbS}$ uniformly denote a certain
class of subsets of $S$, such as singletons, finite sets, paths, finitely 
branching subtrees, etc.
Let $\xi$ be a formula of the modal $\mu$-calculus.
Assuming that $\xi$ is monotone in $p$, we say that it has the 
\emph{$\pX$-property with respect to $p$} if for every pointed model $(\bbS,s)$ 
we have
\begin{equation}
\label{eq:pX}
\bbS,s \sat \xi \ouriff
\bbS[p\rst{X}],s \sat \xi
\text{ for some } X \in \pX_{\bbS},
\end{equation}
where $\bbS[p\rst{X}] \isdef \bbS[p \mapsto V(p) \cap X]$ is the model obtained
from $\bbS$ by \emph{restricting} $p$ to the set $X$.
In the examples that we shall consider, $\pX_{\bbS}$ consists of the following 
sets\footnote{%
   There are a few subtleties here.
   For instance the property might only make sense when we investigate
   \eqref{eq:pX} on certain (tree) models, we might want the set $\pX$ to 
   consist of subsets of $V(p)$ or not, etc.
}:
  
(a) singletons inside $V(p)$, leading to the property of full additivity;

(b) singletons, leading to the property of \emph{complete additivity}\footnote{%
   In the formulation \eqref{eq:pX}, the two versions of additivity differ in
   that for full additivity we require the singletons in $\pX$ to be subsets
   of $V(p)$, for complete additivity we do not.
   };
   
(c) finite sets, leading to the property of continuity;

(d) finitely branching subtrees, leading to the \emph{finite width} property;

(e) noetherian subtrees (i.e., without infinite paths), 
   leading to the \emph{finite depth property};

(f) branches, leading to the \emph{single branch property}.

\noindent
Clearly, there are some interesting relations between some of these
properties. 
For instance, both full and complete additivity imply continuity.
Also, continuity can be seen as the combination of a `horizontal' and a 
`vertical' component: the finite width property and the finite depth 
property, respectively.
The latter equivalence will be put to good use in the paper.

The above abstract presentation allows us to summarize our results in a
concise and uniform manner.
Basically, for each instance (a--f) of $\pX$ we present a \emph{syntactic
characterization} of the $\pX$-property, in the form of a syntactically 
defined fragment $\muML^{X}_{p}\sse\muML$ such that a $\mu$-formula has the
$\pX$-property with respect to $p$ iff it is equivalent to a formula in
$\muML^{X}_{p}$.
Since monotonicity forms part of the definition of each of the properties
(a--f), to facilitate the other proofs we first prove a slightly stronger
version of D'Agostino \& Hollenberg's Lyndon theorem; in the remainder of
this introduction it will be convenient to let $\pX$ also cover the property 
of monotonicity.
\medskip

\noindent\textbf{Main results \& proof method}
Our proofs, though different in each case, follow a uniform method, which goes
back to the proofs of Janin \& Walukiewicz~\cite{jani:expr96} and D'Agostino \&
Hollenberg~\cite{dago:logi00}.
For each property $\pX$, we will exhibit an explicit translation which, given
a $\muML$-formula $\xi$, computes a formula $\xi^{X} \in \muML^{X}_{p}$ such
that 
\begin{equation}
\label{eq:1}
\xi \text{ has the $\pX$-property}
\ouriff 
\xi \text{ is equivalent to } \xi^{X}.
\end{equation}
Since in each case the translation $(\cdot)^{X}$ is effectively computable, 
and the equivalence problem for two given $\muML$-formulas is decidable, as 
a nice corollary of \eqref{eq:1} we obtain various new \emph{decidability} 
results.
With $\pX$-\textsc{prop} being the problem whether a given $\muML$-formula has 
the $\pX$-property, for each $\pX$ as discussed we will show that 
\begin{equation}
\label{eq:2}
\text{the problem $\pX$-\textsc{prop} is decidable.}
\end{equation}
Considering that our main interest here is model-theoretic, we have not 
undertaken an in-depth study of the computational \emph{complexity} of the
$\pX$-\textsc{prop} problems.
What we can say is that each of our translations $(\cdot)^{X}$ will be based 
on the composition of simple transformations, each of which constructs an output 
structure of size at most exponential in the size of the input. 
From this it follows that all of the $\pX$-\textsc{prop} problems that we study 
can be solved in \emph{elementary} time.

While in almost all cases pure logic-based proofs for our results are possible
(and have been given in the dissertation of the first author~\cite{font:moda10}),
the proofs we provide in this paper are automata-theoretic in nature.
Automata for the modal $\mu$-calculus were introduced by Janin \& 
Walukiewicz~\cite{jani:auto95} under the name of \emph{$\mu$-automata}; where
these devices are non-deterministic in nature, Wilke~\cite{wilk:alte01} came
up with an alternating variant.
We will use both alternating and non-deterministic devices here, under the 
names of, respectively, \emph{modal} and \emph{disjunctive modal} automata.
The particular shape of our automata is logic-based: the transition map of our
structures maps states of the automata to so-called \emph{one-step formulas},
and many of our proofs are based on syntactic manipulations on these very
simple modal formulas.
In some sense then, our paper is also a contribution to the model theory of 
modal automata.

Finally, we have formulated our results in the setting of the
\emph{mono-modal} $\mu$-calculus; that is, the language that we consider
has one diamond only, and correspondingly the Kripke models have only one
accessibility relation.
This restriction is solely for the purpose of simplifying the presentation
of our results and proofs.
We want to stress, however, that all of the results in this paper can be
generalized to the setting of the polymodal $\mu$-calculus, with no conceptual
and little technical complication.
\medskip

\noindent\textbf{Overview}
The paper is organized as follows. 
In order to fix our terminology and notation, we give a review of the syntax
and semantics of the modal $\mu$-calculus in the following section, and in
section~\ref{sec:ma} we introduce the modal automata that we will be working 
with.
In section~\ref{sec:fs} we make some first model-theoretic steps, proving the 
characterization of monotonicity, and introducing an automata-theoretic
construction that will be used in the other parts of the paper.
In the subsequent three sections we discuss 
the finite width property (section~\ref{sec:fw}),
the single branch property (section~\ref{sec:sb}) and 
the finite depth property (section~\ref{sec:fd}).
After that we arrive the most important parts of the paper, viz.,
section~\ref{sec:cont} on continuity and 
section~\ref{sec:ca} on full and complete additivity.
In the final section of the paper we draw some conclusions, discuss some 
related results, and list some open problems for future research.

\section{Preliminaries}
\label{sec:prel}

In this paper we assume familiarity with the syntax and semantics of the modal 
$\mu$-calculus, as presented in for
instance~\cite{koze:resu83,arno:rudi01,grae:auto02,brad:moda06,vene:lect12}, 
and with the basic notions concerning infinite games~\cite{grae:auto02}.
Here we fix some notation and terminology.

\begin{conv}
Throughout the text we fix an infinite set $\Propvar$ of propositional variables,
of which we often single out a finite subset $\Prop$.
\end{conv}

\subsection{Parity games}

\begin{defi}
\label{d:game}
A {\em parity game} is a tuple $\bbG = (G_{\eloi},G_{\abel},E,\Om)$ where 
$G_{\eloi}$ and $G_{\abel}$ are disjoint sets, and, with $G \isdef G_{\eloi} 
\cup G_{\abel}$ denoting the \emph{board} of the game, the binary relation 
$E \subseteq G^2$ encodes the moves that are \emph{admissible} to the respective 
players, and the \emph{priority function} $\Om: G \to \om$, which is 
required to be of finite range, determines the \emph{winning condition}
of the game.
Elements of $G_{\eloi}$ and $G_{\abel}$ are called \emph{positions} for the 
players $\eloi$ and $\abel$, respectively; given a position $p$ for player 
$\Pi \in \{ \eloi, \abel\}$, the set $E[p]$ denotes the set of \emph{moves}
that are \emph{legitimate} or \emph{admissible to} $\Pi$ at $p$.
In case $E[p] = \nada$ we say that player $\Pi$ \emph{gets stuck} at $p$.

An \emph{initialized board game} is a pair consisting of a board game $\bbG$
and a \emph{initial} position $p$, usually denoted as $\bbG@p$.
\end{defi}

\begin{defi}
\label{d:match}
A {\em match} of a graph game $\bbG = (G_{\eloi},G_{\abel},E,\Om)$ is a (finite 
or infinite) path through the graph $(G,E)$.
Such a match $\Si$ is called \emph{partial} if it is finite and $E[\last\Si]
\neq\nada$, and \emph{full} otherwise.
We let $\PM{\Pi}$ denote the collection of partial matches $\Si$ ending in a 
position $\last(\Si) \in G_{\Pi}$, and define $\PM{\Pi}@p$ as the set of 
partial matches in $\PM{\Pi}$ starting at position $p$.

The \emph{winner} of a full match $\Si$ is determined as follows.
If $\Si$ is finite, it means that one of the two players got stuck at the 
position $\last(\Si)$, and so this player looses $\Si$, while the opponent
wins.
If $\Si = (p_{n})_{n\in\om}$ is infinite, we declare its winner to be $\eloi$ 
if the maximum value occurring infinitely often in the stream 
$(\Om p_{n})_{n\in\om}$ is even.
\end{defi}

\begin{defi}
A \emph{strategy} for a player $\Pi \in \{ \eloi,\abel \}$ is a map $f:
\PM{\Pi} \to G$.
A strategy is \emph{positional} if it only depends on the last position of a 
partial match, i.e., if $f(\Si) = f(\Si')$  whenever $\last(\Si) = 
\last(\Si')$; such a strategy can and will be presented as a map $f: 
G_{\Pi} \to G$.

A match $\Si = (p_{i})_{i<\kappa}$ is \emph{guided} by a $\Pi$-strategy 
$f$ if $f(p_{0}p_{1}\ldots p_{n-1}) = p_{n}$ for all $n<\kappa$ 
such that $p_{0}\ldots p_{n-1}\in \PM{\Pi}$.
A position is \emph{reachable} by a strategy $f$ is there is an $f$-guided
match $\Si$ of which $p$ is the last position.
A $\Pi$-strategy $f$ is \emph{legitimate} in $\bbG@p$ if the moves that it
prescribes to $f$-guided partial matches in $\PM{\Pi}@p$ are always
admissible to $\Pi$, and \emph{winning for $\Pi$} in $\bbG@p$ if in addition
all $f$-guided full matches starting at $p$ are won by $\Pi$.

A position $p$ is a \emph{winning position} for player $\Pi \in \{ \eloi, \abel 
\}$ if $\Pi$ has a winning strategy in the game $\bbG@p$; the set of these
positions is denoted as $\Win_{\Pi}$.
The game $\bbG = (G_{\eloi},G_{\abel},E,\Om)$ is \emph{determined} if every
position is winning for either $\eloi$ or $\abel$.
\end{defi}

When defining a strategy $f$ for one of the players in a board game, we can 
and in practice will confine ourselves to defining $f$ for partial matches 
that are themselves guided by $f$.

The following fact, independently due to Emerson \& Jutla~\cite{emer:tree91}
and Mostowski~\cite{most:game91}, will be quite useful to us.

\begin{fact}[Positional Determinacy]
\label{f:pdpg}
Let $\bbG = (G_{\eloi},G_{\abel},E,\Om)$ be a parity game.
Then $\bbG$ is determined, and both players have positional winning strategies.
\end{fact}

In the sequel we will often refer to a `positional winning strategy' for one
of the players in a parity game.
With this we mean any positional strategy which is winning for that player when
starting at any of his/her winning positions.

\subsection{Structures}

\begin{defi}
Given a set $S$, an \emph{$A$-marking} on $S$ is a map $m: S \to \funP A$; 
an \emph{$A$-valuation} on $S$ is a map $V: A \to \funP S$.
Any valuation $V : A \to 
\funP S$ gives rise to its \emph{transpose marking} $V^\dagger : S \to \funP A$
defined by $V^\dagger(s) := \{a \in A \mid s \in V(a) \}$, and dually each 
marking gives rise to a valuation in the same manner. 
\end{defi}

Since markings and valuations are interchangeable notions, we will often
switch from one perspective to the other, based on what is more convenient in 
context.

\begin{defi}
A \emph{Kripke structure} over a set $\Prop$ of proposition letters is a triple 
$\bbS = (S,R,V)$ such that $S$ is a set of objects called \emph{points}, $R 
\sse S \times S$ is a binary relation called
the \emph{accessibility} relation, and $V$ is an $\Prop$-valuation on $S$.
A \emph{pointed} Kripke structure is a pair $(\bbS,s)$ where $s$ is a point of 
$\bbS$.

Given a Kripke structure $\bbS = (S,R,V)$, a propositional variable $x$ and a 
subset $U$ of $S$, we define $V[x\mapsto U]$ as the $\Prop \cup \{x\}$-valuation 
given by
\[
V[x \mapsto U](p) \isdef \left\{
\begin{array}{ll}
   V(p) & \text{ if } p \neq x 
\\ U    & \text{otherwise},
\end{array}\right.
\]
and we let $\bbS[x\mapsto U]$ denote the structure $(S,R,V[x\mapsto U])$.
The structure $\bbS[x\mapsto V(x) \cap U]$ is usually denoted as 
$\bbS[x\rst{U}]$.
\end{defi}

It will often be convenient to take a coalgebraic perspective on Kripke
structures.

\begin{defi}
Given a Kripke structure $\bbS = (S,R,V)$ over the set $\Prop$, we 
define its \emph{(coalgebraic) unfolding map} $\si_{\bbS}: S \to \funP\Prop 
\times \funP S$ given by $\si_{\bbS}(s) = (\si_{V}(s),\si_{R}(s))$, where
$\si_{V}(s) \isdef V^\dagger(s)$ and $\si_{R}(s) \isdef R[s]$ are the sets of,
respectively, the proposition letters true at $s$ and the successors of $s$.
We will write $\si$ rather than $\si_{\bbS}$ in case no confusion is likely.
\end{defi}

\begin{defi}
A \emph{path} through a Kripke structure $\bbS = (S,R,V)$ is a sequence 
$(s_i)_{i < \ka}$ such that $(s_i,s_{i+1}) \in R$ for all $i$ with $i+1 < \ka$; 
here $\ka \leq\om$ is the \emph{length} of the path.
We let $\sqsubseteq$ denote the prefix (initial segment) relation between paths,
and use $\sqsubset$ for the strict (irreflexive) version of $\sqsubseteq$.
\end{defi}

\begin{defi}
\label{d:bis}
Given two models $\bbS = (S,R,V)$ and $\bbS' = (S',R',V')$, a relation $Z\sse
S \times S'$ is a \emph{bisimulation} if it satisfies, for all $(s,s')\in Z$,
the conditions 

(prop) $s \in V(p)$ iff $s' \in V'(p)$, for all $q \in \Prop$;

(forth) for all $t \in \si_{R}(s)$ there is a $t' \in \si_{R'}(s')$ with $Ztt'$;
and

(back) for all $t' \in \si_{R'}(s')$ there is a $t \in \si_{R}(s)$ with $Ztt'$.

We say that $s$ and $s'$ are \emph{bisimilar}, notation $\bbS,s \bis \bbS',s'$
if there is some bisimulation $Z$ with $Zss'$.
A function $f: S \to S'$ is a \emph{bounded morphism} from $\bbS$ to $\bbS'$,
notation $f: \bbS \to \bbS'$, if its graph $\{ (s,f(s)) \mid s \in S \}$
is a bisimulation.
\end{defi}

\begin{defi}
The reflexive/transitive closure and the transitive closure of $R$ are 
denoted as $R^{*}$ and $R^{+}$, respectively; elements of the sets
$R^{*}[s]$ and $R^{+}[s]$ are called \emph{descendants} and \emph{proper
descendants} of $s$, respectively.

A pointed structure $(\bbS,s)$ is a \emph{tree} (with \emph{root} $s$) if $S =
R^{*}[s]$ and every state $t\neq s$ has a unique predecessor.
A \emph{branch} of a tree $(\bbS,s)$ is a maximal path through $\bbS$, starting
at the root.
In a tree model $\bbS$, a set $U \subseteq S$ is \emph{downward closed} if
for all $s \in U$, the predecessors of $s$ belongs to $U$.
A \emph{sibling} of a node $t$ in a tree is a node $t'\neq t$ with the same 
predecessor as $t$.

Given a number $\ka\leq\om$, a tree is \emph{$\ka$-expanded} if every node 
(apart from the root) has $\ka-1$ many bisimilar siblings (where $\om-1 =
\om$).
Given a pointed structure $(\bbS,s)$, its \emph{$\ka$-expansion} is the structure
$\bbS^{\ka}_{s} := (S',R',V')$, where $S'$ is the set of all finite sequences
$s_{0}k_{1}s_{1}\ldots k_{n}s_{n}$ ($n\geq0$) such that $s_{0}=s$ and
$k_{i}<\ka$, $s_{i}\in S$ and $s_{i-1}Rs_{i}$ for all $i>0$;
where $R' := \{ (s_{0}k_{1}s_{1}\ldots k_{n}s_{n}, s_{0}k_{1}s_{1}\ldots
k_{n}s_{n}kt) \mid k<\ka \text{ and } s_{n}Rt \}$;
and where $V'(p) := \{ s_{0}k_{1}s_{1}\ldots k_{n}s_{n} \mid s_{n} \in V(p) \}$.
The 1-expansion of a pointed model $(\bbS,s)$ is also called its
\emph{unravelling}.
\end{defi}

\begin{fact}
\label{f:bisimomega}
Fix an ordinal $\ka\leq\om$.
Given a pointed model $(\bbS,s)$, the structure $(\bbS^{\ka}_{s},s)$ is a 
$\ka$-expanded tree which is bisimilar to $(\bbS,s)$ via the canonical bounded
morphism mapping a path $s_{0}k_{1}s_{1}\ldots k_{n}s_{n}$ to its last element 
$s_{n}$.
\end{fact}

\subsection{Syntax}

\begin{defi}
\label{d:syn}
The language $\muML$ of the modal $\mu$-calculus is given by the following 
grammar:
\begin{equation}
\label{eq:mu-syn1}
\phi \isbnf p 
   \divbnf \neg\phi \divbnf \phi\lor\phi 
   \divbnf \dia\phi 
   \divbnf \mu x.\phi,
\end{equation}
where $p$ and $x$ are propositional variables, and the formation of the formula 
$\mu x.\phi$ is subject to the constraint that the variable $x$ is
\emph{positive} in $\phi$, i.e., all occurrences of $x$ in $\phi$ are in the
scope of an even number of negations.
Elements of $\muML$ will be called \emph{modal fixpoint formulas}, 
\emph{$\mu$-formulas}, or simply \emph{formulas}.
\end{defi}

We will often make the assumption that our formulas are in negation normal form.

\begin{defi}
\label{d:syntnnf}
A formula of the modal $\mu$-calculus is in \emph{negation normal form} if it 
belongs to the language given by the following grammar:
\[
\phi \isbnf p \divbnf \neg p \divbnf \phi\lor\phi \divbnf \phi\land\phi \divbnf
   \dia\phi \divbnf \Box\phi \divbnf \mu x.\phi \divbnf \nu x.\phi,
\]
where $p$ and $x$ are propositional variables, and the formation of the formulas
$\mu x.\phi$ and $\nu x.\phi$ is subject to the constraint that the variable $x$
is \emph{positive} in $\phi$.
We use the symbol $\eta$ to range over $\mu$ and $\nu$.
\end{defi}

\begin{conv}
In order to increase readability by reducing the number of brackets, we adopt 
some standard scope conventions.
We let the unary (propositional and modal) connectives, $\neg,\dia$ and $\Box$,
bind stronger than the binary propositional connectives $\land$, $\lor$ and 
$\to$, and use associativity to the left for the connectives $\land$ and $\lor$.
Furthermore, we use `dot notation' to indicate that the fixpoint operators 
preceding the dot have maximal scope. 
For instance, $\mu x. \neg p \lor \dia x \lor \nu y. q \land \Box y$ stands 
for $\mu x. \Big(
\big((\neg p) \lor (\dia x)\big) \lor \nu y. (q \land \Box y)
\Big)$
\end{conv}

We gather some definitions pertaining to formulas.

\begin{defi}
We let $\Sfor(\xi)$ denote the collection of subformulas of a formula $\xi$, 
defined as usual, and we write $\phi\sforeq\xi$ if $\phi$ is a subformula
of $\xi$.
The \emph{size} of $\xi$ is defined as its number of subformulas, 
$\sz{\phi} \isdef \sz{\Sfor(\xi)}$.
We write $\mathit{BV}(\xi)$ and $\mathit{FV}(\xi)$ for, respectively, the set 
of \emph{bound} and \emph{free variables} of a formula $\xi$.
We let $\muML(\Prop)$ denote the set of $\mu$-formulas of which all free 
variables belong to the set $\Prop$.

A $\mu$-formula $\xi$ is \emph{well-named} if $\mathit{BV}(\xi)\cap
\mathit{FV}(\xi) =\nada$, and with every bound variable $x$ of $\xi$ we may
associate a unique subformula of the form $\eta x.\delta$ (with $\eta \in
\{\mu,\nu\}$).
This unique subformula will be denoted as $\eta_{x}.\delta_{x}$, and we call 
$x$ a $\mu$-variable if $\eta_{x}=\mu$, and a $\nu$-variable if $\eta_{x}=
\nu$.
\end{defi}

As a convention, the free variables of a formula $\phi$ are denoted by the 
symbols $p,q,r, \ldots$, and referred to as \emph{proposition letters}, while
we use the symbols $x,y,z,\ldots$ for the bound variables of a formula.

\begin{defi}
\label{d:subst}
Let $\phi$ and $\{ \psi_{z} \mid z \in Z \}$ be modal fixpoint formulas, where
$Z \cap \BV{\phi} = \nada$.
Then we let 
\[
\phi[\psi_{z}/z \mid z \in Z]
\]
denote the formula obtained from $\phi$ by simultaneously substituting each
formula $\psi_{z}$ for $z$ in $\phi$ (with the usual understanding that no 
free variable in any of the $\psi_{z}$ will get bound by doing so).
In case $Z$ is a singleton $z$, we will simply write $\phi[\psi_{z}/z]$, or 
$\phi[\psi]$ if $z$ is clear from context.
\end{defi}

\begin{defi}
Let $\xi$ be a well-named $\mu$--formula.
The \emph{dependency order} $<_{\xi}$ on the bound variables of $\xi$ is defined
as the least strict partial order such that $x <_{\xi} y$ if $\delta_{x}$ is a 
subformula of $\delta_{y}$.

We define a map $\Act: \Sfor(\xi) \to \funP(\FV{\xi})$ assigning 
to each subformula $\phi \sforeq \xi$ the least set $\Act(\phi)$ such that 
$\Act(p) = \{ p \}$ if $p \in \FV{\xi}$, 
$\Act(\phi \ast \psi) = \Act(\phi) \cup \Act(\psi)$ if $\ast \in \{\land,\lor\}$,
$\Act(\hs\phi) = \Act(\phi)$ if $\hs \in \{\dia,\Box\}$,
$\Act(\eta x.\phi) = \Act(\phi)$ if $\eta \in \{ \mu,\nu\}$,
and $\Act(x) = \Act(\de_{x})$ if $x \in \BV{\xi}$.
If $p \in \Act(\phi)$ we say that $p$ is \emph{active} in $\phi$ (relative to 
$\xi$).
\end{defi}

\subsection{Semantics}

\begin{defi}
By induction on the complexity of modal fixpoint formulas, we define a meaning
function $\mng{\cdot}$, which assigns to a formula $\phi \in \muML(\Prop)$ 
its \emph{meaning} $\mng{\phi}^{\bbS} \sse S$ in any Kripke structure $\bbS = 
(S,R,V)$ over $\Prop$.
The clauses of this definition are standard:
\begin{eqnarray*}
   \mng{p}^{\bbS} &\isdef& V(p)
\\ \mng{\neg\phi}^{\bbS} &\isdef& S \setminus \mng{\phi}^{\bbS}
\\ \mng{\phi\lor\psi}^{\bbS} &\isdef& \mng{\phi}^{\bbS} \cup \mng{\psi}^{\bbS}
\\ \mng{\dia\phi}^{\bbS} &\isdef& 
     \{ s \in S \mid R[s] \cap \mng{\phi}^{\bbS} \neq \nada \}
\\ \mng{\mu x.\phi}^{\bbS} &\isdef& 
     \bigcap \{ U \in \funP S \mid \mng{\phi}^{\bbS[x\mapsto U]}\sse U \}.
\end{eqnarray*}
If a point $s \in S$ belongs to the set $\mng{\phi}^{\bbS}$, we write 
$\bbS,s \sat \phi$, and say that $\phi$ is \emph{true at} $s$ or \emph{holds} 
at $s$, or that $s$ \emph{satisfies} $\phi$.
Two formulas $\phi$ and $\psi$ are \emph{equivalent}, notation: $\phi \equiv 
\psi$, if $\mng{\phi}^{\bbS} = \mng{\psi}^{\bbS}$ for any structure $\bbS$.
\end{defi}

Throughout this paper we will rely on the \emph{bisimulation invariance} of 
the modal $\mu$-calculus.

\begin{fact}
\label{f:bisinv}
Let $(\bbS,s)$ and $(\bbS',s')$ be pointed Kripke structures such that 
$(\bbS,s) \bis (\bbS',s')$.
Then $\bbS,s \sat \phi$ iff $\bbS',s' \sat \phi$, for all $\phi \in \muML$.
\end{fact}

We will usually take a \emph{game-theoretic} perspective on the semantics of the
modal $\mu$-calculus.

\begin{defi}
Let $\bbS=(S,R,V)$ be a Kripke model and let $\xi$ be a formula in $\muML$.
We define the \emph{evaluation game} $\EG(\xi,\bbS)$ as the parity game 
$(G,E,\Om)$ of which 
the board consists of the set $\Sfor(\xi) \times S$, and the game graph (i.e.,
the partitioning of $\Sfor(\xi) \times S$ into positions for the two players, 
together with the set $E(z)$ of admissible moves at each position $z$), is 
given in Table~\ref{tb:1}.
Note that we do not assign a player to positions that admit a single move 
only.

\begin{table}[t]
\begin{center}
\begin{tabular}{|ll|c|l|}
\hline
\multicolumn{2}{|l|}{Position $z$} & Player  
& Admissible moves $E(z)$\\
\hline
     $(p,s)$        & with $p\in \FV{\xi}$ and $s \in V(p)$         
   & $\abel$ & $\nada$ 
\\   $(p,s)$        & with $p\in \FV{\xi}$ and $s \notin V(p)$      
   & $\eloi$ & $\nada$ 
\\   $(\neg p,s)$   & with $p\in \FV{\xi}$ and $s \in V(p)$    
   & $\eloi$ & $\nada$ 
\\   $(\neg p,s)$  & with $p\in \FV{\xi}$ and $s \notin V(p)$ 
   & $\abel$ & $\nada$ 
\\   $(x,s)$       & with $x\in \BV{\xi}$ 
   & - & $\{ (\delta_x,s) \}$ 
\\ \multicolumn{2}{|l|}{$(\phi \vee \psi,s)$}   & $\eloi$   
   & $\{ (\phi,s), (\psi,s) \}$ 
\\ \multicolumn{2}{|l|}{$(\phi \wedge \psi,s)$} & $\abel$ 
   & $\{ (\phi,s), (\psi,s) \}$ 
\\ \multicolumn{2}{|l|}{$(\eta x . \phi,s)$}    & - 
   & $\{ (\phi,s) \}$ 
\\ \multicolumn{2}{|l|}{$(\dia \phi,s)$}        & $\eloi$ 
   & $\{ (\phi,t) \mid sRt \}$ 
\\ \multicolumn{2}{|l|}{$(\Box \phi,s) $}       & $\abel$ 
   & $\{ (\phi,t) \mid sRt \}$ 
\\ \hline
\end{tabular}
\end{center}
\caption{The evaluation game $\EG(\xi,\bbS)$}
\label{tb:1}
\end{table}

To define the priority map $\Om$ of $\EG(\xi,\bbS)$, consider an infinite match
$\Si = (\phi_{n},s_{n})_{n\in\om}$, and let $\mathit{Inf}(\Si)$ denote the set
of (bound) variables that get unfolded infinitely often during the match.
This set contains a highest variable $x$ (with respect to the dependency order
$<_{\xi}$), and the winner of $\Si$ is $\eloi$ if $\eta_{x} = \mu$, and $\abel$ 
if $\eta_{x}=\nu$.
It is not difficult to define a priority map $\Om: (\Sfor \times S) \to \om$
that is compatible with this condition, but we do not need the details of the
precise definition.
\end{defi}

The following fact states the \emph{adequacy} of the game semantics.

\begin{fact}
\label{f:0a}
Let $\xi$ be a well-named formula of the modal $\mu$-calculus, and let 
$(\bbS,s)$ be some pointed Kripke structure.
Then
$\bbS,s \sat \xi \ouriff (\xi,s) \in \Win_{\eloi}(\EG(\xi,\bbS))$.
\end{fact}      

Finally, for our decidability results we will use the following result, due
to Emerson \& Jutla~\cite{emer:comp88}.

\begin{fact}
\label{f:dec}
The question whether a given $\muML$-formula $\xi$ is satisfiable can be 
decided in time exponential in the size of $\xi$.
As a corollary, the question whether two $\muML$-formula $\xi$ and $\xi'$ are 
equivalent is decidable in time exponential in the sum of the sizes of $\xi$
and $\xi'$.
\end{fact}      

Since this paper is about various syntactic fragments of the modal 
$\mu$-calculus, we gather here some specific formulas that will be discussed as 
(non-)instances of these fragments throughout the paper.

\begin{exa}
\label{ex:fmas}
Below we introduce a number of formulas, $\phi_{0}, \ldots, \phi_{6}$, and for
each listed formula we explain in words its meaning in an arbitrary tree model
$\bbS$ at the root $r$.

$\phi_{0} \isdef p$ simply states that $p$ holds at $r$;

$\phi_{1} \isdef \nu y. (q \land \dia y)$ expresses the existence of an
   infinite $q$-path at $r$;

$\phi_{2} \isdef p \land \nu y. (q \land \dia y)$ is the conjunction of
   $\phi_{0}$ and $\phi_{1}$;

$\phi_{3} \isdef \mu x. p \lor \dia\dia x$ expresses the existence of 
   a $p$-state at an even distance of $r$;

$\phi_{4} \isdef \mu x. p \lor \Box\Box x$ states that any path starting
   at $r$ has a $p$-state at even distance from $r$;

$\phi_{5} \isdef \nu y.\mu x. ((p \land \dia y) \lor \dia x)$ expresses
   the existence of a path with infinitely many $p$-nodes;

$\phi_{6} \isdef \mu x. (p \lor (\dia (q \land x) \land \dia (\neg q \land x)))$
      expresses the existence of a finite binary subtree rooted at $r$,
      of which all leaves satisfy $p$, and all inner nodes have both 
      a $q$- and a $\neg q$-successor.
\end{exa}

\section{Modal automata}
\label{sec:ma}

In this section we give a detailed introduction to the automata that feature in
our proofs.
Basically, we will be working with the guarded specimens of Wilke's alternating
tree automata~\cite{wilk:alte01}, and with the non-deterministic versions of 
these which correspond to the $\mu$-automata of Janin \& 
Walukiewicz~\cite{jani:expr96}.
As mentioned in the introduction, many of our proofs will be based on 
manipulating the one-step formulas featuring as the co-domain of the transition
map of a modal automaton. 

\subsection{One-step logic}

The transition map of our modal automata will be based on a so-called \emph{modal 
one-step language}, consisting of modal formulas of rank 1, built up from 
\emph{proposition letters} (which may have negative occurrences but must appear
unguarded) and \emph{variables} (which must appear both guarded and positively). 

\begin{defi}
\label{d:MLone}
Given a set $P$, we define the set $\Latt(P)$ of \emph{lattice terms} over $P$ 
through the following grammar:
\[
\pi \isbnf p 
   \divbnf \bot \divbnf \top 
   \divbnf \pi \land \pi \divbnf \pi \lor \pi,
\]
where $p \in P$.
Given two sets $\Prop$ and $A$, we define the set $\MLone(\Prop,A)$ of 
\emph{modal one-step formulas} over $A$ with respect to $\Prop$ inductively by
\[
\al
\isbnf p \divbnf \neg p
   \divbnf \dia \pi \divbnf \Box \pi
   \divbnf \bot \divbnf \top 
   \divbnf \al \land \al \divbnf \al \lor \al,
\]
with $p \in \Prop$ and $\pi \in \Latt(A)$.
\end{defi}

These formulas are naturally interpreted in so-called one-step models.

\begin{defi}
Fix sets $\Prop$ and $A$.
A \emph{one-step frame}  is a pair $(\mathtt{Y},S)$  where $S$ is any set, and
$\mathtt{Y} \subseteq \Prop$. 
A \emph{one-step model} is a triple $(\mathtt{Y},S,m)$ such that 
$(\mathtt{Y},S)$ is a one-step frame and $m: S \to \funP A$ is an 
\emph{$A$-marking} on $S$.
\end{defi}

Observe that with this definition, the coalgebraic representation of a Kripke 
structure $(S,R,V)$ can now be seen as a function $\si_{\bbS}$ mapping any state 
$s \in S$ to a one-step frame of which the carrier is a subset of $S$.
This plays a fundamental role in the acceptance game for modal automata, and 
explains the following \emph{one-step satisfaction relation} $\satone$ between 
one-step models and one-step formulas.

\begin{defi}
Fix a one-step model $(\mathtt{Y},S,m)$.
First we define the value $ \mng{\pi} \sse S$ of a lattice formula $\pi$ over 
$A$ by induction, setting $\mng{a} \isdef \{s \in S \mid a \in m(s)\}$ for $a 
\in A$, and treating the boolean connectives in the obvious manner.

The \emph{one-step satisfaction relation} $\satone$ is inductively defined
as follows. 
For the literals and modal operators we set:

- $(\mathtt{Y},S,m) \satone p$ \text{ iff } $p \in \mathtt{Y}$,

- $(\mathtt{Y},S,m) \satone \neg p $ \text{ iff } $p \notin \mathtt{Y}$,

- $(\mathtt{Y},S,m) \satone \Box \pi$ \text{ iff } $\mng{\pi} = S$,

- $(\mathtt{Y},S,m) \satone \Diamond \pi$ \text{ iff } $\mng{\pi} \neq \nada$,

\noindent
while we have the standard clauses for the boolean connectives.
\end{defi}

\subsection{Modal automata}

\begin{defi}
\label{d:modaut}
A \emph{modal $\Prop$-automaton} $\bbA$  is a triple $(A,\Th,\Om)$ where 
$A$ is a non-empty finite set of \emph{states}, $\Om : A \to \omega$ is the 
\emph{priority map}, while the \emph{transition map} 
$$\Th : A \to \MLone(\Prop,A)$$
maps states to one-step formulas.
An \emph{initialized} modal automaton is pair $(\bbA,a)$, usually denoted as 
$\bbA\init{a}$, consisting of a modal automaton $\bbA$ together with a 
designated \emph{initial} state $a$.
The classes of (initialized) modal automata over the set $\Prop$ are 
denoted as $\IAut(\Prop)$ and $\Aut(\Prop)$, respectively.
\end{defi}

The operational semantics of modal automata is defined in terms of 
\emph{acceptance games}.

\begin{defi}
\label{d:AG}
Let $\bbA = (A,\Th,\Om)$ and $\bbS = (S,R,V)$ be a modal $\Prop$-automaton and
a Kripke structure, respectively. 
The \emph{acceptance game} $\AG(\bbA,\bbS)$ for $\bbA$ with respect to $\bbS$ is
defined as the parity game given by the following table:

\begin{center}
\begin{tabular}{|l|c|l|l|}
\hline
    Position             & Player  
    &  Admissible moves  & Priority
\\ \hline
   $(a,s)\in A\times S$  & $\eloi$  
   & $\{m : \si_{R}(s) \to \funP A \mid \si(s),m \satone \Th(a) $\}
   & $\Om(a)$
\\ $m$                   & $\abel$ 
   & $\{(b,t)\mid b \in m(t) \}$
   & 0
\\
\hline
\end{tabular}
\end{center}

We say that $\bbA\init{a}$ \emph{accepts} the pointed structure $(\bbS,s)$ if
$(a,s)$ is a winning position in the acceptance game $\AG(\bbA,\bbS)$, and 
we write $\bbS,s \sat \bbA\init{a}$ to denote this.
\end{defi}

\begin{conv}
We will usually identify a match $\Si = (a_{0},s_{0})m_{0}(a_1,s_1)m_1(a_2,s_2)
m_2 \ldots$ of the acceptance game $\AG(\bbA,\bbS)$ with the sequence 
$(a_{0},s_{0})(a_1,s_1)(a_2,s_2)\ldots$ of its \emph{basic positions}.
\end{conv}

Some basic concepts concerning modal automata are introduced in the following 
definition.

\begin{defi}
Fix a modal $\Prop$-automaton $\bbA = (A,\Th,\Om)$.

The \emph{size} $\sz{\bbA}$ of $\bbA$ is defined as the cardinality $\sz{A}$ 
of its carrier set, while its \emph{weight} $\wt{\bbA} \isdef \max \{ \sz{\Th(a)} 
\mid a \in A \}$ is given as the maximum size of the one-step formulas in the
range of $\Th$.

Given a state $a$ of $\bbA$, we write $\eta_{a} = \mu$ if $\Om(a)$ is odd, and
$\eta_{a} = \nu$ if $\Om(a)$ is even; we call $\eta_{a}$ the 
\emph{(fixpoint) type} of $a$ and say that $a$ is an $\eta_{a}$-state.

The \emph{occurrence graph} of $\bbA$ is the directed graph $(G,E_{\bbA})$, 
where $aE_{\bbA}b$ if $a$ occurs in $\Th(b)$.
We let $\act_{\bbA}$ denote the transitive closure of $E_{\bbA}$ and say that
$a$ is \emph{active} in $b$ if $a \act_{\bbA}b$.
We write $a \biact_{\bbA} b$ if $a \act_{\bbA} b$ and $b \act_{\bbA} a$.
A \emph{cluster} of $\bbA$ is a cell of the equivalence relation generated by 
$\biact_{\bbA}$ (i.e., the smallest equivalence relation on $A$ containing 
$\biact_{\bbA}$); a cluster $C$ is \emph{degenerate} if it is of the form
$C = \{ a \}$ with $a \not\biact_{\bbA} a$. 
The unique cluster to which a state $a \in A$ belongs is denoted as $C_{a}$.
We write $a \blw_{\bbA} b$ if $\Om(a) < \Om(b)$, and $a \blweq_{\bbA} b$ if 
$\Om(a) \leq \Om(b)$. 
\end{defi}

\subsection{Disjunctive modal automata}

Many of our proofs involve the non-deterministic version of modal automata
that we call \emph{disjunctive}.
The transition map of these automata makes use of the so-called \emph{cover 
modality} $\nb$, which was independently introduced in coalgebraic 
logic by Moss~\cite{moss:coal99} and in automata theory by Janin \& 
Walukiewicz~\cite{jani:auto95} (where the authors used a different notation).
It is a slightly non-standard connective that takes a finite set of formulas as 
its argument.

\begin{defi}
\label{d:nb}
Given a finite set $\Phi$, we let $\nb\Phi$ abbreviate the formula
\[
\nb\Phi \isdef \bw\dia\Phi \land \Box\bv\Phi,
\]
where $\dia\Phi$ denotes the set $\{ \dia\phi \mid \phi \in \Phi\}$.
\end{defi}

\begin{rem}
\label{r:nbeq}
In words, $\nb\Phi$ holds at $s$ iff $\Phi$ and $\si_{R}(s)$ `cover' one 
another, in the sense that every successor of $s$ satisfies some formula in 
$\Phi$ and every formula in $\Phi$ holds in some successor of $s$.
From this observation it is easy to derive that, conversely, the standard modal 
operators can be expressed in terms of the cover modality:
\[\begin{array}{lll}
   \dia\phi &\equiv& \nb\{ \phi,\top\}
\\ \Box\phi &\equiv& \nb\{\phi\} \lor \nb\nada,
\end{array}\]
where we note that $\nb\nada$ holds at a point $s$ iff $s$ is a `blind' world,
that is, $R[s] = \nada$.
\end{rem}

\begin{defi}
\label{d:DMLone}
Let $\Prop$ be a given set of proposition letters and $A$ any finite set. 
We define a \emph{literal} over $\Prop$ to be a formula of the form $p$ or 
$\neg p$ with $p \in \Prop$, and define the language $\LitC(\Prop)$ of 
conjunctions of literals to be generated by $\pi$ in the grammar:
\[
\pi \isbnf p \divbnf \neg p \divbnf \top \divbnf \pi \wedge \pi
\]
where $p \in \Prop$. 
We now define the set $\DMLone(\Prop,A)$ of \emph{disjunctive} formulas in 
$\MLone(\Prop,A)$ as follows:
\[
\alpha \isbnf \pi \ybullet \nb B \divbnf \bot \divbnf \alpha \vee \alpha,
\]
where $\pi \in \LitC(\Prop)$ and $B \sse A$.
\end{defi}

\begin{defi}
\label{d:daut}
A modal $\Prop$-automaton $\bbA = (A,\Th,\Om)$ is \emph{disjunctive} if $\Th(a) 
\in \DMLone(\Prop,A)$ for all $a \in A$.
\end{defi}

The non-deterministic nature of disjunctive automata is exemplified by the
following fact, which can easily be verified.

\begin{fact}
\label{f:daut}
Let $\bbA$ be a disjunctive automaton and $\bbS$ an $\om$-unravelled tree.
Then without loss of generality we may assume $\eloi$'s positional winning 
strategy to be such that each marking $m$ picked by the strategy at a position
$(a,s)$ satisfies $\sz{m(t)} = 1$, for all $t \in \si_{R}(s)$.
\end{fact}

In many branches of automata theory a crucial role is played by a 
\emph{simulation theorem}, stating that automata of certain given type can be 
transformed into, or `simulated' by, an equivalent automaton of which the 
transition structure is of a conceptually simpler kind.
In the case of modal automata, such a theorem can be proved along the lines of
Janin \& Walukiewicz' characterization of the $\mu$-calculus in terms of the 
non-deterministic $\mu$-automata~\cite{jani:auto95}.
For a detailed proof of the result below we refer to Venema~\cite{vene:lect12}.

\begin{fact}[Simulation Theorem]
\label{f:sim}
There is an effective procedure transforming an initialized modal automaton
$\bbA\init{a}$ into an equivalent initialized disjunctive modal automaton 
$\bbD\init{d}$.
The size of $\bbD$ is exponential in the size of $\bbA$, and its weight is 
at most exponential in the product of the size and the weight of $\bbA$.
\end{fact}

\subsection{Formulas and automata}

The automata-theoretic approach towards the modal $\mu$-calculus hinges on
the existence of truth-preserving translations between formulas and automata,
which testify that automata and formulas have the same expressive power.
In the direction from formulas to automata we only need the result as such.

\begin{fact}
\label{f:fma-aut}
(i) There is an effective procedure providing for any formula $\xi \in 
\muML(\Prop)$ an equivalent initialized modal automaton 
$\bbA_{\xi}\init{a_{\xi}}$.
The size of $\bbA_{\xi}$ is linear and its weight is at most exponential in 
the size of $\xi$.

\noindent
(ii) There is an effective procedure providing for any formula $\xi \in 
\muML(\Prop)$ an equivalent initialized disjunctive modal automaton 
$\bbD_{\xi}\init{d_{\xi}}$.
The size of $\bbD_{\xi}$ is exponential and its weight is at most 
doubly exponential in the size of $\xi$.
\end{fact}

The translations mentioned in Fact~\ref{f:fma-aut} originate with Janin \& 
Walukiewicz~\cite{jani:expr96} for part (ii)  and with Wilke~\cite{wilk:alte01} 
for part (i).
Concerning the size matters, note that we get an exponential (and not 
polynomial) weight bound in Fact~\ref{f:fma-aut}(1) because our modal automata 
are guarded (that is, all occurrences of states/variables in one-step formulas 
must occur in the scope of a modality; see Bruse et alii for 
details~\cite{brus:guar15}).
The size bounds in Fact~\ref{f:fma-aut}(ii) are obtained simply by combining 
part (i) with Fact~\ref{f:sim}.

In the opposite direction we will need an actual map transforming an 
initialized modal automaton into an equivalent $\mu$-calculus formula.
For our definition of such a map, which is a variation of the one found 
in~\cite{grae:auto02}, we need some preparations.
For a proper inductive formulation of this definition it is convenient to 
extend the class of modal automata, allowing states of the automaton to appear
in the scope of a modality in a one-step formula.

\begin{defi} 
\label{d:gma}
A \emph{generalized modal automaton} over $\Prop$ is a triple $\bbA = (A,\Th,
\Om)$ such that $A$ is a finite set of states, $\Om: A \to \om$ is a priority 
map, and $\Th: A \to \MLone(\Prop, A \cup \Prop)$ maps states of $A$ to 
\emph{generalized one-step formulas}.
\end{defi}

To mark the difference with standard modal automata, the formula $\dia a \land 
\Box(p \lor b)$ is a generalized one-step formula but not a proper one-step 
formula.
Whenever possible, we will apply concepts that have been defined for modal 
automata to these generalized structures without explicit notification.
For the operational semantics of generalized modal automata we may extend the 
notion of a one-step model in the obvious way.
Readers who are interested in the details may 
consult~\cite{enqv:comp16a}.

\begin{defi} 
\label{d:lma}
A (generalized) modal automaton $\bbA = (A,\Th,\Om)$ is called  \emph{linear}
if the relation $\blw_{\bbA}$ is a linear order (i.e., the priority map $\Om$
is injective), and satisfies $\Om(a) > \Om(b)$ in case $b$ is active in $a$ but 
not vice versa.
A \emph{linearization} of $\bbA$ is a linear automaton $\bbA' = (A,\Th,\Om')$ 
such that (1) for all $a \in A$, $\Om'(a)$ has the same parity as $\Om(a)$, and 
(2) for all $a,b \in A$ that belong to the same cluster we have $\Om'(a) < 
\Om'(b)$ iff $\Om(a) < \Om(b)$. 
\end{defi}

The following proposition is easy to verify.

\begin{prop}
\label{p:ma-lin}
Let $\bbA$ be a modal automaton. 

(1) $\bbA$ has some linearization.

(2) If $\bbA'$ is a linearization of $\bbA$, then $\bbA\init{a} \equiv 
\bbA'\init{a}$.
\end{prop}

\begin{defi} 
\label{d:ma-tr}
We define a map 
\[
\tr_{\bbA}: A \to \muML(\Prop)
\]
for any linear generalized $\Prop$-automaton $\bbA = (A,\Th, \Om)$.
These maps are defined by induction on the size of the automaton $\bbA$.

In case $\sz{\bbA} = 1$, we set
\[
\tr_{\bbA}(a) \isdef \eta_{a} a. \Th(a),
\]
where $a$ is the unique state of $\bbA$.

In case $\sz{\bbA} > 1$, by linearity there is a unique state $m$ reaching 
the maximal priority of $\bbA$, that is, with $\Om(m) = 
\max(\Ran(\Om))$.
Let $\bbA^{-} = (A^{-},\Th^{-},\Om^{-})$ be the $\Prop\cup \{ m \}$-automaton
given by $A^{-} \isdef A \setminus \{ m \}$, while $\Th^{-}$ and $\Om^{-}$ are 
defined as the restrictions of, respectively, $\Th$ and $\Om$ to $A^{-}$.
Since $\sz{\bbA^{-}} < \sz{\bbA}$, inductively\footnote{%
   Observe that since $m$ is a proposition letter and not a variable in 
   $\bbA^{-}$, the latter structure need not be a modal automaton, even if 
   $\bbA$ is. 
   It is for this reason that we introduced the notion of a \emph{generalized} 
   modal automaton.
   }
we may assume a map 
$\tr_{\bbA^{-}}: A \to \muML(X \cup \{ m \})$.

Now we first define
\[
\tr_{\bbA}(m) \isdef \eta_{m} m. \Th(m) [\tr_{\bbA^{-}}(a)/a \mid a \in A^{-}],
\]
and then set
\[
\tr_{\bbA}(a) \isdef \tr_{\bbA^{-}}(a)[\tr_{\bbA}(m)/m]
\]
for the states $a \neq m$.
\end{defi}

The following proposition states that this map is truth-preserving, and of
exponential size.
The proof of the first statement is rather standard; details can be found 
in~\cite{grae:auto02,vene:lect12}; the size bound can be established via a
straightforward induction on the size of the automaton.

\begin{fact}
\label{f:ma-lin-tr}
Let $\bbA\init{a}$ be an initialized linear modal automaton.
Then $\tr_{\bbA}(a)$ is equivalent to $\bbA\init{a}$ and its size is at most
exponential in the sum of the size and the weight of $\bbA$.
If $\bbA$ is positive in $p\in\Prop$ then so is $\tr_{\bbA}(a)$.
\end{fact}

Based on Proposition~\ref{p:ma-lin} and Fact~\ref{f:ma-lin-tr} we may define
a truth-preserving translation $\tr_{\bbA}: A \to \muML$ for every modal 
automaton $\bbA = (A,\Th,\Om)$ by picking a linearization $\bbA^{l}$ of $\bbA$
and defining $\tr_{\bbA}(a) \isdef \tr_{\bbA^{l}}(a)$ for all $a \in A$.
Clearly the exact shape of $\tr_{\bbA}(a)$ will depend on the choice of the 
linearization.

\subsection{Bipartite automata}

As mentioned in our introduction, the results in this paper will be based on 
constructions that transform a given automaton into one of a particular shape.
In this final subsection on modal automata we look at the target automata of 
these constructions in some more detail.

\begin{defi}
\label{d:bipaut} 
A modal automaton $\bbA$ is called a \emph{bipartite} automaton if its carrier
$A$ can be partitioned into an \emph{initial} part $A_{0}$ and a \emph{final} 
part $A_{1}$, in such a way that we have ${\act_{\bbA}} \cap 
(A_{0}\times A_{1}) = \nada$.
It will sometimes be convenient to represent a bipartite automaton as a
quadruple $\bbA = (A_{0},A_{1},\Th,\Om)$.

Given such an automaton and a state $a$ of $\bbA$, we say that $\bbA\init{a}$ 
is an \emph{initialized} bipartite automaton if $a$ belongs to the initial 
part of $\bbA$; and given a class $C$ of bipartite automata, we let 
$\mathit{IC}$ denote the class of corresponding initialized bipartite automata.
\end{defi}

The intuition underlying this definition is that, in any acceptance game related
to a bipartite automaton $\bbA = (A_{0},A_{1},\Th,\Om)$ there will be two kinds
of matches: either $\bbA$ stays in its initial part throughout the match, or at
some stage it moves to its final part, where it remains throughout the remainder 
of the match.
The bipartite automata that we will meet in this paper will be such that its 
initial and final part behave differently with respect to the designated 
propositional variable $p$.

Recall that each of our results concerns a certain syntactic fragment $\Frag(P)$
of the modal $\mu$-calculus, consisting of formulas where we have imposed 
restrictions on the occurrence of the proposition letters belonging to a 
certain set $P \sse \Propvar$.
(Note that while these restrictions generally concern the free variables of the
formula, not the bound ones, there is no relation between the fragment 
$\Frag(P)$ and the set $\muML(\Prop)$ of formulas of which the free 
variables belong to $\Prop$).
In each case considered, part of our proof will consist in showing that given 
a bipartite automaton $\bbA$ of a certain kind, one can find a a translation 
$\tr: A \to \muML$, such that for all states $a$ in the initial part of $\bbA$, 
the formula $\tr(a)$ belongs to the fragment $\Frag(P)$.
Since these proofs are all quite similar, we have extracted their pattern
in the form of the following definition and proposition, so that in each 
specific case we may confine our attention to a verification that the fragment 
and the automata at stake meet the required conditions.

\begin{defi}
Suppose that we have defined, for each finite set $P$ of proposition letters, 
a fragment $\Frag(P) \sse \muML$ of $\mu$-calculus formulas.
In the sequel we shall consider the following properties that are applicable to 
such a family $\Frag = \{ \Frag(P) \mid P \sse_{\om} \Propvar\}$:
\\
(EP) \emph{extension property}: $\phi \in \Frag(P \cup \{ q \})$ whenever 
   $\phi \in \Frag(P)$ and $q \not\in \FV{\phi}$;
\\   
(SP1) \emph{first substitution property}:
   $\phi[\psi/x] \in \Frag(P)$ whenever $\phi \in \Frag(P \cup \{x\})$ and 
   $\psi \in \Frag(P)$;
\\
(SP2) \emph{second substitution property}:
   $\phi[\psi/x] \in \Frag(P)$ whenever $\phi \in \Frag(P)$, $x \not\in P$
   and $\FV{\psi} \cap P = \nada$;
\\
(CA$_{\eta}$) \emph{$\eta$-fixpoint closure property}:
   $\eta x.\phi \in \Frag(P)$ whenever $\phi \in \Frag(P \cup \{x\})$.
\end{defi}

\begin{prop}
\label{p:tr}
Let $\Frag = \{ \Frag(Q) \mid Q \sse \Propvar \}$ be a family of fragments of 
$\muML$, and let $P,\Prop$ be two finite sets of proposition letters.
Let $\bbA = (A,B,\Th,\Om)$ be a bipartite $\Prop$-automaton, and let $\bbB$ 
denote the automaton $(B,\Th_{\rst{B}},\Om_{\rst{B}})$.
Assume that

(i) $F$ has the properties (EP), (SP1), (SP2) and, for all $a \in A$, 
   (CA$_{\eta_{a}}$);

(ii) $\Th(a) \in \Frag(P \cup A)$ for all $a \in A$;

(iii) there is a translation $\tr: B \to \muML(\Prop)$ such that
   $\bbB\init{b} \equiv \tr(b) \in \Frag(P)$, for all $b \in B$;

\noindent
Then there is a translation $\tr: A \to \muML(\Prop)$ such that 
$\tr_{\bbA}(a) \equiv \bbA\init{a}$ and $\tr_{\bbA}(a) \in \Frag(P)$ for all
$a \in A$.
\end{prop}

\begin{proof}
By Proposition~\ref{p:ma-lin} we may assume without loss of generality that 
$\bbA$ is linear itself and that $\Om(a) > \Om(b)$ for all $a \in A$ and $b 
\in B$.
We will also assume that the translation $\tr$ of (iii) is identical to 
the map $\tr_{\bbB}$ obtained by applying Definition~\ref{d:ma-tr} to the 
automaton $\bbB$.
(The more general case can be proved by a straightforward modification of the 
proof given below.)

Enumerate $A = \{ a_{1},\ldots,a_{n}\}$ in such a way that $\Om(a_{i}) < 
\Om(a_{j})$ iff $i < j$.
For $0 \leq k \leq n$ we define $A_{k} \isdef \{ a_{i} \mid 0 < i \leq k \}$, 
and we let $\bbA_{k}$ be the linear generalized automaton $(B \cup A_{k},
\Th_{\rst{B \cup A_{k}}},\Om_{\rst{B \cup A_{k}}})$.
It is straightforward to verify that $\bbA_{0} = \bbB$, $\bbA_{n} = \bbA$, and 
that $\bbA_{i} = (\bbA_{i+1})^{-}$ for all $i < n$.
We abbreviate $\tr_{k} \isdef \tr_{\bbA^{k}}$.

Our task is then to show that $\tr_{n}(a) \in \Frag(P)$ for all $a \in A$, and
our approach will be to prove, by induction on $k$, that
\begin{equation}
\label{eq:tr1}
\tr_{k}(c) \in \Frag(P \cup \{ a_{i} \mid k < i \leq n\}), 
\text{ for all } c \in A \cup B \text{ and } k \leq n.
\end{equation}

In the base step of the induction, where $k=0$, we are dealing with the 
automaton $\bbA_{0} = \bbB$, corresponding to the final part of $\bbA$, and we 
only have to worry about states $b \in \bbB$.
It is easy to see that
\[
\tr_{0}(b) = \tr_{\bbB}(b) \text{ for all } b \in B,
\]
and so by assumption (iii) and the assumption (i) that $F$ has the extension 
property, we obtain that $\tr_{0}(b) \in \Frag(P) \sse \Frag(P \cup A)$.
This suffices to prove \eqref{eq:tr1} in the base case.

In the inductive case, where $k > 0$, we consider the automaton $\bbA_{k}$, with
maximum-priority state $a_{k} \in A$.
Observe that by definition we have
\begin{equation}
\label{eq:tr2}
\tr_{k}(a_{k}) \isdef \eta_{a_{k}}a_{k}. 
   \Th(a_{k})[\tr_{k-1}(a_{i})/a_{i}\mid i < k, \tr_{k-1}(b)/b\mid b \in B]
\end{equation}
and, for $i<k$,
\begin{equation}
\label{eq:tr3}
\tr_{k}(a_{i}) \isdef \tr_{k-1}(a_{i})[\tr_{k}(a_{k})/a_{k}].
\end{equation}
Observe that we may conclude from \eqref{eq:tr2} that 
\begin{equation}
\label{eq:tr2a}
\tr_{k}(a_{k}) \isdef \eta_{a_{k}}a_{k}. 
   \Th(a_{k})[\tr_{k-1}(a_{i})/a_{i}\mid i < k][\tr_{k-1}(b)/b\mid b \in B],
\end{equation}
since no $b \in B$ has a free occurrence in any formula $\tr_{k-1}(a_{i})$.

We now turn to the proof of \eqref{eq:tr1}.
Starting with the states in $B$, it is easy to see that
\[
\tr_{k}(b) = \tr_{\bbB}(b) \text{ for all } b \in B 
\text{ and } k \leq n,
\]
and so we find $\tr_{k}(b) \in \Frag(P \cup \{ a_{i} \mid k < i \leq 
n\} )$
by assumption (iii) and the assumption (i) that $F$ has the extension property.

We now consider the formula $\tr_{k}(a_{k})$.
By the inductive hypothesis we have 
$\tr_{k-1}(a_{i}) \in \Frag(P \cup \{a_{i} \mid k \leq i \leq n\})$,
while by assumption (iii) we have $\Th(a_{k}) \in 
\Frag(P \cup A) = 
\Frag((P \cup \{ a_{i} \mid k \leq i \leq n \}) \cup \{ a_{i} \mid i < k \})$.
But then by the first substitution property of $F$ we find that 
\[
\Th(a_{k})[\tr_{k-1}(a_{i})/a_{i}\mid i < k]
\in \Frag(P \cup \{ a_{i} \mid k \leq i \leq n\} ),
\]
and by the fact that $\Frag$ is closed under the application of 
$\eta_{a_{k}}$-fixpoint operators, we obtain
\begin{equation}
\label{eq:tr4}
\eta_{a_{k}}a_{k}. \Th(a_{k})[\tr_{k-1}(a_{i})/a_{i}\mid i < k]
\in \Frag(P \cup \{ a_{i} \mid k < i \leq n\} ).
\end{equation}
Now observe that for all $b \in B$ we have $b \not\in P$ and 
$\tr_{k-1}(b) \in \Frag(P \cup \{ a_{i} \mid k < i \leq n\} )$, so that by 
successive applications of the second substitution property we obtain from
\eqref{eq:tr4} and \eqref{eq:tr2a} that
\begin{equation}
\label{eq:tr5}
\tr_{k}(a_{k}) \in \Frag(P \cup \{ a_{i} \mid k < i \leq n\} ), 
\end{equation}
as required.

Finally we consider the formula $\tr_{k}(a_{j})$ for some fixed but arbitrary
$j < k$.
From the inductive observation that $\tr_{k-1}(a_{j})$
belongs to $\Frag(P \cup \{a_{i} \mid k \leq i \leq n \})
= \Frag((P \cup \{a_{i} \mid k < i \leq n \}) \cup \{ a_{k}\} )$, we may then 
conclude by \eqref{eq:tr5}, \eqref{eq:tr3} and the first substitution property 
that 
\[
\tr_{k}(a_{j}) \in \Frag(P \cup \{ a_{i} \mid k < i \leq n\}),
\]
and so we are done.
\end{proof}

\section{First steps}
\label{sec:fs}

In this section we provide both an automata-theoretic construction and a first
preservation result that that will be of use throughout the paper.
We start with the first, introducing a simple operation on automata that will 
feature throughout the remainder of the paper.

\begin{conv}
As mentioned in the introduction, in this paper we shall focus on the 
contribution of one specific proposition letter in the semantics of formulas 
and automata.
It will be convenient to fix this letter from now on, and reserve the name 
`$p$' for it.
\end{conv}

\begin{defi}
\label{d:Abot}
Given a modal automaton $\bbA = (A,\Th,\Om)$, we define $\bbA^{\bot} = (A^{\bot},
\Th^{\bot},\Om^{\bot})$ to be the automaton where $A^{\bot} \isdef 
\{ a^{\bot} \mid a \in A \}$, $\Th^{\bot}$ is given by putting
\[
\Th^{\bot}(a) \isdef \Th(a)[\bot/p],
\]
and $\Om^{\bot}$ is simply defined by $\Om^{\bot}(a^{\bot}) \isdef \Om (a)$.
\end{defi}

\begin{prop}
\label{p:fs2}
Let $\bbA$ be a modal automaton which is positive in $p$.
Then for any state $a$ in $\bbA$ and any pointed Kripke model $(\bbS,s)$ we have
\[
\bbS,s \sat \bbA^{\bot}\init{a} \text{ iff } 
\bbS[p \mapsto \nada],s \sat \bbA\init{a}.
\]
\end{prop}

\begin{proof}
We will show that the two respective acceptance games, $\AG^{\bot} \isdef
\AG(\bbA^{\bot},\bbS)$ and $\AG_{\nada} \isdef \AG(\bbA,\bbS[p\mapsto\nada])$ 
are in fact, \emph{identical}.
The key observation here is that for any position $(a,t) \in A \times S$, the 
set of moves available to $\eloi$ in both games is the same.
To see this, it suffices to prove that, for any point $t \in S$ and any one-step
formula $\al \in \MLone$ which is positive in $p$, we have
\begin{equation*}
\si_{V}(t),\si_{R}(t),m \satone \al[\bot/p] \text{ iff } 
\si_{V}(t) \setminus \{p\}, \si_{R}(t), m) \satone \al,
\end{equation*}
and the proof of this statement proceeds by a straightforward induction on $\al$.
\end{proof}

We now turn to our first characterization result, which concerns the notion of
\emph{monotonicity}.
D'Agostino and Hollenberg~\cite{dago:logi00} already proved a Lyndon theorem 
for the modal $\mu$-calculus, stating that monotonicity with respect to a 
propositional variable $p$ is captured by the formulas that are syntactically 
\emph{positive} in $p$.
Here we strengthen their result, providing both an explicit translation and
a decidability result.
Our proof proceeds via some auxiliary lemmas on automata that we shall need 
further on.

\begin{defi}
\label{d:mon}
A formula $\xi \in \muML(\Prop)$ is \emph{monotone} in $p \in \Prop$ if for
all pointed Kripke structures $(\bbS,s)$ and for all pairs of sets $U,U'\sse S$
such that $U \sse U'$, it holds that $\bbS[p \mapsto U],s \sat \xi$ implies
$\bbS[p \mapsto U'],s \sat \xi$.

Syntactically, the set $\muML^{M}_{p}$ of formulas that are \emph{positive}
in $p$, is defined by the following grammar:
\[
\phi \isbnf p 
   \divbnf q \divbnf  \neg q
   \divbnf \phi\lor\phi \divbnf \phi\land\phi 
   \divbnf \dia\phi \divbnf \Box\phi 
   \divbnf \mu x.\phi \divbnf \nu x.\phi,
\]
where $q \neq p$ and $x$ are propositional variables, and the formation of the
formulas $\mu x.\phi$ and $\nu x.\phi$ is subject to the constraint that the
formula $\phi$ is \emph{positive} in $x$.
\end{defi}

Definition~\ref{d:mon} gives an explicit grammar for generating the formulas
that are positive in $p$.
It is easy to see that this definition coincides with the one we gave earlier,
viz., that a formula is positive in $p$ iff all occurrences of $p$ are in the 
scope of an even number of negations.

The theorem below can be seen as a strong version of D'Agostino \& Hollenberg's
characterization result.

\begin{thm}
\label{t:mon}
There is an effective translation which maps a given $\muML$-formula $\xi$ to
a formula $\xi^{M} \in \muML^{M}_{p}$ such that 
\begin{equation}
\label{eq:t-m}
\xi \mbox{ is monotone in } p \ouriff 
\xi \equiv \xi^{M},
\end{equation}
and it is decidable in elementary time whether a given formula $\xi$ is monotone 
in $p$.
\end{thm}

It is routine to prove that all formulas in $\muML^{M}_{p}$ are monotone in $p$,
so we focus on the hard part of Theorem~\ref{t:mon}, for which we shall involve 
automata.

\begin{defi}
\label{d:mon-aut1}
A modal automaton $\bbA = (A,\Th,\Om)$ is \emph{positive in $p$} if for all $a
\in A$ the one-step formula $\Th(a)$ is positive in $p$.
\end{defi}

It is easy to see that for any linear modal automaton that is positive in $p$,
the translation map given in Definition~\ref{d:ma-tr} produces formulas that
are also positive in $p$. 
From this the following is immediate.

\begin{prop}
\label{p:mon0}
Let $\bbA = (A,\Th,\Om)$ be a modal automaton which is positive in $p$.
Then there is a truth-preserving translation $\tr_{\bbA}: A \to\muML^{M}_{p}$.
\end{prop}

We now turn to the key lemma underlying the proof of Theorem~\ref{t:mon}, for
which we need the following definition.

\begin{defi}
\label{d:mon-aut2}
Let $(\cdot)^{M}: \LitC(\Prop) \to \LitC(\Prop)$ be the translation which 
replaces every occurrence of the literal $\neg p$ with $\top$,
and let $(\cdot)^{M}: \DMLone(\Prop,A) \to \MLone(\Prop,A)$ be the one-step 
translation given by the following inductive definition:
\[\begin{array}{lll}
(\pi\ybullet\nb B)^{M} & \isdef & \pi^{M}\ybullet\nb B
\\ \bot^{M} & \isdef & \bot
\\ (\al \lor \be)^{M} & \isdef & \al^{M} \lor \be^{M}.
\end{array}\]

Let $\bbA = (A, \Th, \Om)$ be a disjunctive modal automaton.
We define the automaton $\bbA^{M}$ as the structure $(A,\Th^{M},\Om)$
where the map $\Th^{M}$ is given by putting
\[
\Th^{M}(a) \isdef \Th(a)^{M}
\]
for every state $a \in A$.
\end{defi}

The following proposition is fairly obvious; we list it for future 
reference.

\begin{prop}
\label{p:mon1}
Let $\bbA$ be a disjunctive automaton.
Then its transformation $\bbA^{M}$ is disjunctive as well, and positive in $p$.
\end{prop}

\begin{prop}
\label{p:mon-aut3}
Let $\bbA\init{\ai}$ be an initialized disjunctive modal automaton.
If $\bbA\init{\ai}$ is monotone in $p$, then $\bbA\init{\ai} \equiv 
\bbA^{M}\init{\ai}$.
\end{prop}

\begin{proof}
Given the nature of the translation $(\cdot)^{M}$ at the level of conjunctions 
of literals, it is easy to see that $\bbA\init{\ai}$ implies $\bbA^{M}\init{\ai}$,
and so we focus on the opposite implication.
It suffices to take an arbitrary $\om$-unravelled tree model $(\bbS,r)$ and show
that 
\begin{equation}
\label{eq:mon1}
\bbS,r \sat \bbA^{M}\init{\ai} \text{ only if } \bbS,r \sat \bbA\init{\ai}.
\end{equation}
Assume that $\bbS,r \sat \bbA^{M}\init{\ai}$.
Our aim is to find a subset $U \sse S$ such that 
\begin{equation}
\label{eq:mon2}
\bbS[p\rst{U}],r \sat \bbA\init{\ai},
\end{equation}
from which it will immediately follow by monotonicity that $\bbS,r \sat 
\bbA\init{\ai}$.

By Fact~\ref{f:daut} $\eloi$ has a positional winning strategy $f$ in 
$\AG(\bbA^{M},\bbS)$ such that each marking $m$ picked by the strategy at a 
position $(a,s)$ satisfies $\sz{m(t)} = 1$, for all $t \in \si_{R}(s)$.
From this it easily follows that for all $t \in S$ there is exactly one state
$a_{t} \in A$ such that the position $(a_{t},t)$ may occur in an $f$-guided
match of $\AG(\bbA^{M},\bbS)$ starting at $(\ai,r)$.
It easily follows that the pair $(a_{t},t)$ is a winning position for $\eloi$ in 
$\AG(\bbA^{M},\bbS)$; in particular, with $m_{t}: \si_{R}(t) \to \funP A$ being
the marking picked by $f$ at this position, by the legitimacy of this move we
have that
\[
\si_{V}(t),\si_{R}(t),m_{t} \satone \Th^{M}(a_{t}).
\]
In order to determine whether this $t$ should be a member of $U$ or not, we
check whether the one-step formula $\Th(a_{t})$ holds at the one-step modal 
$(\si_{V}(t),\si_{R}(t),m_{t})$.
If so, we are happy with $\si_{V}(t)$ as it is and put $t \in U$, but if not,
we will want to make $p$ false at $t$, claiming that 
\begin{equation}
\label{eq:mon11}
\si_{V}(t)\setminus \{ p \},\si_{R}(t),m_{t} \satone \Th(a_{t}).
\end{equation}
To see why this is the case, observe that the only reason why we can have 
$\si_{V}(t),\si_{R}(t),m_{t} \satone \Th^{M}(a_{t})$ but 
$\si_{V}(t),\si_{R}(t),m_{t} \not\satone \Th(a_{t})$ is 
that $\Th(a_{k})$ has a disjunct $\pi\ybullet \nb B$ such that $\neg p$ is a 
conjunct of $\pi$ and $p \in \si_{V}(t)$.
But then it is immediate that $\si_{V}(t)\setminus \{ p \},\si_{R}(t),m_{t} 
\satone \pi\ybullet\nb B$, which implies~\eqref{eq:mon11}.

These observations reveal that if we define
\[
U \isdef \{ t \in S \mid (a_{t},t) \in \Win_{\eloi}(\AG(\bbA^{M},\bbS))
 \text{ and } \si_{V}(t),\si_{R}(t),m_{t} \satone \Th(a_{t}) \},
\]
we obtain for the valuation $V[p \rst{U}]$ that
\[
\si_{V[p \rst{U}]}(t),\si_{R}(t),m_{t} \satone \Th(a_{t}),
\]
whenever $(a_{t},t)$ is a winning position for $\eloi$ in $\AG(\bbA,\bbS)$.
From this it is easy to derive that $f$ itself is a winning strategy for
$\eloi$ in the acceptance game $\AG(\bbA,\bbS[p \rst{U}])@(\ai,r)$, as required
to prove \eqref{eq:mon2}.
\end{proof}

We now have all material needed to prove Theorem~\ref{t:mon}.

\begin{proof}[{\rm\bf Proof of Theorem~\ref{t:mon}}]
Let $\xi$ be an arbitrary modal $\mu$-formula; and let $\bbD_{\xi}\init{d_{\xi}}$
be an initialized disjunctive automaton that is equivalent to $\xi$.
Such a structure exists by Fact~\ref{f:fma-aut}, and clearly $\xi$ is monotone
in $p$ iff $\bbD_{\xi}\init{d_{\xi}}$ is so.
It then follows by Proposition~\ref{p:mon-aut3} that $\bbD_{\xi}\init{d_{\xi}}$
(and hence $\xi$) is monotone in $p$ iff it is equivalent to the 
initialized disjunctive automaton $\bbD_{\xi}^{M}\init{d_{\xi}}$.
Now define
\[
\xi^{M} \isdef \tr_{\bbD_{\xi}^{M}}(d_{\xi}).
\]
It is easy to verify that $\xi^{M} \in \muML^{M}_{p}$, while we have 
$\bbD_{\xi}^{M}\init{d_{\xi}} \equiv \xi^{M}$ by Fact~\ref{f:fma-aut}.
Putting these observations together we obtain that, indeed, $\xi$ is monotone
in $p$ iff $\xi \equiv \xi^{M}$, where $\xi^{M}$ is effectively obtained from 
$\xi$. 
This establishes the first part of the theorem.

For the statement concerning decidability, it suffices to observe that all 
constructions that are involved in the definition of the map $(\cdot)^{M}:
\muML(\Prop) \to \muML^{M}_{p}$ have uniformly elementary size bounds, and 
that the problem, whether two $\mu$-calculus formulas are equivalent or not,
can be decided in exponential time.
From this the decidability claim is immediate.
\end{proof}

\section{Finite width property}
\label{sec:fw}

The first new property that we consider is that of the \emph{finite width 
property}.

\begin{defi}
\label{d:w-prop}
A formula $\xi \in \muML(\Prop)$ has the \emph{finite width property} for 
$p \in \Prop$ if $\xi$ is monotone in $p$, and, for every tree model $(\bbS,s)$,
\[
\bbS,s \sat \xi \ouriff
\bbS[p\rst{U}],s \sat \xi, \mbox{ for some finitely branching subtree }
  U \sse S,
\]
where a subset $U\sse S$ is a \emph{finitely branching subtree} if $U$ is 
downward closed and the set $R(u) \cap U$ is finite for every $u \in U$.
\end{defi}

We will associate the following syntactic fragment of the modal $\mu$-calculus
with this property.

\begin{defi}
\label{d:w-frag}
Given a set $P \sse \Prop$, we define the fragment $\muML^{W}_{P}$ by the 
following grammar:
\[
\phi \isbnf p \divbnf \psi 
   \divbnf \phi\lor\phi \divbnf \phi\land\phi 
   \divbnf \dia \phi 
   \divbnf \mu x. \phi' \divbnf \nu x.\phi',
\]
where $p \in P$, $\psi\in \muML(\Prop\setminus P)$ is a $P$-free formula
and $\phi' \in \muML^{W}_{P \cup \{ x \}}$.
In case $P$ is a singleton, say, $P = \{ p \}$, we will write 
$\muML^{W}_{p}$ rather than $\muML^{W}_{\{p\}}$.
\end{defi}

In words, the fragment $\muML^{W}_{p}$ consists of those formulas that are
positive in $p$ and do not admit any occurrence of a $p$-active subformula 
in the scope of a box modality.
All formulas from Example~\ref{ex:fmas} belong to $\muML^{W}_{p}$, except 
$\phi_{4}$.

The following theorem states that modulo equivalence, $\muML^{W}_{p}$ is the
syntactic fragment of the modal $\mu$-calculus that captures the finite width
property, and that it is decidable whether a given $\mu$-formula has this 
property.

\begin{thm}
\label{t:fw}
There is an effective translation which maps a given $\muML$-formula $\xi$ to
a formula $\xi^{W} \in \muML^{W}_{p}$ such that 
\begin{equation}
\label{eq:t-w}
\xi \mbox{ has the finite width property for } p \ouriff 
\xi \equiv \xi^{W},
\end{equation}
and it is decidable in elementary time whether a given formula $\xi$ has
the finite width property for $p$.
\end{thm}

First we prove the easy part of the theorem, stating that formulas in the 
fragment $\muML^{W}_{P}$ indeed have the required semantic property.

\begin{prop} 
\label{p:w1}
Every formula $\xi \in \muML^{W}_{p}$ has the finite width property with respect 
to $p$.
\end{prop}

\begin{proof}
Let $\xi$ be a formula in $\muML_W(p)$, then $\xi$ is obviously positive,
and hence, monotone in $p$.
Fix a tree model $\bbS$ with root $r \in S$.
We have to prove 
\begin{equation} 
\label{eq:widtheasy}
\bbS,r \sat \xi \ouriff
\bbS[p\rst{U}],r \sat \xi, \mbox{ for some finitely branching subtree }
  U \sse S.
\end{equation}

The direction from right to left follows from the fact that $\xi$ is monotone
in $p$.
For the opposite direction, suppose that $\bbS,r \sat \xi$.
We need to find a finitely branching subtree $U$ of $S$ that is downward closed
and such that $\bbS [p \rst{U}],r \sat \xi$.
Let $f$ be a positional winning strategy of $\eloi$ in the game
$\EG_0:=\EG(\xi,\bbS)@(\xi,r)$.
We define $U \subseteq S$ such that
\[
u \in U \ouriff \text{ there is a } \phi \text{ such that } (\phi,u) 
\text{ is $f$-reachable in $\EG_0$} \text{ and } p \text{ is active in $\phi$}.
\]

It is easy to see that the set $U$ is downward closed.
Indeed, if a position $(\phi,t)$ is reached during an $\EG_0$-match $\Si$ and 
$p$ is not active in $\phi$, then all positions occurring after $(\phi,t)$ will 
be of the form $(\psi,u)$, where $p$ is not active in $\psi$. 

Hence it suffices to show that $U$ is finitely branching.
Fix $u \in U$ and let us show that $\si_{R}(u) \cap U$ is finite.
Let $t \in U$ be a successor of $u$. Since $u$ is the only predecessor of $t$,
by definition of $U$, there must be an $f$-guided match during which a move
occurs from $(\triangle \phi_t,u)$ to $(\phi_t,t)$, where $\triangle \in 
\{ \Box, \dia \}$ and $\phi_t$ is a $p$-active subformula of $\xi$.
Because of the syntactic constraints on $\muML^{W}_{p}$, this can only happen 
if $\triangle = \dia$.
But then $(\dia \phi_t,u)$ is a position which belongs to $\eloi$ and
so $t$ is her choice as dictated by $f$.
From this, it follows that for all $t$ and $t'$ in $R(u) \cap U$, we have
$\phi_t \neq \phi_{t'}$ if $t \neq t'$.
Putting this together with the fact that $\Sfor(\xi)$ is finite, we obtain that
$R(u) \cap U$ is finite. This finishes the proof that $U$ is downward closed
and finitely branching.

It remains to show that  $\bbS [p \rst{U}],r \sat \xi$. Let $\EG$ be the
game $\EG(\xi,\bbS[p \rst{U}])@(\xi,r)$.
We show that $f$ itself is a winning strategy for $\eloi$ in the game $\EG$.
The winning conditions for $\EG_0$ and $\EG$ are the same.
Moreover, the rules of the two games are the same, except when we reach a
position of the form $(p,t)$. So to prove that $f$ is a winning strategy for
$\eloi$ in $\EG$, it suffices to show that if an $f$-guided $\EG$-match
$\Si$ arrives at a position $(p,t)$, then $\bbS  [p \rst{U}],t \sat p$,
that is, $t \in V(p) \cap U$.
Suppose that we are in this situation.
Since $\Si$ is also an $f$-guided $\EG_0$-match and since $f$ is a winning
strategy for $\eloi$ in $\EG_0$, $t$ belongs to $V(p)$.
It remains to show that $t$ belongs to $U$. 
That is, we have to find a $p$-active formula $\phi$ such that $(\phi,t)$ is 
$f$-reachable in $\EG_0$.
Clearly, the formula $p$ itself satisfies these two conditions.
This proves that $f$ is a winning strategy for $\eloi$ in the game $\EG$ and
hence shows that $\bbS [p \rst{U}],r \sat \xi$.
\end{proof}

For the hard part of the theorem, we involve automata.
The particular modal automata that we associate with the finite width property 
are given below; recall that we introduced \emph{bipartite} automata in
Definition~\ref{d:bipaut}.

\begin{defi}
\label{d:w-aut1} 
A bipartite modal automaton $\bbA = (A,B,\Th,\Om)$ belongs to the class 
$\Aut^{W}_{p}$ of \emph{finite-width automata} if the one-step language 
associated with $B$ is the language $\MLone(\Prop\setminus\{p\},B)$, and
the one-step language associated with $A$ is given by the following grammar:
\begin{equation} 
\label{eq:1st-W}
\al \isbnf
   p \divbnf \dia\pi_{0} \divbnf \be 
   \divbnf \al \land \al \divbnf \top \divbnf \al \lor \al \divbnf \bot
\end{equation}
where $\pi_{0} \in \Latt(A)$ and $\be \in \MLone(\Prop\setminus\{p\},B)$.
\end{defi}

In words, an initialized modal automaton $\bbA\init{\ai}$, with $\bbA = (A,\Th,
\Om)$, belongs to the class $\Aut^{W}_{p}$ if $A$ can be partitioned as $A =
A_{0} \uplus A_{1}$ such that 
(0) $\ai$ belongs to $A_{0}$,
(1) $p$ occurs only positively in $\Th(a)$, for $a \in A_{0}$, 
(2) $p$ does not occur in any $\Th(a)$, $a \in A_{1}$,
(3) if $a,b \in A_{0}$ then $a$ may only occur in $\Th(b)$ in the scope of a 
   diamond (not of a box) modality,
and (4) if $a \in A_{0}$ and $b \in A_{1}$ then $a$ may not occur in $\Th(b)$.

\begin{prop}
\label{p:w-aut1}
Let $\bbA = (A, B, \Th, \Om)$ be a bipartite modal automaton in 
$\Aut^{W}_{p}$.
Then there is a translation $\tr_{\bbA}: A \to \muML$ such that
$\tr_{\bbA}(a) \in \muML^{W}_{p}$ for every state $a \in A$.
\end{prop}

\begin{proof}
It suffices to check that the fragment $\muML^{W} \isdef \{ \muML^{W}_{P}\mid 
P \sse_{\om} \Propvar \}$ and the automaton $\bbA$ satisfy the conditions
(i) -- (iii) of Proposition~\ref{p:tr}, with $P = \{ p \}$.

Starting with (i), we need to show that $\muML^{W}$ satisfies the properties
(EP), (SP1), (SP2) and (AC$_{\eta}$) for $\eta \in \{ \mu, \nu \}$.
All these results can be established by routine proofs; we consider the property
(SP1) as an example (the other properties are easier to show).
By a straightforward formula induction on $\phi$ one may show that if $\phi \in 
\muML^{W}_{P \cup \{ x \}}$, then for all $\psi \in \muML^{W}_{P}$ the formula
$\phi[\psi/x]$ belongs to the fragment $\muML^{W}_{P}$.

Confining ourselves to the inductive case where $\phi$ is of the form $\mu y. 
\phi'$, we reason as follows.
Without loss of generality we may assume that $y \not\in \FV{\psi}$, so that 
$\psi \in \muML^{W}_{P\cup \{y\}}$ by the extension property.
Also, if $\mu y. \phi' \in \muML^{W}_{P\cup \{ x \}}$, then by the formulation
rules of $\muML^{W}_{P}$ it must be the case that $\phi' \in \muML^{W}_{P \cup \{x,y\}}$.
It then follows by the inductive hypothesis that $\phi'[\psi/x] \in 
\muML^{W}_{P \cup \{ y \}}$, so that $\phi[\psi/x] \in \muML^{W}_{P}$, again
by definition of the fragment.

To check that (ii) the formula $\Th(a)$ belongs to $\muML^{W}_{\{ p \}\cup A }$ 
for all $a \in A$, one may proceed via a straightforward induction on the 
complexity of the one-step formulas generated by the grammar~\eqref{eq:1st-W}.

Finally, it is easy to verify that there is a translation $\tr: B \to 
\muML(\Prop\setminus \{ p \})$; from this it is immediate that 
(iii)  $\tr(b) \in \muML^{W}_{\{ p \}}$, for all $b \in B$.
\end{proof}

Note that it follows from Proposition~\ref{p:w-aut1} and Proposition~\ref{p:w1} 
that initialized automata in $\IAut^{W}_{p}$ have the finite width property.

Preparing for the main technical result of this section, we define the following
transformation of automata.
Recall from Definition~\ref{d:DMLone} that $\DMLone(\Prop,A)$ denotes the set 
of \emph{disjunctive} one-step formulas over $\Prop$ and $A$, and that 
$\bbA^{\bot}$ is the automaton given in Definition~\ref{d:Abot}.

\begin{defi}
\label{d:w-aut2}
Let $(\cdot)^{W}: \DMLone(\Prop,A) \to \MLone(\Prop,A\uplus A^{\bot})$ be the 
one-step translation given by the following inductive definition:
\[\begin{array}{lll}
(\pi\ybullet\nb B)^{W} & \isdef & {\displaystyle \bigvee}
  \Big\{ \pi \land  \bwsmall \dia B_{1} \land \nb B_{2}^{\bot} 
  \Bigm| B_{1} \cup B_{2} = B
  \Big\}
\\ \bot^{W} & \isdef & \bot
\\ (\al \lor \be)^{W} & \isdef & \al^{W} \lor \be^{W},
\end{array}\]
where $B_{2}^{\bot}$ denotes the set $B_{2}^{\bot} \isdef \{ b^{\bot} \mid b \in 
B_{2} \}$.

Let $\bbA = (A, \Th, \Om)$ be a disjunctive modal automaton which is positive in
$p$.
We define $\bbA^{W}$ as the bipartite automaton $(A^{W},\Th^{W},\Om^{W})$ where 
$A^{W} \isdef A \uplus A^{\bot}$, and the maps $\Th^{W}$ and $\Om^{W}$ are 
given by putting
\[\begin{array}{lll}
   \Th^{W}(a)        & \isdef & \Th(a)^{W} 
\\ \Th^{W}(a^{\bot}) & \isdef & \Th^{\bot}(a)
\end{array}
\hspace*{10mm}\text{ and } \hspace*{10mm}
\begin{array}{lll}
   \Om^{W}(a)        & \isdef & \Om(a)
\\ \Om^{W}(a^{\bot}) & \isdef & \Om(a),
\end{array}\]
for an arbitrary state $a \in A$.
\end{defi}

In words, $\bbA^{W}$ is a bipartite automaton that we obtain from $\bbA$ 
by putting a copy of $\bbA$ `in front of' of a copy of $\bbA^{\bot}$, changing 
the transition map $\Th$ of the initial part $A$ of $\bbA^{W}$ via the one-step
translation $(\cdot)^{W}$.
The final part of the structure $\bbA^{W}$ is isomorphic to the automaton 
$\bbA^{\bot}$, so that we have
\begin{equation}
\label{eq:AWbot}
\bbA^{W}\init{a^{\bot}} \equiv \bbA^{\bot}\init{a^{\bot}}
\end{equation}
for every state $a \in A$.
Finally, observe that, while we define the transformation $(\cdot)^{W}$ for 
disjunctive automata only, the resulting structures are generally \emph{not}
disjunctive.

The following proposition is easy to verify, we leave the details for the
reader.

\begin{prop}
\label{p:w-aut2}
Let $\bbA\init{a}$ be an initialized disjunctive modal automaton which is 
positive in $p$.
Then its transformation $\bbA^{W}\init{a}$ belongs to the class $\IAut^{W}_{p}$.
\end{prop}

We are now ready for the main technical lemma of this section.

\begin{prop}
\label{p:w-aut3}
Let $\bbA\init{a}$ be an initialized disjunctive modal automaton which is 
positive in $p$.
If $\bbA\init{\ai}$ has the finite width property, then $\bbA\init{\ai} \equiv 
\bbA^{W}\init{\ai}$.
\end{prop}

\begin{proof}
Let $\bbA = (A,\Th,\Om)$ be a disjunctive modal automaton, and assume that,
for some state $\ai \in A$, the initialized automaton $\bbA\init{\ai}$ has the 
finite width property.
In order to prove the equivalence of $\bbA\init{\ai}$ and $\bbA^{W}\init{\ai}$,
it suffices to take an arbitrary $\om$-unravelled Kripke tree $(\bbS,r)$ and
prove that 
\begin{equation}
\label{eq:w1}
\bbS,r \sat \bbA\init{\ai} \ouriff \bbS,r \sat \bbA^{W}\init{\ai}.
\end{equation}

We first consider the direction from left to right of \eqref{eq:w1}.
Assume that $\bbS,r \sat \bbA\init{\ai}$, then
it follows from the finite width property of $\bbA\init{a}$ that there is a 
finitely branching subtree $U \sse S$ such that $\bbS[p\rst{U}],r \sat 
\bbA\init{\ai}$.
Without loss of generality we may assume that $U \neq \nada$, or equivalently,
that $r \in U$.
By monotonicity of $\bbA^{W}$ it suffices to show that $\bbS[p\rst{U}],r \sat 
\bbA^{W}\init{\ai}$; that is, we need to supply $\eloi$ with a winning strategy
$h$ in the game $\AG^{W} \isdef \AG(\bbA^{W}, \bbS[p\rst{U}])@(\ai,r)$.
In order to define this strategy, we will make use of two auxiliary strategies:
let $f$ and $g$ be positional winning strategies for $\eloi$ in the acceptance 
games $\AG(\bbA,\bbS[p\rst{U}])$ and $\AG(\bbA^{W},\bbS[p\rst{U}])$ itself,
respectively.
\smallskip

It will be $\eloi$'s goal to maintain the following condition throughout any 
match of the acceptance game $\AG^{W}$:

\noindent($\dag_{W}$)\hspace{3mm}
\begin{minipage}{14cm}
With $\Si = (a_{n},s_{n})_{n\leq k}$ a partial match of $\AG^{W}$, one of the 
following holds:

($\dag_{W}^{1}$)
$(a_{k},s_{k}) \in A \times U$ and $\Si$ corresponds to an $f$-guided match of 
  $\AG(\bbA,\bbS[p\rst{U}])$,
\\($\dag_{W}^{2}$) $(a_{l},s_{l}) \in A^{\bot} \times S$
for some $l \leq k$ such that $\bbS[p\rst{U}],s_{l} \sat 
\bbA^{W}\init{a_{l}}$,
\\\hspace*{12mm} and $(a_{i},s_{i})_{l\leq i \leq k}$ is a $g$-guided 
$\AG(\bbA^{W}, \bbS[p\rst{U}])$-match.
\end{minipage}
\smallskip

In words, $\eloi$ will make sure that the match either stays in $A \times U$
and corresponds to an $f$-guided match of $\AG(\bbA,\bbS[p\rst{U}]$, or it 
moves to $A^{\bot} \times S$ at a moment when it is safe for her to follow the 
strategy $g$.
Let us first see that $\eloi$ can keep this condition during one single round 
of the game.

\begin{claimfirstyv}
\label{cl:w1}
Let $\Si$ be a partial match of $\AG^{W}$ satisfying ($\dag_{W}$).
Then $\eloi$ has a legitimate move guaranteeing that, after any response move 
by $\abel$, ($\dag_{W}$) holds again.
\end{claimfirstyv}

\begin{pfclaim}
Let $\Si = (a_{n},s_{n})_{n\leq k}$ be as in the claim, and distinguish cases.

The easiest case is where $\Si$ satisfies ($\dag_{W}^{2}$): here $\eloi$ can 
simply continue to use her winning strategy $g$.

If $\Si$ satisfies ($\dag_{W}^{1}$), then obviously its final position 
$(a_{k},s_{k})$ is a winning position for $\eloi$ in $\AG(\bbA,\bbS[p\rst{U}])$.
Note that since $s_{k} \in U$, for its coalgebraic unfolding we have 
$\si_{\bbS[p\rst{U}]}(s_{k}) = \si_{\bbS}(s_{k})$, so we may simply denote this
object as $\si(s_{k})$ without causing confusion.
Let $m: \si_{R}(s_{k}) \to \funP A$ be the marking given by her positional 
winning strategy $f$.
By the legitimacy of this move we have that 
$\si(s_{k}), m \satone \pi\ybullet\nb B$ for some 
disjunct $\pi\ybullet \nb B$ of $\Th(a_{k})$.

If $B = \nada$ it immediately follows that $\si_{R}(s_{k}) = \nada$.
It is easy to see that in this case, taking $\eloi$'s move to be the 
\emph{empty} map $m^{W}$ and with $B_{1} = B_{2} = \nada$, we find 
$\si(s_{k}),m^{W} \satone \pi \land \bwsmall \dia B_{1} \land \nb B_{2}^{\bot}$.
This means that $m^{W}$ is a legitimate move for her, while $\abel$ has no 
legitimate response to it.
Thus $\eloi$ immediately wins (and the condition of the claim is satisfied).

Assuming in the sequel that $B \neq \nada$, we arrive at the heart of the proof.
Define
\[\begin{array}{lll}
   B_{1} & \isdef & \{ a \in A \mid a \in m(u) \text{ for some } u \in 
      \si_{R}(s_{k}) \cap U \},
\\ B_{2} & \isdef & \{ a \in A \mid a \in m(t) \text{ for some } t \in 
      \si_{R}(s_{k}) \setminus U \}.
\end{array}\]
Then clearly it follows from $\si(s_{k}), m \satone \nb B$ that $B_{1} \cup 
B_{2} = B$.
A crucial observation is that since $\bbS$ is $\om$-unravelled, and $U$ is 
finitely branching, any $u \in \si_{R}(s_{k}) \cap U$ has a sibling $\ul{u}
\in \si_{R}(s_{k}) \setminus U$ such that $\bbS,u \bis \bbS,\ul{u}$.
Note that this bisimilarity does not hold for the structure $\bbS[p\rst{U}]$, 
but for the structure $\bbS[p\mapsto\nada]$ it does follow that 
\begin{equation}
\label{eq:w2}
\bbS[p\mapsto\nada],u \bis \bbS[p\mapsto\nada],\ul{u}.
\end{equation}

We can now define the desired move for $\eloi$ as the $A^{W}$-marking $m^{W}: 
\si_{R}(s_{k}) \to \funP A^{W}$ given by
\[
m^{W}(t) \isdef 
\left\{ \begin{array}{ll}
     m(t) \cup m(\ul{t})^{\bot} & \text{if } t \in U
  \\ m(t)^{\bot}           & \text{if } t \not\in U,
\end{array}\right.
\]
where $m(t)^{\bot}$ denotes the set $\{ a^{\bot} \mid a \in m(t) \}$.
We claim that $m^{W}$ is a legitimate move for $\eloi$ at position $(a_{k},
s_{k})$ in $\AG(\bbA, \bbS[p\rst{U}])$, and to prove this we need to show that 
\begin{equation}
\label{eq:w3}
\si(s_{k}), m^{W} \satone \Th^{W}(a_{k}).
\end{equation}
To do so, it clearly suffices to prove that 
\begin{equation}
\label{eq:w4}
\si(s_{k}), m^{W} \satone 
\bwsmall \dia B_{1} \land \nb B_{2}^{\bot}.
\end{equation}
For this purpose, first consider an arbitrary state $b \in B_{1}$.
It is immediate by the definitions of $B_{1}$ and $m^{W}$ that $b \in m^{W}(u)$ 
for some $u \in \si_{R}(s_{k}) \cap U$; hence we find that $\si(s_{k}), m^{W} 
\satone \bwsmall \dia B_{1}$, as required.

Now we consider the other conjunct, viz., the formula $\nb B_{2}^{\bot}$.
First take an arbitrary element $b \in B_{2}$; it is immediate from the 
definitions of $B_{2}$ and $m^{W}$ that $b \in m^{W}(u)$ for some $u \in 
\si_{R}(s_{k}) \setminus U$, so that $\si(s_{k}), m^{W} \satone \dia b$, as
required.
Conversely, take an arbitrary successor $t$ of $s_{k}$.
If $ t \not\in U$, it follows from $\si(s_{k}),m \satone \nb B$, the
non-emptiness of $B$ and the definition of $B_{2}$, that $m(t) \cap B_{2}^{\top}
\neq \nada$.
If, on the other hand, $t$ belongs to $U$, then by the same reasoning, but now
applied to its \emph{sibling} $\ul{t} \not\in U$, we find that $m(\ul{t}) \cap 
B_{2}^{\top} \neq \nada$.
In both cases it is immediate by the definition of $m^{W}$ that 
$m^{W}(t) \cap B_{2}^{\top} \neq \nada$, so that $\si(s_{k}),m^{W} \satone 
\Box\bv B_{2}^{\bot}$.
Thus we obtain $\si(s_{k}),m^{W} \satone \nb B_{2}^{\bot}$, which means that we
have established both \eqref{eq:w4} and \eqref{eq:w3}, showing that $m^{W}$ is
a legitimate move indeed.
\smallskip

It remains to show that, playing $m^{W}$, $\eloi$ ensures that ($\dag_{W}$)
continues to hold after any response by $\abel$.
So suppose that $\abel$ picks a basic position $(b,t)$ such that $b \in 
m^{W}(t)$.
There are three cases to distinguish.
First, if $b \in m(t)$ and $t \in U$, then by our assumption on $m$ the 
continuation $\Si\cdot (b,t)$ of $\Si$ clearly satisfies condition 
($\dag_{W}^{1})$.

Second, suppose that $t \in \si_{R}(s_{k}) \setminus U$.
Then by definition of $m^{W}$, $b$ belongs to the $A^{\bot}$-part of $\bbA^{W}$, 
say $b = a^{\bot}$ for some $a \in A$.
From the definition of $m^{W}$ it follows that $a \in m(t)$, and since $m$ is 
part of $\eloi$'s winning strategy $f$, this means that $\bbS[p\rst{U}], t \sat 
\bbA\init{a}$.
Now observe that since $t \not\in U$ and $U$ is downward closed, the entire 
subtree generated by $t$ is disjoint from $U$, so that we find 
$\bbS[p\mapsto\nada], t \sat \bbA\init{a}$.
We may now use Proposition~\ref{p:fs2} and obtain 
$\bbS[p\rst{U}], t \sat \bbA^{\bot}\init{a^{\bot}}$.
By \eqref{eq:AWbot} it is immediate from this that 
$\bbS[p\rst{U}], t \sat \bbA^{W}\init{a^{\bot}}$.
In other words, in this case the continuation match $\Si \cdot (a^{\bot},t)$ 
satisfies condition ($\dag_{W}^{2}$).

Third, we consider the case where $t \in \si_{R}(s_{k}) \cap U$, and $b$, 
belonging to the $A^{\bot}$-part of $\bbA^{W}$, is of the form $b = a^{\bot}$
for some $a \in A$ such that $a \in m(\ul{t})$.
Reasoning as in the previous case, but now for the sibling $\ul{t}$ of $t$,
we find that $\bbS[p\mapsto\nada], \ul{t} \sat \bbA\init{a}$.
But then it follows from \eqref{eq:w2} that $\bbS[p\mapsto\nada], t \sat 
\bbA\init{a}$ as well.
And, again reasoning as before, we find that in this case the match $\Si \cdot
(a^{\bot},t)$ satisfies condition ($\dag_{W}^{2}$) as well.
\end{pfclaim}

Based on Claim~\ref{cl:w1} we may provide $\eloi$ with the following strategy
$h$.
Given a partial match $\Si$, $h$ picks any move for $\eloi$ as given by 
Claim~\ref{cl:w1} in case $\Si$ satisfies ($\dag_{W})$, while $h$ picks a 
random move if $\Si$ does not meet mentioned condition.
We will now prove that this is in fact a winning strategy.

\begin{claimyv}
\label{cl:w2}
Any $h$-guided full match $\Si$ of $\AG^{W}$ is won by $\eloi$.
\end{claimyv}

\begin{pfclaim}
Let $\Si$ be an $h$-guided full match of $\AG^{W}$.
Note that $(\ai,r)$ is the first position of $\Si$, and that by our assumption 
that $r \in U$, the partial match $\Si_{0} \isdef (\ai,r)$ satisfies condition 
($\dag_{W}^{1})$.
It then easily follows from Claim~\ref{cl:w1} that, playing $h$, $\eloi$ will
never get stuck.

To show that $h$ is a winning strategy, we may thus focus on the case where $\Si
= (a_{n},s_{n})_{n<\om}$ is infinite;
again by Claim~\ref{cl:w1} it follows that every initial part $\Si_{l} \isdef
(a_{n},s_{n})_{n\leq l}$ of $\Si$ satisfies ($\dag_{W})$.
From this it is obvious that we may distinguish the following two cases.

If every $\Si_{l}$ satisfies ($\dag_{W}^{1})$, then $\Si$ itself corresponds to  
an $f$-guided full match of $\AG(\bbA^{W},\bbS[p\rst{U}])$.
Since $f$ was assumed to be winning for $\eloi$ in $\AG(\bbA,\bbS[p\rst{U}])$
from position $(a_{0},s_{0})$, this means that the $A$-stream $(a_{n})_{n\in\om}$ 
satisfies the acceptance condition of $\bbA$.
But then this stream satisfies the acceptance condition of $\bbA^{W}$ as well, 
which means that $\Si$ is won by $\eloi$ indeed.

Alternatively, there is a first $l \in \om$ such that $\Si_{l}$ satisfies 
($\dag_{W}^{2})$.
It easily follows from Claim~\ref{cl:w1} that in this case, the match 
$(a_{n},s_{n})_{l\leq n <\om}$ is a $g$-guided full match of 
$\AG(\bbA^{W},\bbS[p\rst{U}])$, which starts at a position, viz., $(a_{l},s_{l})$,
which is winning for $\eloi$.
Clearly then the $A$-stream $(a_{n})_{l \leq n\in\om}$ satisfies the acceptance 
condition of $\bbA^{W}$; from this it is immediate that the full $A$-stream
$(a_{n})_{n\in\om}$ induced by $\Si$ does so as well.
\end{pfclaim}

Clearly, it follows from Claim~\ref{cl:w2} that $h$ is a winning strategy for
$\eloi$ in the game $\AG^{W} = \AG(\bbA^{W}, \bbS[p\rst{U}])@(\ai,r)$.
This proves the direction from left to right of \eqref{eq:w1}.
\medskip

In order to prove the opposite, right-to-left, direction of \eqref{eq:w1}, our
line of reasoning is similar (but simpler).
Assume that $\bbS,r \sat \bbA^{W}\init{\ai}$, with $\ai \in A$.
We will supply $\eloi$ with a winning strategy $h$ in the game $\AG(\bbA, \bbS)
@(\ai,r)$.
For this purpose we will make use of arbitrary but fixed positional winning 
strategies $f$ and $g$ for $\eloi$ in the acceptance games $\AG(\bbA^{W},\bbS)$ 
and $\AG(\bbA,\bbS)$, respectively.

$\eloi$'s strategy in $\AG(\bbA, \bbS)@(\ai,r)$ will be based on maintaining
the following condition:
\smallskip

\noindent($\ddag_{W}$)\hspace{3mm}
\begin{minipage}{14cm}
With $\Si = (a_{n},s_{n})_{n\leq k}$ a partial match of 
$\AG(\bbA, \bbS)@(\ai,r)$, one of the following holds:

($\ddag_{W}^{1}$)
$\Si$ corresponds to an $f$-guided match of $\AG(\bbA^{W},\bbS)$,
\\($\ddag_{W}^{2}$) 
$\bbS,s_{l} \sat \bbA\init{a_{l}}$ for some $l \leq k$, and
$(a_{i},s_{i})_{l\leq i \leq k}$ is a $g$-guided $\AG(\bbA,\bbS)$-match.
\end{minipage}
\smallskip

As before, our main claim is that $\eloi$ can keep condition ($\ddag_{W}$)
during one single round of the game.

\begin{claimyv}
\label{cl:w3}
Let $\Si$ be a partial match of $\AG(\bbA,\bbS)@(\ai,r)$ satisfying 
($\ddag_{W}$).
Then $\eloi$ has a legitimate move guaranteeing that, after any response move 
by $\abel$, ($\ddag_{W}$) holds again.
\end{claimyv}

\begin{pfclaim}
Let $\Si = (a_{n},s_{n})_{n\leq k}$ be as in the claim, and distinguish cases.
If $\Si$ satisfies ($\ddag_{W}^{2}$), $\eloi$ can simply continue to use the
strategy $g$.

If $\Si$ satisfies ($\ddag_{W}^{1}$), then obviously its final position 
$(a_{k},s_{k})$ is a winning position for $\eloi$ in $\AG(\bbA^{W},\bbS)$.
Let $m: \si_{R}(s_{k}) \to \funP A^{W}$ be the marking given by her positional 
winning strategy $f$.
By the legitimacy of this move we have that $\si(s_{k}), m \satone 
\Th^{W}(a_{k})$, which means that
\begin{equation}
\label{eq:w11}
\si(s_{k}), m \satone \pi \land \bwsmall \dia B_{1} \land \nb B_{2}^{\bot},
\end{equation}
where $\pi\ybullet\nb (B_{1} \cup B_{2})$ is some disjunct of $\Th(a_{k})$.

We claim that the $A$-marking $m_{W}: \si_{R}(s_{k}) \to \funP A$, defined by
\[
m_{W}(t) \isdef \{ a \in A \mid a \in m(t) \text{ or } a^{\bot} \in m(t) \}
\]
is the right move for $\eloi$ in the partial match $\Si$.
In order to prove the legitimacy of  $m_{W}$, we show that:
\begin{equation}
\label{eq:w12}
\si(s_{k}),m_{W} \satone \pi\ybullet \nb (B_{1} \cup B_{2}).
\end{equation}
which clearly implies that $\si(s_{k}),m_{W} \satone \Th(a_{k})$.
It easily follows from \eqref{eq:w11} that $\si(s_{k}),m_{W} \satone \pi$,
and so we may focus on the formula $\nb (B_{1} \cup B_{2})$.
First take an arbitrary successor $t$ of $s_{k}$; it follows from \eqref{eq:w11}
that $m(t)$ contains some variable $b^{\bot} \in B_{2}^{\bot}$.
But this immediately gives that $b \in m_{W}(t)$, showing that $\si(s_{k}),m_{W}
\satone \Box\bv (B_{1} \cup B_{2})$.
Conversely, take an arbitrary state $b \in B_{1} \cup B_{2}$; we need to show 
that $\si(s_{k}),m_{W} \satone \dia b$.
Given the definition of $m_{W}$, this easily follows 
from $\si(s_{k}),m \satone \bwsmall \dia B_{1}$ if $b \in B_{1}$, and 
from $\si(s_{k}),m \satone \nb B_{2}^{\bot}$ if $b \in B_{2}$.
Thus we have proved \eqref{eq:w11}.

It is left to show that each response $(b,t)$ of $\abel$ to $\eloi$'s move 
$m_{W}$ constitutes a continuation $\Si \cdot (b,t)$ of $\Si$ that satisfies 
($\ddag_{W}$).
Let $(b,t)$ be an arbitrary such move, that is, an arbitrary pair such that 
$b \in m_{W}(t)$.
Again we distinguish cases: if $b \in m(t)$ then by our assumption of $m$ being
provided by $\eloi$'s strategy $f$, the partial match $\Si \cdot (b,t)$
satisfies condition ($\ddag_{W}^{1}$).
Alternatively, if $b^{\bot} \in m(t)$, then by the assumption that $f$ is a
winning strategy for $\eloi$ in $\AG(\bbA^{W},\bbS)@(\ai,r)$, we find that 
$\bbS,t \sat \bbA^{W}\init{b^{\bot}}$.
By \eqref{eq:AWbot} this obviously implies $\bbS,t \sat
\bbA^{\bot}\init{b^{\bot}}$, whence by Proposition~\ref{p:fs2} and monotonicity 
of $\bbA$ we obtain, subsequently, $\bbS[p\mapsto\nada],t \sat \bbA\init{b}$
and $\bbS,t \sat \bbA\init{b}$.
But then the continuation match $\Si\cdot (b,t)$ satisfies condition
($\ddag_{W}^{2}$).
\end{pfclaim}

We now define a strategy for $\eloi$ in the game $\AG(\bbA,\bbS)@(\ai,r)$:
Given a partial match $\Si$, $h$ picks any move for $\eloi$ as given by 
Claim~\ref{cl:w3} in case $\Si$ satisfies ($\dag_{W})$, while $h$ picks a 
random move if $\Si$ does not meet mentioned condition.
We claim that $h$ is in fact a winning strategy at position $(\ai,r)$, and 
since one may prove this claim in the same manner as above, we leave the 
details for the reader.
\end{proof}

\begin{proof}[{\rm\bf Proof of Theorem~\ref{t:fw}}]
To define the required map $(\cdot)^{W}: \muML(\Prop) \to \muML^{W}_{p}(\Prop)$,
fix a $\muML$-formula $\xi$.

By Fact~\ref{f:fma-aut}(2) there is an initialized disjunctive modal automaton 
$\bbD_{\xi}\init{d_{\xi}}$ that is equivalent to $\xi$; using the constructions
of Definition~\ref{d:mon-aut2} and Definition~\ref{d:w-aut2}, we transform 
$\bbD_{\xi}$ into a finite-width automaton $\bbD_{\xi}^{MW}$. 
Finally, we apply the translation map of Definition~\ref{d:ma-tr}, instantiated 
at the automaton $\bbD_{\xi}^{MW}$, to its state $d_{\xi}$, and we let $\xi^{W}$
be the resulting formula.
Summarizing, we define
\[
\xi^{W} \isdef \tr_{\bbD^{MW}_{\xi}}(d_{\xi}).
\]

It then follows by the Propositions~\ref{p:mon1}, \ref{p:w-aut2}
and~\ref{p:w-aut1} that the
formula $\xi^{W}$ belongs to the fragment $\muML^{W}_{p}$, and by the 
Propositions~\ref{p:mon-aut3} and~\ref{p:w-aut3}, together with the equivalences $\xi 
\equiv \bbD_{\xi}\init{d_{\xi}}$ and $\bbD_{\xi}^{MW}\init{d_{\xi}} \equiv
\tr_{\bbD^{MW}_{\xi}}(d_{\xi})$, that $\xi$ has the finite width property for $p$
iff $\xi \equiv \xi^{W}$.

Finally, all constructions that are involved in the definition of the map
$(\cdot)^{W}: \muML(\Prop) \to \muML^{W}_{p}(\Prop)$ are effective, and can
be performed in elementary time. 
In addition, the problem whether two formulas of the modal $\mu$-calculus are
equivalent or not can be solved in exponential time.
From this it is immediate that the problem whether a given $\mu$-formula has
the finite width property, is decidable, and its time complexity is bounded by
an elementary function on the size of the formula.
\end{proof}

\section{Finite depth property}
\label{sec:fd}

The property that we consider in this section can be informally classified as 
`vertical' in the sense that its definition involves subtrees containing no 
infinite paths.

\begin{defi}
\label{d:d-prop}
A formula $\xi\in \muML(\Prop)$ has the \emph{finite depth property} for 
$p \in \Prop$ if $\xi$ is monotone in $p$, and, for every tree model $(\bbS,s)$,
\[
\bbS,s \sat \xi \ouriff
\bbS[p\rst{U}],s \sat \xi, \mbox{ for some noetherian subtree } U \sse S,
\]
where we call $U$ a \emph{noetherian subtree} of $\bbS$ if it is downward 
closed and contains no infinite paths.
\end{defi}

We will associate the following syntactic fragment of the modal $\mu$-calculus
with this property.

\begin{defi}
\label{d:d-frag}
Given a set $P \sse \Prop$, we define the fragment $\muML^{D}_{P}$ by the 
following grammar:
\[
\phi \isbnf p \divbnf \psi 
   \divbnf \phi\lor\phi \divbnf \phi\land\phi 
   \divbnf \dia \phi \divbnf \Box \phi
   \divbnf \mu x. \phi',
\]
where $\psi$ is a $p$-free formula and $\phi' \in \muML^{D}_{P \cup \{ x \}}$.
In case $P$ is a singleton, say, $P = \{ p \}$, we will write 
$\muML^{D}_{p}$ rather than $\muML^{D}_{\{p\}}$.
\end{defi}

In words, the fragment $\muML^{D}_{p}$ consists of those formulas of which no
$p$-active subformula is in the scope of a greatest fixpoint operator $\nu$.
All formulas from Example~\ref{ex:fmas} belong to $\muML^{D}_{p}$, except 
$\phi_{5}$.

The following theorem states that modulo equivalence, $\muML^{D}_{p}$ is the
syntactic fragment of the modal $\mu$-calculus that captures the finite depth
property, and that it is decidable whether a given $\mu$-formula has this 
property.

\begin{thm}
\label{t:depth}
There is an effective translation which maps a given $\muML$-formula $\xi$ to
a formula $\xi^{D} \in \muML^{D}_{p}$ such that 
\begin{equation}
\label{eq:t-d}
\xi \mbox{ has the finite depth property for } p \ouriff 
\xi \equiv \xi^{D},
\end{equation}
and it is decidable in elementary time whether a given formula $\xi$ has
the finite depth property for $p$.
\end{thm}

First we prove the easy part of the theorem, stating that formulas in the 
fragment $\muML^{D}_{P}$ indeed have the required semantic property.

\begin{prop} 
\label{p:d1}
Every formula $\xi \in \muML^{D}_{p}$ has the finite depth property with respect 
to $p$.
\end{prop}

\begin{proof}
Let $\xi$ be a formula in $\muML_W(p)$, then $\xi$ is obviously positive,
and hence, monotone in $p$.
Fix a tree model $\bbS$ with root $r \in S$.
We have to prove 
\begin{equation} 
\label{eq:deptheasy}
\bbS,r \sat \xi \ouriff
\bbS[p\rst{U}],r \sat \xi, \mbox{ for some noetherian subtree }  U \sse S.
\end{equation}

The direction from right to left follows from the fact that $\xi$ is monotone
in $p$.
For the opposite direction, suppose that $\bbS,r \sat \xi$.
We need to find a noetherian subtree $U$ of $S$ such that $\bbS [p \rst{U}],r 
\sat \xi$.

Let $f$ be a positional winning strategy of $\eloi$ in the game
$\EG_0:=\EG(\xi,\bbS)@(\xi,r)$.
We define $U \subseteq S$ such that
\[
u \in U \ouriff \text{ there is } \phi \text{ such that } (\phi,u) 
\text{ is $f$-reachable in $\EG_0$} \text{ and } p \text{ is active in $\phi$}.
\]

As before it is easy to see that the set $U$ is downward closed, so we omit the
details.
Suppose for contradiction that $U$ contains an infinite path $P$.
We let $\mathcal{A}$ be the set of all finite $f$-guided $\EG_0$-matches $\Si$
such that for all positions $(\phi,u)$ occurring in $\Si$, $u$ belongs to $P$
and $p$ is active in $\phi$. 
Recall that $\sqsubseteq$ denotes the initial-segment relation on paths 
(including matches).

Clearly, the structure $(\mathcal{A},\sqsubseteq)$ is a tree.
Moreover, it is finitely branching as $P$ is a single path, and all the
formulas occurring in matches in $\mathcal{A}$ belong to the finite set
$\Sfor(\xi)$. Next we show that the set $\mathcal{A}$ is infinite.
It suffices to define an injective map $h$ from $P$ to $\mathcal{A}$.
Fix $t$ in $P$. In particular, $t$ belongs to $U$ and by definition of $U$,
there is a formula $\phi$ such that $(\phi,t)$ is $f$-reachable in $\EG_0$ and
$p$ is active in $\psi$.
We let $h(t)$ be a finite $f$-guided $\EG_0$-match with last position $(\phi,t)$.
It is easy to check that any such map $h$ is an injection from $P$ to
$\mathcal{A}$.

By K\"onig's lemma, since $(\mathcal{A},\sqsubseteq)$ is infinite and finitely
branching, it must contain an infinite path.
This infinite path corresponds to an infinite $f$-guided $\EG_0$-match $\Si$ 
such that for all positions $(\phi,t)$ occurring in $\Si$, $t$ belongs to $P$
and $p$ is active in $\phi$.
Since $\xi$ belongs to the fragment $\muML^{D}_{p}$, this can only happen if all
the variables unfolded in $\Si$ are $\mu$-variables.
This implies that $\Si$ is lost by $\eloi$ and thus contradicts the fact that
$f$ is a winning strategy for $\eloi$ in $\EG_0$.
As a consequence, $U$ contains no infinite path.

It remains to show that  $\bbS [p \rst{U}],r \sat \xi$.
Here we omit the details since the proof is similar to that of the corresponding
statement in the proof of Proposition~\ref{p:w1}.
\end{proof}

For the hard part of the theorem, we involve so-called finite-depth automata.

\begin{defi}
\label{d:d-aut1} 
A bipartite modal automaton $\bbA = (A,B,\Th,\Om)$ belongs to the class 
$\Aut^{D}_{p}$ of \emph{finite-depth automata} if the one-step language 
associated with $B$ is the language $\MLone(\Prop\setminus\{p\},B)$, and
the one-step language associated with $A$ is given by the following grammar:
\begin{equation} 
\label{eq:1st-D}
\al \isbnf
   p \divbnf \al_{0} \divbnf \be 
   \divbnf \al \land \al \divbnf \top \divbnf \al \lor \al \divbnf \bot
\end{equation}
where $\al_{0} \in \MLone(\Prop\setminus\{p\},A)$ and $\be \in \MLone(\Prop\setminus\{p\},B)$.
Most importantly, we require that $\Om(a)$ is \emph{odd}, for every $a \in A$.
\end{defi}

In words, an initialized bipartite modal automaton $\bbA\init{\ai}$, with 
$\bbA = (A,B,\Th,\Om)$, belongs to the class $\IAut^{D}_{p}$ if
(1) $p$ occurs only positively in $\Th(a)$, for $a \in A$, 
(2) $p$ does not occur in any $\Th(b)$, $b \in B$,
(3) $\Th(a)$ is a (negation-free) propositional formulas obtained $p$ and 
   $p$-free one-step formulas over $A$ and $B$, respectively,
(4) all states in $A$ have an odd priority.
Note that it follows from (4) that in order to win a match of the satisfiability
game, unless $\abel$ gets stuck, $\eloi$ has to make sure that the automaton 
leaves its initial part at some moment.

\begin{prop}
\label{p:d-aut1}
Let $\bbA = (A_{0},A_{1}, \Th, \Om)$ be a bipartite modal automaton in 
$\Aut^{D}_{p}$.
Then there is a translation $\tr_{\bbA}: A \to \muML$ such that
$\tr_{\bbA}(a) \in \muML^{D}_{p}$ for every state $a \in A_{0}$.
\end{prop}

\begin{proof}
It suffices to check that the fragment $\muML^{D} \isdef \{ \muML^{D}_{P}\mid 
P \sse_{\om} \Propvar \}$ and the automaton $\bbA$ satisfy the conditions
(i) -- (iii) of Proposition~\ref{p:tr}, with $P = \{ p \}$.
In all cases, the proof is routine.
Note that we only need to show that $\muML^{D}$ satisfies the condition 
(AC$_{\mu}$), not (AC$_{\nu}$).
\end{proof}

It follows from Proposition~\ref{p:d-aut1} and Proposition~\ref{p:d1} that any
initialized automaton in $\IAut^{D}_{p}$ has the finite depth property for $p$.

The main technical result of this section will be based on the following
transformation of automata.
Recall that $\bbA^{\bot}$ is the automaton given in Definition~\ref{d:Abot}.

\begin{defi}
\label{d:d-aut2}
Let $\bbA = (A, \Th, \Om)$ be a modal automaton which is positive in $p$.
We define the automaton $\bbA^{D}$ as the structure $(A^{D},\Th^{D},\Om^{D})$
where $A^{D} \isdef A \uplus A^{\bot}$, and the maps $\Th^{D}$ and $\Om^{D}$
are given by putting
\[\begin{array}{lll}
   \Th^{D}(a)        & \isdef & \Th(a)  \lor \Th(a)[b^{\bot}/b \mid b \in A]
\\ \Th^{D}(a^{\bot})                 & \isdef & \Th^{\bot}(a)
\end{array}
\hspace*{10mm}\text{ and } \hspace*{10mm}
\begin{array}{lll}
   \Om^{D}(a)        & \isdef & 1
\\ \Om^{D}(a^{\bot}) & \isdef & \Om(a),
\end{array}\]
for an arbitrary state $a \in A$.
\end{defi}

To obtain the automaton $\bbA^{D}$ from $\bbA$ we put a modified copy of $\bbA$
`in front of' of a copy of $\bbA^{\bot}$.
The modifications consist of changing the priority map $\Om$ on $A$ by assigning
each state $a \in A$ priority 1, and of allowing the same transitions from $A$ 
to the final part $A^{\bot}$ of $\bbA^{D}$ as to the initial part $A$ itself.
As before, this makes the final part of the structure $\bbA^{D}$ isomorphic to
the automaton $\bbA^{\bot}$, so that again we have
\begin{equation}
\label{eq:ADbot}
\bbA^{D}\init{a^{\bot}} \equiv \bbA^{\bot}\init{a^{\bot}}
\end{equation}
for every state $a \in A$.
Finally, observe that we define the transformation $(\cdot)^{D}$ for 
\emph{arbitrary} (that is, not necessarily disjunctive) modal automata.

The following proposition is easy to verify.

\begin{prop}
\label{p:d-aut2}
Let $\bbA\init{a}$ be an initialized modal automaton which is positive in $p$.
Then its transformation $\bbA^{D}\init{a}$ belongs to the class $\IAut^{D}_{p}$.
\end{prop}

We are now ready for the main technical lemma of this section.

\begin{prop}
\label{p:d-aut3}
Let $\bbA\init{a}$ be an initialized modal automaton which is positive in $p$.
If $\bbA\init{\ai}$ has the finite depth property with respect to $p$, then 
$\bbA\init{\ai} \equiv \bbA^{D}\init{\ai}$.
\end{prop}

\begin{proof}
Let $\bbA = (A,\Th,\Om)$ be a modal automaton, and assume that 
$\bbA\init{\ai}$ has the finite depth property for some state $\ai \in A$.
In order to prove the equivalence of $\bbA\init{\ai}$ and $\bbA^{D}\init{\ai}$,
it suffices to take an arbitrary Kripke tree $(\bbS,r)$ and
prove that 
\begin{equation}
\label{eq:d1}
\bbS,r \sat \bbA\init{\ai} \ouriff \bbS,r \sat \bbA^{D}\init{\ai}.
\end{equation}

We first consider the direction from left to right of \eqref{eq:d1}.
Assume that $\bbS,r \sat \bbA\init{\ai}$, then
it follows from the finite depth property of $\bbA\init{a}$ that there is a 
noetherian subtree $U \sse S$ such that $\bbS[p\rst{U}],r \sat \bbA\init{\ai}$.
We may assume that $r \in U$, and that $U$ is \emph{balanced} in the sense that 
for any $u \in U$, either $\si_{R}(u) \cap U = \nada$ or $\si_{R}(u) \sse U$.
(This is without loss of generality: should $U$ itself not meet these 
conditions, then we may proceed with the set $U' \isdef U \cup \{ r \} \cup
\bigcup \{ \si_{R}(u) \mid u \in U, \si_{R}(u) \cap U \neq \nada \}$.)
By monotonicity of $\bbA^{D}$ it suffices to show that $\bbS[p\rst{U}],r \sat 
\bbA^{D}\init{\ai}$; that is, we need to supply $\eloi$ with a winning strategy
$h$ in the game $\AG^{D} \isdef \AG(\bbA^{D}, \bbS[p\rst{U}])@(\ai,r)$.
In order to define this strategy, we will make use of two auxiliary strategies:
let $f$ and $g$ be positional winning strategies for $\eloi$ in the acceptance 
games $\AG(\bbA,\bbS[p\rst{U}])$ and $\AG(\bbA^{D},\bbS[p\rst{U}])$,
respectively.

$\eloi$'s strategy in $\AG^{D}$ will be based on maintaining the following 
condition:
\smallskip

\noindent($\dag_{D}$)\hspace{3mm}
\begin{minipage}{14cm}
With $\Si = (a_{n},s_{n})_{n\leq k}$ a partial match of 
$\AG^{D}$, one of the following holds:

($\dag_{D}^{1}$)
$(a_{k},s_{k}) \in A \times U$ and $\Si$ corresponds to an $f$-guided match of 
  $\AG(\bbA,\bbS[p\rst{U}])$,
\\($\dag_{D}^{2}$) $(a_{l},s_{l}) \in A^{\bot} \times S$
for some $l \leq k$ such that $\bbS[p\rst{U}],s_{l} \sat 
\bbA^{D}\init{a_{l}}$,
\\\hspace*{12mm} and $(a_{i},s_{i})_{l\leq i \leq k}$ is a $g$-guided 
$\AG(\bbA^{D}, \bbS[p\rst{U}])$-match.
\end{minipage}
\smallskip

In words, $\eloi$ will make sure that the match either stays in $A \times U$
while she can play $f$, or it moves to $A^{\bot} \times S$ at a moment when it
is safe for her to follow the strategy $g$.

Let us first see that $\eloi$ can maintain this condition during one single 
round of the game.

\begin{claimfirstyv}
\label{cl:d1}
Let $\Si$ be a partial match of $\AG^{D}$ satisfying 
($\dag_{D}$).
Then $\eloi$ has a legitimate move guaranteeing that, after any response move 
by $\abel$, ($\dag_{D}$) holds again.
\end{claimfirstyv}

\begin{pfclaim}
Let $\Si = (a_{n},s_{n})_{n\leq k}$ be as in the claim, and distinguish cases.
Leaving the easy case, where $\Si$ satisfies ($\dag_{D}^{2}$), for the reader,
we assume that $\Si$ satisfies ($\dag_{D}^{1}$).
Then obviously its final position  $(a_{k},s_{k})$ is a winning position for 
$\eloi$ in $\AG(\bbA,\bbS[p\rst{U}])$.
Since $s_{k} \in U$ we have $\si_{\bbS[p\rst{U}]}(s_{k}) = \si_{\bbS}(s_{k})$,
so we may denote this object as $\si(s_{k})$ without causing confusion.
Let $m: \si_{R}(s_{k}) \to \funP A$ be the marking given by her positional 
winning strategy $f$ at position $(a_{k},s_{k})$.

By the legitimacy of this move we have that 
$\si(s_{k}), m \satone \Th(a_{k})$.
In order to define $\eloi$'s move $m^{D}$ in $\AG(\bbA^{D},\bbS[p\rst{U}])$, we 
distinguish cases.

If $\si_{R}(s_{k}) \sse U$ we simply put $m^{D} \isdef m$, and we leave it for
the reader to verify that this move of $\eloi$ satisfies the requirements 
formulated in the claim.

Alternatively, by our assumption on $U$ we have that $\si_{R}(s_{k}) \cap U 
= \nada$; in this case we put
\[
m^{D}(t) \isdef m(t)^{\bot},
\]
for each $t \in \si_{R}(u)$.
It then easily follows from $\si(s_{k}),m \satone \Th(a)$ that $\si(s_{k}),m^{D} 
\satone \Th^{\bot}(a)$, and so $\si(s_{k}),m^{D} \satone \Th^{D}(a)$ by 
definition of $\Th^{D}$.
In other words, $m^{D}$ is a legitimate move.

To see that by picking this move $\eloi$ maintains the condition
($\dag_{D}$), consider an arbitrary response $(b^{\bot},t)$ of $\abel$ such 
that $b^{\bot} \in m^{D}(t)$.
By definition of $m^{D}$ it follows that $b \in m(t)$, and so by the assumption 
that $m$ is part of $\eloi$'s winning strategy $f$, we obtain $\bbS[p\rst{U}],t
\sat\bbA\init{b}$.
But now we may reason as before (cf.~Claim~\ref{cl:w1} in the proof of 
Proposition~\ref{p:w-aut3}): 
Since $t \not\in U$ and $U$ is downward closed, the entire subtree generated 
by $t$ is disjoint from $U$, so that we find $\bbS[p\mapsto\nada], t \sat 
\bbA\init{b}$.
We may now use Proposition~\ref{p:fs2} and \eqref{eq:ADbot} to obtain 
$\bbS[p\rst{U}], t \sat \bbA^{D}\init{b^{\bot}}$.
In other words, in this case the continuation match $\Si \cdot (b^{\bot},t)$ 
satisfies condition ($\dag_{D}^{2}$).
\end{pfclaim}

Based on Claim~\ref{cl:d1} we may provide $\eloi$ with a winning strategy $h$
in $\AG^{D}$, exactly as in the proof of Proposition~\ref{p:w-aut3}.
This proves the direction from left to right of \eqref{eq:d1}.
\medskip

We now turn to the right-to-left direction of \eqref{eq:d1}, for which we
assume that $\bbS,r \sat \bbA^{D}\init{\ai}$, with $\ai \in A$.
In order to supply $\eloi$ with a winning strategy $h$ in the game $\AG(\bbA, 
\bbS)@(\ai,r)$, we will make use of arbitrary but fixed positional winning 
strategies $f$ and $g$ for $\eloi$ in the acceptance games $\AG(\bbA^{D},\bbS)$ 
and $\AG(\bbA,\bbS)$, respectively.

$\eloi$'s strategy in $\AG(\bbA, \bbS)@(\ai,r)$ will be based on maintaining
the following condition:
\smallskip

\noindent($\ddag_{D}$)\hspace{3mm}
\begin{minipage}{14cm}
With $\Si = (a_{n},s_{n})_{n\leq k}$ a partial match of 
$\AG(\bbA, \bbS)@(\ai,r)$, one of the following holds:

($\ddag_{D}^{1}$)
$\Si$ corresponds to an $f$-guided match of $\AG(\bbA^{D},\bbS)$,
\\($\ddag_{D}^{2}$) 
$\bbS,s_{l} \sat \bbA\init{a_{l}}$ for some $l \leq k$, and
$(a_{i},s_{i})_{l\leq i \leq k}$ is a $g$-guided $\AG(\bbA,\bbS)$-match.
\end{minipage}
\smallskip

Once more our main claim is that $\eloi$ can maintain condition ($\ddag_{D}$)
during one single round of the game.

\begin{claimyv}
\label{cl:d3}
Let $\Si$ be a partial match of $\AG(\bbA,\bbS)@(\ai,r)$ satisfying 
($\ddag_{D}$).
Then $\eloi$ has a legitimate move guaranteeing that, after any response move 
by $\abel$, ($\ddag_{D}$) holds again.
\end{claimyv}

\begin{pfclaim}
Let $\Si = (a_{n},s_{n})_{n\leq k}$ be as in the claim, and distinguish cases.
If $\Si$ satisfies ($\ddag_{D}^{2}$), $\eloi$ can simply continue to use the
strategy $g$.

If $\Si$ satisfies ($\ddag_{D}^{1}$), then obviously its final position 
$(a_{k},s_{k}) \in A \times S$ is a winning position for $\eloi$ in 
$\AG(\bbA^{D},\bbS)$.
Let $m: \si_{R}(s_{k}) \to \funP A^{D}$ be the marking given by her positional 
winning strategy $f$, then we have $\si(s_{k}),m \satone \Th^{D}(a_{k})$ by
legitimacy of $m$.
Now distinguish cases, as to which disjunct of $\Th^{D}(a_{k}) = \Th(a_{k}) 
\lor \Th(a_{k})[b^{\bot}/b\mid b \in A]$ holds at the one-step model 
$(\si(s_{k}),m)$.

In case $\si(s_{k}),m \satone \Th(a_{k})$, we may assume without loss of 
generality that $m(t) \in \funP(A)$, for all $t \in \si_{R}(a_{k})$.
We now simply define $\eloi$'s move $m_{D}$ in $\AG(\bbA, \bbS)@(\ai,r)$ by 
setting $m_{D} \isdef m$.
In this case it is straightforward to verify that $m_{D}$ is legitimate and that 
for any response move $(b,t)$ of $\abel$, the resulting partial match 
$\Si \cdot (b,t)$ satisfies condition ($\ddag_{D}^{1}$).

In case $\si(s_{k}),m \satone \Th(a_{k})[b^{\bot}/b\mid b \in A]$, we may assume
without loss of generality that $m(t) \in \funP(A^{\bot})$, for all $t \in 
\si_{R}(a_{k})$.
We define $\eloi$'s move $m_{D}$ by putting
\[
m_{D}(t) \isdef \{ b \in A \mid b^{\bot} \in m(t) \},
\]
for each $t \in \si_{R}(s_{k})$.
It is straightforward to verify that $\si(s_{k}),m_{D} \satone \Th(a_{k})$, which 
shows that $m_{D}$ is a legitimate move for $\eloi$ in $\AG(\bbA,\bbS)$ at position
$(a_{k},s_{k})$.

Now consider an arbitrary (legitimate) response $(b,t)$ of $\abel$; we claim that 
the resulting continuation $\Si \cdot (b,t)$ of $\Si$ satisfies 
$(\ddag_{D}^{2})$.
To see this, observe that by definition of $m_{D}$ it follows from $b \in 
m_{D}(t)$ that $b^{\bot} \in m(t)$, which means that $(b^{\bot},t)$ is a 
legitimate move for $\abel$ in $\AG(\bbA^{D},\bbS)$ at position $m$.
But since $m$ is part of $\eloi$'s winning strategy $f$, this means that 
$\bbS,t \sat \bbA^{D}\init{b^{\bot}}$.
It then follows from \eqref{eq:ADbot} that 
$\bbS,t \sat \bbA^{\bot}\init{b^{\bot}}$.
Hence by Proposition~\ref{p:fs2} we obtain that 
$\bbS[p\mapsto\nada],t \sat \bbA\init{b}$, and then by monotonicity of $\bbA$
that $\bbS,t \sat \bbA\init{b}$.
This means that $\Si \cdot (b,t)$ satisfies $(\ddag_{D}^{2})$ indeed.
\end{pfclaim}

Based on this claim, we define the usual strategy $h$ for $\eloi$ in 
$\AG(\bbS,\bbS)$.
To see why this strategy is winning for her, we consider an arbitrary 
$h$-guided match $\Si$ of $\AG(\bbS,\bbS)@(\ai,r)$, and focus on the case where
$\Si = (a_{n},s_{n})_{n\in\om}$ is infinite.
Observe that because $\Om^{D}$ assigns an \emph{odd} priority to states in the 
initial part $A$ of $\bbA^{D}$, it cannot be the case that $\Si$ corresponds to 
a \emph{full} $f$-guided $\AG^{D}(\bbA^{D},\bbS)$-match.
Hence, there must be a stage $k \in \om$ when the match moves out of the initial
part of $\bbA^{D}$, which means that $\bbS,s_{k} \sat \bbA\init{a_{k}}$ and
$(a_{n},s_{n})_{k\leq n< \om}$ is a $g$-guided match of $\AG(\bbA,\bbS)$.
From this it is straightforward to derive that $\Si$ is won by 
$\eloi$.
\end{proof}

\begin{proof}[{\rm\bf Proof of Theorem~\ref{t:depth}}]
To define the required map $(\cdot)^{D}: \muML(\Prop) \to \muML^{D}_{p}(\Prop)$,
fix a $\muML$-formula $\xi$.
We define
\[
\xi^{D} \isdef \tr_{\bbA_{\xi}^{\mathit{MD}}}(a_{\xi}),
\]
where $\bbA_{\xi}\init{a_{\xi}}$ is an initialized modal automaton that is 
equivalent to $\xi$, and $\tr_{\bbA_{\xi}^{\mathit{MD}}}$ is the translation 
associated with the automaton $\bbA_{\xi}^{\mathit{MD}}$.
We leave it for the reader to verify that this map satisfies the requirements 
stated by the theorem --- the proof follows the same lines as that of
Theorem~\ref{t:fw}.
\end{proof}

\section{The single branch property}
\label{sec:sb}

As a variation of the finite-width property, we consider the single-branch 
property.

\begin{defi}
\label{d:b-prop}
A formula $\xi \in \muML(\Prop)$ has the \emph{single branch property} for 
$p \in \Prop$ if $\xi$ is monotone in $p$, and, for every tree model $(\bbS,r)$,
\[
\bbS,r \sat \xi \ouriff
\bbS[p\rst{U}],r \sat \xi, \mbox{ for some branch }
  U \sse S,
\]
where a subset $U\sse S$ is a \emph{branch} if $U = \{ s_{n} \mid n < \kappa 
\}$ for some (finite or infinite) path $(s_{n})_{n<\kappa}$ starting at $s_{0}
= r$.
\end{defi}

The syntactic fragment corresponding to this property is the following.

\begin{defi}
\label{d:b-frag}
Given a set $Q \sse \Prop$ such that $p \in Q$, we define the fragment 
$\muML^{B}_{p,Q}$ by the following grammar:
\[
\phi \isbnf p \divbnf q
   \divbnf \psi 
   \divbnf \phi\lor\phi \divbnf p\land\phi \divbnf \psi\land\phi 
   \divbnf \dia \phi 
   \divbnf \mu x. \phi' \divbnf \nu x.\phi',
\]
where $q \in Q$, $\psi\in \muML(\Prop\setminus Q)$ is a $Q$-free formula 
and $\phi' \in \muML^{B}_{p,Q \cup \{ x \}}$.
In case $Q$ is the singleton $\{ p \}$, we will write $\muML^{B}_{p}$ rather 
than $\muML^{B}_{\{p\},\{p\}}$.
\end{defi}

Formulas $\xi$ in this fragment share with those in $\muML^{W}_{p}$ the property 
that no $p$-active subformula of $\xi$ may be in the scope of a box modality.
The difference with $\muML^{W}_{p}$ lies in the role of conjunctions: 
if $\phi_{0} \land \phi_{1}$ is a subformula of $\xi \in \muML^{B}_{p}$, then 
$p$ cannot be active in both conjuncts $\phi_{i}$, unless one of 
these conjuncts is actually identical to $p$.
All formulas from Example~\ref{ex:fmas}, except $\phi_{4}$ and $\phi_{6}$, 
belong to $\muML^{S}_{p}$.

A further difference with the earlier (and later) languages is that 
$\muML^{B}_{p}$ is defined using fragments of the form $\muML^{B}_{p,Q}$,
where the propositional variables in $Q$ have a slightly different role than
$p$.
This difference can be best explained by means of an example: where the formula
$\nu x. \dia (x \land \dia p)$ does not belong to the fragment $\muML^{W}_{p}$
(and does not have the single-branch property), the formula $\nu x. \dia (p 
\land \dia x) \in \muML^{W}_{p}$ does have the property.

Our characterization theorem for the single-branch property reads as follows.

\begin{thm}
\label{t:sb}
There is an effective translation which maps a given $\muML$-formula $\xi$ to
a formula $\xi^{B} \in \muML^{B}_{p}$ such that 
\begin{equation}
\label{eq:t-b}
\xi \mbox{ has the single branch property for } p \ouriff 
\xi \equiv \xi^{B},
\end{equation}
and it is decidable in elementary time whether a given formula $\xi$ has
the single branch property for $p$.
\end{thm}

As before we start with the easy part of the proof.

\begin{prop} 
\label{p:b1}
Every formula $\xi \in \muML^{B}_{p}$ has the single branch property with 
respect to $p$.
\end{prop}

\begin{proof}
Let $\xi$ be a formula in $\muML^{B}_{p}$, then clearly $\xi$ is monotone in $p$.
Fix a tree model $\bbS$ with root $r$. We have to show
\begin{equation}
\label{eq:path1}
\bbS,r \sat \xi \ouriff
\bbS[p\rst{U}],r \sat \xi, \mbox{ for some branch } U \sse S.
\end{equation}
The direction from right to left follows from the monotonicity of $\xi$ in
$p$. 
For the direction from left to right, suppose that $\bbS,r \sat \xi$.
We need to find a branch $U$ such that $\bbS [p \rst{U}], r \sat \xi$.

Since $\xi$ is true at $r$ in $\bbS$, $\exists$ has a positional winning
strategy $f$ in the game $\EG_0=\EG(\xi,\bbS)@(\xi,s)$.

\begin{claimfirstyv}
\label{cl:b0}
For every $k<\om$ there is at most one $f$-guided match $\Si = (\phi_{n},
s_{n})_{n\leq k}$ of length $k$, such that $p$ is active in every $\phi_{n}$
but $\phi_{k} \neq p$.
\end{claimfirstyv}

\begin{pfclaim}
We prove the claim by induction on $k$.
Call a match $\Si = (\phi_{n},s_{n})_{n\leq k}$ \emph{$p$-active} if $p$ is 
active in every $\phi_{n}$, $n \leq k$.

For the base case there is nothing to prove since there is only one match of
length $0$ altogether, viz., the match consisting of the starting position 
$(\xi,r)$.

For the inductive case, it suffices to show that if $\Si = 
(\phi_{n},s_{n})_{n\leq k}$ is a $p$-active match, then it has at most one
$p$-active continuation $\Si\cdot (\phi_{k+1},s_{k+1})$.
To show this, we distinguish cases as to the shape of $\phi_{k}$.

If $\phi_{k}$ is a proposition letter, then $\Si$ is a full match, and thus
has no continuation at all.
(Observe that by assumption we must have $\phi_{k} = p$, but this is not
relevant for the argument.)
If $\phi_{k} = x$ for some bound variable $x$, then $(\phi_{k+1},s_{k+1}) = 
(\de_{x},s_{k})$, and $p$ is active in $\de_{k}$.
If $\phi_{k}$ is either a disjunction or a formula of the form $\dia\phi'$, 
then, since we fixed $\eloi$'s strategy, there is exactly one continuation
$\Si\cdot (\phi_{k+1},s_{k+1})$, and therefore at most one $p$-active such
continuation.
Finally, if $\phi = \phi' \land \phi''$ is a conjunction, then by the 
constraints of well-formed formulas of $\muML^{B}_{p}$, at most one of the 
conjuncts can be active in $p$, unless $\phi'$ or $\phi''$ is identical to $p$.
From this it is immediate that $\Si$ has at most one $p$-active continuation
$\Si\cdot (\phi_{k+1},s_{k+1})$ where $\phi_{k+1} \neq p$.

This finishes the proof of the claim.
\end{pfclaim}

We define $U \subseteq S$ such that
\[
u \in U \ouriff \text{ there is a } \phi \text{ such that } (\phi,u) 
\text{ is $f$-reachable in $\EG_0$} \text{ and } p \text{ is active in $\phi$}.
\]
It easily follows from the claim that $U$ is a branch of $\bbS$.
It is now a routine exercise to verify that $(\xi,r)$ is a winning position
for $\eloi$ in $\EG(\xi,\bbS[p\rst{U}])$ --- she may in fact use the very same 
strategy $f$.
\end{proof}

For the hard part of the proof we involve the following class of automata.

\begin{defi}
\label{d:b-aut1} 
A bipartite modal automaton $\bbA = (A,B,\Th,\Om)$ belongs to the class 
$\Aut^{B}_{p}$ of \emph{single-branch automata} if the one-step language 
associated with $B$ is the language $\MLone(\Prop\setminus\{p\},B)$, and
the one-step language associated with $A$ is given by the following grammar:
\begin{equation} 
\label{eq:1st-S}
\al \isbnf
   p \divbnf \dia a \divbnf \be 
   \divbnf \al \land \be \divbnf \top \divbnf \al \lor \al \divbnf \bot
\end{equation}
where $a \in A$ and $\be \in \MLone(\Prop\setminus\{p\},B)$.
\end{defi}

In words, the two key conditions on a bipartite automaton $\bbA = (A,B,\Th,\Om)$ 
are that (i) the initial part $A$ of $\bbA$ can only be accessed from itself 
through a formula of the form $\dia a$ (with $a \in A$), and that (ii), as 
before, the propositional variable $p$ may ony occur positively in the initial 
part, and not at all in the final part of $\bbA$.

\begin{prop}
\label{p:b-aut1}
Let $\bbA = (A_{0}, A_{1}, \Th, \Om)$ be a bipartite modal automaton in 
$\Aut^{B}_{p}$.
Then there is a translation $\tr_{\bbA}: A \to \muML$ such that
$\tr_{\bbA}(a) \in \muML^{B}_{p}$ for every state $a \in A_{0}$.
\end{prop}

\begin{proof}
We prove this proposition using a (minor) variation of Proposition~\ref{p:tr},
where we define the fragment $\Frag(Q) \isdef \muML^{B}_{p,Q}$ only for subsets
$Q \sse \Propvar$ such that $p \in Q$.
(It is easy to check that this modification does not have an effect on the proof
of the proposition.)
The verification that (i) $\muML^{B} \isdef \{ \muML^{W}_{p,Q} \mid p \in Q 
\sse \Propvar \}$ satisfies the properties (EP), (SP1), (SP2) and (AC$_{\eta}$)
for $\eta \in \{ \mu, \nu \}$, can be established by routine proofs.
We may prove that (ii) for all $a \in A$, the formula $\Th(a)$ belongs to the 
set $\muML^{B}_{\{p\},A}$ by a straightforward formula induction on the 
formulas generated by the grammar \eqref{eq:1st-S}.
And, finally, as in the earlier cases it is easy to show the existence of a 
translation $\tr: B \to \muML(\Prop\setminus \{p\})$; from this it is immediate
that (iii) $\tr(b) \in \muML^{W}_{\{p\}}$, for all $b \in B$.
\end{proof}

Note that it follows from Proposition~\ref{p:b-aut1} and Proposition~\ref{p:b1} 
that initialized automata in $\IAut^{B}_{p}$ have the single branch property.

We now turn to the key definition of this section, viz., the transformation 
$(\cdot)^{B}$ of a disjunctive automaton into an $\Aut^{B}_{p}$-structure.
This construction is very similar to the one we used in the section on the
finite-width property; the difference lies in the one-step translation, where 
now, in order to translate a disjunct $\pi\ybullet\nb B$, we do not consider 
all ways to write $B$ as a union $B = B_{1} \cup B_{2}$, but only the ones 
where $B_{1}$ is either empty or a singleton.

\begin{defi}
\label{d:b-aut2}
Let $(\cdot)^{B}: \DMLone(\Prop,A) \to \MLone(\Prop,A\uplus A^{\bot})$ be the 
one-step translation given by the following inductive definition:
\[\begin{array}{lll}
(\pi\ybullet\nb B)^{B} & \isdef & 
   \big(\pi \land \nb B^{\bot}\big) \lor
      {\displaystyle \bigvee}
        \big\{ \pi \land \dia b \land \nb B_{2}^{\bot} 
        \bigm| \{ b \} \cup B_{2} = B \big\}
\\ \bot^{B} & \isdef & \bot
\\ (\al \lor \be)^{B} & \isdef & \al^{B} \lor \be^{B}.
\end{array}\]

Let $\bbA = (A, \Th, \Om)$ be a disjunctive modal automaton which is positive in
$p$.
We define the automaton $\bbA^{B}$ as the bipartite structure 
$(A^{B},\Th^{B},\Om^{B})$ where $A^{B} \isdef A \uplus A^{\bot}$, and the maps 
$\Th^{B}$ and $\Om^{B}$ are given by putting
\[\begin{array}{lll}
   \Th^{B}(a)        & \isdef & \Th(a)^{B} 
\\ \Th^{B}(a^{\bot}) & \isdef & \Th^{\bot}(a)
\end{array}
\hspace*{10mm}\text{ and } \hspace*{10mm}
\begin{array}{lll}
   \Om^{B}(a)        & \isdef & \Om(a)
\\ \Om^{B}(a^{\bot}) & \isdef & \Om(a),
\end{array}\]
for an arbitrary state $a \in A$.
\end{defi}

We leave it for the reader to verify the following proposition.

\begin{prop}
\label{p:b-aut2}
Let $\bbA\init{a}$ be an initialized disjunctive modal automaton which is 
positive in $p$.
Then its transformation $\bbA^{B}\init{a}$ belongs to the class $\IAut^{B}_{p}$.
\end{prop}

We are now ready for the main technical lemma of this section.

\begin{prop}
\label{p:b-aut3}
Let $\bbA\init{a}$ be an initialized disjunctive modal automaton which is 
positive in $p$.
If $\bbA\init{\ai}$ has the single-branch property, then $\bbA\init{\ai} \equiv 
\bbA^{B}\init{\ai}$.
\end{prop}

\begin{proof}
Let $\bbA = (A,\Th,\Om)$ be a disjunctive modal automaton, and assume that 
$\bbA\init{\ai}$ has the single branch property, where $\ai \in A$.
In order to prove the equivalence of $\bbA\init{\ai}$ and $\bbA^{B}\init{\ai}$,
it suffices to take an arbitrary $\om$-unravelled Kripke tree $(\bbS,r)$ and
prove that 
\begin{equation}
\label{eq:b1}
\bbS,r \sat \bbA\init{\ai} \ouriff \bbS,r \sat \bbA^{B}\init{\ai}.
\end{equation}

Our proof of \eqref{eq:b1} is very similar to that of the corresponding 
statement \eqref{eq:w1} in the proof of Proposition~\ref{p:w-aut3}.
For this reason we will focus on the details where the two proofs differ.

We first consider the direction from left to right of \eqref{eq:b1}.
Assume that $\bbS,r \sat \bbA\init{\ai}$, then it follows from the single branch 
property of $\bbA\init{a}$ that there is a path $(u_{n})_{n<\kappa}$ 
starting at $r$ and such that, with $U \isdef \{ u_{n} \mid n<\kappa\}$, we have 
$\bbS[p\rst{U}],r \sat \bbA\init{\ai}$.
Without loss of generality we may assume that $(u_{n})_{n<\kappa}$ is of 
\emph{maximal} length, that is, either $\kappa = \om$, or $\kappa < \om$ and 
$\si_{R}(u_{\kappa -1}) = \nada$; it follows that in particular, $r \in U$.
By monotonicity of $\bbA^{B}$ it suffices to show that $\bbS[p\rst{U}],r \sat 
\bbA^{B}\init{\ai}$; that is, we need to supply $\eloi$ with a winning strategy
$h$ in the game $\AG^{B} \isdef \AG(\bbA^{B}, \bbS[p\rst{U}])@(\ai,r)$.
In order to define this strategy, we will make use of two auxiliary strategies:
let $f$ and $g$ be positional winning strategies for $\eloi$ in the acceptance 
games $\AG(\bbA,\bbS[p\rst{U}])$ and $\AG(\bbA^{B},\bbS[p\rst{U}])$ itself,
respectively.
By Fact~\ref{f:daut} and the disjunctivity of $\bbA$ we may assume without loss
of generality that at any position $(a,s)$ that is winning for $\eloi$ in 
$\AG(\bbA,\bbS[p\rst{U}])$, the marking picked by $f$ assigns a \emph{singleton} 
to each $t \in \si_{R}(s)$.

The condition that $\eloi$ will maintain when playing $\AG^{B}$ is the 
following.
\smallskip

\noindent($\dag_{B}$)\hspace{3mm}
\begin{minipage}{14cm}
With $\Si = (a_{n},s_{n})_{n\leq k}$ a partial match of 
$\AG^{B}$, one of the following holds:

($\dag_{B}^{1}$)
$(a_{k},s_{k}) \in A \times U$ and $\Si$ is an $f$-guided match of 
  $\AG(\bbA,\bbS[p\rst{U}])$,
\\($\dag_{B}^{2}$) $(a_{l},s_{l}) \in A^{\bot} \times S$
for some $l \leq k$ such that $\bbS[p\rst{U}],s_{l} \sat 
\bbA^{B}\init{a_{l}}$,
\\\hspace*{12mm} and $(a_{i},s_{i})_{l\leq i \leq k}$ is a $g$-guided 
$\AG(\bbA^{B}, \bbS[p\rst{U}])$-match.
\end{minipage}
\smallskip

Observe that this corresponds seamlessly to the condition ($\dag_{B}$)
featuring in the finite-width case.

\begin{claimfirstyv}
\label{cl:b1}
Let $\Si$ be a partial match of $\AG^{B}$ satisfying 
($\dag_{B}$).
Then $\eloi$ has a legitimate move guaranteeing that, after any response move 
by $\abel$, ($\dag_{B}$) holds again.
\end{claimfirstyv}

\begin{pfclaim}
Let $\Si = (a_{n},s_{n})_{n\leq k}$ be as in the claim.
Leaving the easy case, where $\Si$ satisfies ($\dag_{B}^{2}$), as a exercise for 
the reader, we focus on the case where ($\dag_{B}^{1}$) holds.
Here the final position $(a_{k},s_{k})$ of $\Si$ is a winning position for 
$\eloi$ in $\AG(\bbA,\bbS[p\rst{U}])$, and with $m: \si_{R}(s_{k}) \to \funP A$ 
being the marking given by $\eloi$'s positional winning strategy $f$, we have 
$\si(s_{k}), m \satone \pi\ybullet\nb B$ for some 
disjunct $\pi\ybullet \nb B$ of $\Th(a_{k})$.

If $s_{k}$ is the last point of the branch $U$, we define $m^{B}: \si_{R}(s_{k})
\to \funP A^{B}$ by putting $m^{B}(t) \isdef m(t)^{\bot}$, and we leave it for 
the reader to verify that $m^{B}$, as a move for $\eloi$ in 
$\AG(\bbA^{B},\bbS[p\rst{U}]$, satisfies the conditions of the claim.

Alternatively, let $u \in \si_{R}(s_{k}) \cap U$ be the (unique) successor of 
$s_{k}$ on the branch $U$.
By our assumption on $m$, there is a unique state $b\in A$ such that 
$m(u) = \{ b \}$.
As in the finite-width case, $u$ has a sibling $\ul{u} \in \si_{R}(s_{k}) 
\setminus U$ such that $\bbS,u \bis \bbS,\ul{u}$, and hence,
$\bbS[p \mapsto\nada],u \bis \bbS[p\mapsto\nada],\ul{u}$.
Now define $m^{B}: \si_{R}(s_{k}) \to \funP A^{B}$ by putting 
\[
m^{B}(t) \isdef
   \left\{\begin{array}{ll}
      m(t) \cup m(\ul{t})^{\bot} & \text{if } t = u
   \\ m(t)^{\bot}           & \text{if } t \neq u.
   \end{array}\right.
\]
In order to show that $m^{B}$ is a legitimate move for $\eloi$ in $\AG(\bbA^{B},
\bbS[p\rst{U}])$ at position $(a_{k},s_{k})$, we define $B_{2} \isdef 
\bigcup_{t\neq u} m(t)$.
Then clearly $B = \{ b \} \cup B_{2}$ and so it suffices to prove that 
\begin{equation}
\label{eq:b4}
\si(s_{k}), m^{B} \satone \dia b \land \nb B_{2}^{\bot}.
\end{equation}
The proof of \eqref{eq:b4} proceeds exactly as that of \eqref{eq:w4} in the 
finite-width case.
We omit the details, as we do for the proof that, for any response move $(t,c)$
by $\abel$ to $\eloi$'s move $m^{B}$, the resulting continuation $\Si\cdot(t,c)$
satisfies condition ($\dag_{B}$).
\end{pfclaim}

Based on this claim it is straightforward to define a winning strategy $h$ for 
$\eloi$ in $\AG^{B}$.
This proves the direction from left to right of \eqref{eq:b1}.
\medskip

In order to prove the opposite, right-to-left, direction of \eqref{eq:b1}, our
line of reasoning is the same as in the finite-width case.
Again we leave the details for the reader.
\end{proof}

\begin{proof}[{\rm\bf Proof of Theorem~\ref{t:sb}}]
Given a formula $\xi \in \muML(\Prop)$, we define
\[
\xi^{B} \isdef \tr_{\bbD^{\mathit{MB}}_{\xi}}(d_{\xi}).
\]
It is not difficult to verify that this map $(\cdot)^{B}: \muML(\Prop) \to
\muML^{B}_{p}(\Prop)$ satisfies the requirements stated by the theorem ---
the proof follows exactly the same lines as that of Theorem~\ref{t:fw}.
\end{proof}

\section{Continuity}
\label{sec:cont}

In this section we prove one of our main results, namely, we give a syntactic 
characterization of the continuous fragment of the modal $\mu$-calculus.
We recall the definition of continuity.

\begin{defi}
A $\mu$-formula $\xi \in \muML(\Prop)$ is {\em continuous} in $p \in \Prop$ if 
\[
\bbS,s \sat \xi \ouriff  
\bbS[p\rst U],s \sat \xi, \text{ for some finite subset }
  U \sse S
\] 
for every pointed model $(\bbS,s)$.
\end{defi}

We leave it for the reader to verify that continuity implies monotonicity:
Any formula that is continuous in $p$ is also monotone in $p$.

The property of continuity is of interest for at least two reasons: its link
to the well-known topological notion of \emph{Scott continuity}~\cite{gier:comp80}
(which also explains the name `continuity'), and its connection with the notion
of \emph{constructivity}.
\medskip

\noindent\textbf{Scott topology}
Recall that a {\em complete lattice} $\bbP$ is a partially ordered set $(P,
\leq)$ in which each subset has a greatest lower bound (called the {\em meet})
and a least upper bound (called the {\em join}). 
If $U$ is a subset of the lattice, we denote by $\bigwedge U$ the meet of $U$ 
and by $\bigvee U$ the join of $U$. 
For example, for all sets $S$, the power set $\funP(S)$, ordered by set
inclusion, is a complete lattice. 

Given a complete lattice $(P,\leq)$, a subset $D \sse P$ is
\emph{directed} if for every pair $d_{1},d_{2} \in D$ there is a $d\in D$ such 
that $d_{1 }\leq d$ and $d_{2}\leq d$.
A subset $U \sse P$ is called \emph{Scott open} if it upward closed (that is,
if $u \in U$ and $u\leq v$ then $v \in U$), and satisfies, for any directed
$D\sse P$, the property that $U \cap D \neq \nada$ whenever $\bv D \in U$.
It is not hard to prove that the Scott open sets indeed provide a topology,
the so-called \emph{Scott} topology.
For the associated topological notion of continuity, let $(P,\leq)$ and 
$(P',\leq')$ be two complete lattices. 
A map $f : P \to P'$ is {\em Scott continuous} if for all Scott opens
$U' \subseteq P'$, the set $f^{-1} [U']$ is Scott open.
It is a standard result that a map $f: P \to P'$ is Scott continuous iff $f$
preserves \emph{directed joins} (that is, if $D \sse P$ is directed, then 
$f(\bv D) = \bv' f[D]$). 

To connect this to our notion of continuity, recall from the introduction that,
given a formula $\xi$, a proposition letter $p$ and a model $\bbS = (S,R,V)$,
the map $\xi^{\bbS}_{p} :\funP (S) \to \funP(S)$ is defined by
\begin{equation}
\label{eq:sem1}
\xi^{\bbS}_{p} (U) = \{ s \in S \mid \bbS[p \mapsto U], s \sat \xi \}.
\end{equation}
Now the link is given by the following Proposition.
The (routine) proof of this Proposition is left as an exercise for the
reader.
\begin{prop} 
\label{scottappendix}
A $\mu$-formula $\xi$ is continuous in $p$ iff for all models $\bbS=(S,R,V)$,
the map $\xi^{\bbS}_{p} : \funP (S) \to \funP(S)$ is Scott continuous.
\end{prop}

\noindent\textbf{Constructivity}
Basically, a formula is constructive if the ordinal approximation of its
least fixpoint is always reached in at most $\omega$ steps. 

Given a formula $\xi$ and a model $\bbS = (S,R,V)$, we define, by induction
on $i < \om$, a map $\xi_{p}^i : \funP(S) \to \funP(S)$. 
We let $\xi_{p}^0$ be the identity map, and for $i < \om$ define 
$\xi_{p}^{i+1} := \xi^{\bbS}_{p} \circ \xi_{p}^{i}$, where 
$\xi^{\bbS}_{p}$ is as in \eqref{eq:sem1}.
A monotone formula $\xi$ is \emph{bounded} in $p$ if for some natural number
$n$, the least fixpoint of $\xi$ is always reached in $n$ steps (that is,
$\xi_{p}^n(\nada) = \xi_p^{n+1}(\nada)$ for all models $\bbS$), and 
\emph{constructive} in $p$ if the least fixpoint is always reached in 
$\om$ many steps (that is, for all models $\bbS=(S,R,V)$, the least fixpoint 
of the map $\xi_p$ is equal to $\bigcup \{ \xi_p^i (\nada) \mid i < \om 
\}$. 

In~\cite{otto:elim99}, Otto proved that it is decidable in exponential time
whether a basic modal formula is bounded (and whether a given formula of the 
modal $\mu$-calculus is equivalent to a basic modal formula).
But to the best of our knowledge, decidability of constructivity (that is, the 
problem whether a given $\mu$-formula is constructive in $p$) is an open 
problem.
In passing we mention that Czarnecki~\cite{czar:fast10} found, for each 
ordinal $\beta < \omega^2$, a formula $\xi_\beta$ for which $\beta$ is the 
least ordinal such that the least fixpoint of $\xi_{\beta}$ is always reached 
in $\beta$ steps.

The connection between the notions of continuity and constructivity is an
intriguing one.
It is a routine exercise to prove that continuity implies constructivity:
if $\xi$ is continuous in $p$, then it is also constructive in $p$. 
The opposite inclusion does not hold, as the examples 
$\xi_{1}(p) = \Box p \wedge \Box  \Box \bot$ and
$\xi_{2}(p) = \nu x. p \wedge \dia x$ testify.
However, in the previous examples, we have 
$\mu p. \xi_{1} \equiv \mu p. \Box \Box \bot$ and 
$\mu p. \xi_{2} \equiv \mu p. \bot$. 
That is, in each case there is a continuous formula $\psi_{i}$ that is 
equivalent to $\xi_{i}$ `modulo an application of the least fixpoint operation'.
This suggests the following question concerning the link between continuity
and constructivity.
Can we find, for any formula $\xi \in \muML$ which is constructive in $p$, a
continuous formula $\psi$ such that $\mu p. \xi \equiv \mu p.\psi$?
We leave this as an open problem, as we do with the broader question whether 
there is a `nice' syntactic fragment of the modal $\mu$-calculus that captures
constructivity in the sense that a formula $\xi$ is constructive in $p$ iff
it is equivalent to a formula $\xi'$ which belongs to the given fragment (and
which would preferably be effectively obtainable from $\xi$).
\medskip

We now turn to the main result of this section, namely, our characterization
result for continuity.
Our approach here is based on the observation that continuity can be seen as 
the combination of monotonicity, the finite depth property and the finite width
property.

\begin{prop}
\label{p:cont0}
A $\mu$-formula $\xi$ is continuous in $p$ iff it is monotone in $p$ and has
both the finite depth and the finite width property with respect to $p$.
\end{prop}

\begin{proof}
It is fairly easy to see that if a $\mu$-calculus formula $\xi$ is continuous 
in $p$, then it has  both the finite depth and the finite width property for $p$.
For some detail, let $(\bbS,r)$ be a tree model such that $\bbS,r \sat \xi$.
Assuming that $\xi$ is continuous in $p$, we can find a finite set $U \sse S$
such that $\bbS[p\rst{U}] \sat \xi$.
Now consider the set ${\uparrow}U \isdef \{ s \in S \mid U \sse R^{*}[s] \}$.
Since $U$ is finite, the set ${\uparrow}U$ is a subtree of $\bbS$ that is both 
noetherian and finitely branching; and since $U \sse {\uparrow}U$, it follows 
by monotonicity that $\bbS[p\rst{({\uparrow}U)}] \sat \xi$.
This suffices to show that $\phi$ has both the finite depth property and the 
finite width property for $p$.

For the opposite implication, assume that $\xi\in\muML$ has all three 
properties mentioned.
Fix a Kripke structure $\bbS = (S,R,V)$ and a point $s \in S$. 
We have to show
\[
\bbS,s \sat \xi \ouriff 
\bbS[p\rst U],s \sat \xi, \text{ for some finite subset }
  U \sse S.
\]
The direction from right to left follows from the fact that $\xi$ is monotone
in $p$.
For the opposite direction, suppose that $\bbS,s \sat \xi$. 
Let $\bbT$, with root $r$, be a tree unravelling of $(\bbS,s)$, with
$f: \bbT \to \bbS$ denoting the canonical bounded morphism.

Since $\xi$ has the finite width property with respect to $p$, there is a
downward closed subset $U_1 \subseteq T$ which is finitely branching and such
that $\bbT [p \rst{U_1}],r \sat \xi$.
But $\xi$ also has the finite depth property with respect to $p$. 
Hence there is a subset $U_2$ of $T$ such that $U_2$ is downward closed, does 
not contain any infinite path and satisfies $\bbT [p \rst{U_1 \cap U_2}],r \sat
\xi$.
By K\"onig's Lemma, the set $U := U_1 \cap U_2$ is finite.

We claim that 
\[
\bbS[p\rst{f[U]}] \sat \xi.
\]
To see this, define $U' \isdef f^{-1}[f[U]]$, then clearly $U \sse U'$.
By monotonicity it follows from $\bbT [p \rst{U}],r \sat\xi$ that 
$\bbT [p \rst{U'}],r \sat\xi$.
It is straightforward to verify that $\bbT[p\rst{U'}],r \bis \bbS[p \rst{f[U]}],
s$, and so we find that $\bbS[p\rst{f[U]}],s \sat \xi$ by invariance under 
bisimilarity (Fact~\ref{f:bisinv}).
This suffices, since clearly $f[U]$ is finite.
\end{proof}

The syntactic fragment corresponding to continuity can be defined as follows.

\begin{defi} 
\label{def:cont}
Given a set $P \sse \Prop$, we define the fragment $\muML^{C}_{P}$ by the 
following grammar:
\[
\phi \isbnf p \divbnf \psi 
   \divbnf \phi\lor\phi \divbnf \phi\land\phi 
   \divbnf \dia \phi 
   \divbnf \mu x. \phi',
\]
where $\psi\in \muML(\Prop\setminus P)$ is a $P$-free formula
and $\phi' \in \muML^{C}_{P \cup \{ x \}}$.
In case $P$ is a singleton, say, $P = \{ p \}$, we will write 
$\muML^{C}_{p}$ rather than $\muML^{C}_{\{p\}}$.
\end{defi}

Observe that this fragment is the intersection of the fragments $\muML^{W}_{p}$ 
and $\muML^{D}_{p}$ (defined in, respectively, Definition~\ref{d:w-frag} 
and~\ref{d:d-frag}). 
That is, $\muML^{C}_{p}$ consists of those formulas $\xi \in \muML^{M}_{p}$
such that no $p$-active subformula $\psi$ of $\xi$ occurs in the scope of either
a box modality or a greatest fixpoint operator.
For instance, all formulas from Example~\ref{ex:fmas} belong to $\muML^{C}_{p}$,
except $\phi_{4}$ and $\phi_{5}$.

From these observations, the earlier results on $\muML^{M}_{p}$, $\muML^{W}_{p}$
and $\muML^{D}_{p}$, and Proposition~\ref{p:cont0} the following is immediate.

\begin{prop}
\label{p:cont1}
Every formula $\xi \in \muML^{C}_{p}$ is continuous in $p$.
\end{prop}

Our characterization then is as follows.
As mentioned in the introduction, an earlier presentation of this result was 
given by the first author in~\cite{font:cont08}.
Note that van Benthem~\cite{bent:logi92} gave a characterization of the 
\emph{first-order} formulas that are continuous (or `finitely distributive', 
in his terminology) in a given predicate letter $P$.

\begin{thm}
\label{t:cont}
There is an effective translation which, given a $\muML$-formula $\xi$,
computes a formula $\xi^{C} \in \muML_{C}(p)$ such that 
\begin{equation}
\label{eq:t-c}
\xi \mbox{ is continuous in } p \ouriff \xi \equiv \xi^{C},
\end{equation}
and it is decidable in elementary time whether a given formula $\xi$ is 
continuous in $p$.
\end{thm}

The automata characterizing continuity are defined as follows.

\begin{defi}
\label{d:c-aut1} 
A bipartite modal automaton $\bbA = (A_{0},A_{1},\Th,\Om)$ belongs to the class
$\Aut^{C}_{p}$ of \emph{$p$-continuous automata} if, relative to this partition,
$\bbA$ is both a finite-width and a finite-depth automaton, with respect to $p$.
\end{defi}

\begin{prop}
\label{p:c-aut1}
Let $\bbA = (A_{0},A_{1},\Th,\Om)$ be a bipartite modal automaton in 
$\Aut^{C}_{p}$.
Then there is a translation $\tr_{\bbA}: A \to \muML$ such that
$\tr_{\bbA}(a) \in \muML^{C}_{p}$ for every state $a \in A_{0}$.
\end{prop}

The proof of Proposition~\ref{p:c-aut1} is straightforward, on the basis of the 
corresponding proofs for the finite-depth and the finite-width automata.

We are now ready for the main technical lemma of this section.

\begin{prop}
\label{p:c-aut3}
Let $\bbA\init{\ai}$ be an initialized disjunctive modal automaton which is 
positive in $p$.
If $\bbA\init{\ai}$ is continuous in $p$, then $\bbA\init{\ai} \equiv 
\bbA^{\mathit{WD}}\init{\ai}$.
\end{prop}

\begin{proof}
Let $\bbA\init{\ai}$ be as in the formulation of the proposition, then 
$\bbA\init{\ai}$ obviously is monotone in $p$.
The proposition then follows from the following chain of equivalences and 
implications:
\begin{align*}
   & \bbA\init{\ai} \text{ is continuous in } p 
\\ \iff 
   & \bbA\init{\ai} \text{ has both the finite width and the finite depth 
        property for } p
  & \text{(Proposition~\ref{p:cont0})}
\\ \Longrightarrow\; 
   & \bbA\init{\ai} \equiv \bbA^{W}\init{\ai} \text{ and } 
        \bbA\init{\ai} \text{ has the finite depth property for } p
  & \text{(Proposition~\ref{p:w-aut3})}
\\ \iff 
   & \bbA\init{\ai} \equiv \bbA^{W}\init{\ai} \text{ and } 
        \bbA^{W}\init{\ai} \text{ has the finite depth property for } p
  & \text{(obvious)}
\\ \Longrightarrow\;  
   & \bbA\init{\ai} \equiv \bbA^{W}\init{\ai} \text{ and } 
        \bbA^{W}\init{\ai} \equiv \bbA^{\mathit{WD}}\init{\ai}
  & \text{(Proposition~\ref{p:d-aut3})}
\\ \Longrightarrow\; 
   & \bbA\init{\ai} \equiv \bbA^{\mathit{WD}}\init{\ai}
  & \text{(obvious)}
\end{align*}
\end{proof}

\begin{proof}[{\rm\bf Proof of Theorem~\ref{t:cont}}]
Given a formula $\xi \in \muML(\Prop)$, define
\[
\xi^{C} \isdef \tr_{\bbD^{\mathit{MWD}}_{\xi}}(d_{\xi}).
\]
It is then straightforward to show that $\xi \equiv \xi^{C}$ iff $\xi$ is 
continuous in $p$, and this characterization can also be used to prove the
decidability of continuity.
(Alternatively and more efficiently, we can use Proposition~\ref{p:cont0},
together with the decidability of each of the three properties mentioned 
there that are, jointly taken, equivalent to continuity.)
\end{proof}

\section{Full and complete additivity}
\label{sec:ca}

The last two properties of formulas that we look at both concern the way in
which the semantics of the formula depends on the proposition letter $p$ 
being true at some single point.

\begin{defi}
\label{d:add}
A formula $\xi \in \muML$ is \emph{fully additive} in $p$ if we have
\[
\bbS,s \sat \xi \ouriff
\bbS[p\rst{u}],s \sat \xi, \mbox{ for some } u \in V(p),
\]
\emph{completely additive} in $p$ if 
\[
\bbS,s \sat \xi \ouriff
\bbS[p\rst{u}],s \sat \xi, \mbox{ for some } u \in S,
\]
and \emph{normal} if $\bbS[p \mapsto \nada], s \not\sat \xi$, for every pointed
model $(\bbS,s)$.
Here and in the sequel we write $\bbS[p \rst{u}]$ instead of 
$\bbS[p\rst{\{u\}}]$.
\end{defi}

The properties defined in Definition~\ref{d:add} go back to J\'onsson \&
Tarski~\cite{jons:bool52I,jons:bool52II}, as do the terms `completely additive'
and `normal'; we have introduced the term `fully additive' here.
It is easy to see (especially using the definitions in terms of the semantic map 
$\phi^{\bbS}_{p}$ that we gave in the introduction), that full additivity is 
equivalent to the combination of normality and complete additivity.

In the context of modal logic, the property of full additivity is of interest
for at least two reasons: its role in the \emph{duality} theory of modal 
logic~\cite{vene:alge06}, and its link with the notion of \emph{safety} for 
bisimulations.

\paragraph{Discrete duality}
In the algebraic approach to modal logic, two dualities feature prominently:
a topological duality linking modal algebras to certain topological Kripke
frames consisting of a relational structure which is expanded with a nicely
fitting Stone topology, and a \emph{discrete} duality connecting ordinary 
Kripke frames to so-called \emph{perfect} modal algebras.
The latter structures consist of a complete and atomic Boolean algebra, which
is expanded with an additional operation that is fully additive (that is,
preserves all joins of the algebra).
This discrete duality, formulated by Thomasson~\cite{thom:cate75}, is based on
a 1--1 correspondence, going back to J\'onsson \&
Tarski~\cite{jons:bool52I,jons:bool52II}, between the fully additive maps on the
power set of $S$ and the binary relations on $S$.
Here, the relation associated with a fully additive map $f$ on $\funP S$
is given as 
\begin{equation}
\label{eq:dual1}
Q_{f} \isdef \{ (s,s') \mid s \in f(\{s'\}) \},
\end{equation}
while conversely, every binary relation $R$ on $S$ gives rise to the fully
additive map $\mop{R}$ defined by
\begin{equation}
\label{eq:dual2}
\mop{R}(U) \isdef \{ s \in S \mid R[s] \cap U \neq \nada \}.
\end{equation}
In other words, the discrete duality concerns the semantics of the modal 
\emph{diamond}.

\paragraph{Safety for bisimulation}
A second and more specific reason for studying full additivity concerns its
key role in the characterization of formulas that are safe for bisimulation.
To define this notion, consider a formula $\al(x,y)$ 
in some appropriate language for describing Kripke models.
This formula induces, on every Kripke model $\bbS$, a binary relation 
$R_{\al}^{\bbS} \isdef \{ (s,t) \mid \bbS \models \al(s,t) \}$.
Given two models $\bbS$ and $\bbS'$, we call a relation $Z \sse S \times S'$
an \emph{$\al$-bisimulation} if it is a bisimulation for the relations
$R_{\al}^{\bbS}$ and $R_{\al}^{\bbS'}$ (in the sense of 
Definition~\ref{d:bis}, with $R_{\al}^{\bbS}$ and $R_{\al}^{\bbS'}$ replacing
the relations $R$ and $R'$, respectively), and we say that $\al$ is \emph{safe 
for bisimulation} if every ordinary bisimulation is also an $\al$-bisimulation.
This notion was introduced by van Benthem~\cite{bent:expl96}, who also gave a
characterization of the safe fragment of first-order logic, that is, the set 
of first-order formulas $\al(x,y)$ that are safe for bisimulation.

The link with the notion of full additivity is provided by the discrete
duality just described: the idea is that we can encode the transformations
\eqref{eq:dual1} and \eqref{eq:dual2} in the syntax of the ambient logic,
which in our context is monadic second-order logic (\MSO).
More specifically, we already saw how a formula $\al(x,y)$ induces a binary 
relation $R_{\al}^{\bbS}$ on every Kripke model.
Now consider a formula $\be(x;p)$, where $x$ is an individual variable and $p$
is a monadic predicate, which, in the style of modal correspondence theory, we 
will also think of as a proposition letter.
Such an $\MSO$-formula $\be$ induces a map $f^{\bbS}_{\be}$ on any Kripke model
$\bbS$, given by $f^{\bbS}_{\be}(U) \isdef \{ s \in S \mid \bbS[p\mapsto U]
\models \be(s;p) \}$.

Now, given an \MSO-formula $\al(x,y)$, define the formula 
\begin{equation}
\label{eq:dual1a}
\al^{*}(x;p) \isdef \exists y (\al(x,y) \land p(y)),
\end{equation}
where $p$ is a fresh monadic predicate, and conversely, given an \MSO-formula 
$\be(x;p)$ which is fully additive in $p$, define the formula 
\begin{equation}
\label{eq:dual2a}
\be_{*}(x,y) \isdef \be(x)[\lambda z. z{=}y/p(z)],
\end{equation}
where $y$ is a fresh individual variable, and $[\lambda z. z{=}y/p(z)]$ is the 
substitution replacing all atomic formulas of the form $p(z)$ by $z = y$.
It is not hard to show that $\al^{*}(x;p)$ is always fully additive in
$p$, and we leave it for the reader to verify that \eqref{eq:dual1a} and 
\eqref{eq:dual2a} encode, respectively, \eqref{eq:dual2} and \eqref{eq:dual1},
in the sense that 
\[
\mop{R^{\bbS}_{\al}} = f^{\bbS}_{\al^{*}}
\text{ and } 
Q_{f^{\bbS}_{\be}} = R^{\bbS}_{\be_{*}}.
\]
As a manifestation of the discrete duality, we find that $\al \equiv 
(\al^{*}(x;p))_{*}(y)$ and $\be \equiv (\be_{*}(x,y))^{*}(x;p)$.

The key observation, which can be proved by a routine argument, is now that 
\begin{equation}
\label{eq:safe1}
\al(x,y) \text{ is safe for bisimulation iff }
\al^{*}(x;p) \text{ is  bisimulation invariant}.
\end{equation}
At this point we invoke the Janin-Walukiewicz Theorem~\cite{jani:expr96}, which
states that an \MSO-formula $\ga(x)$ is bisimulation invariant iff is equivalent 
to (the standard translation of) a $\muML$-formula $\ga^{\dia}$, which may be 
effectively obtained from $\ga$.
Combining this with \eqref{eq:safe1}, we obtain that
\begin{equation}
\label{eq:safe2}
\al(x,y) \text{ is safe for bisimulation iff }
\al^{*}(x;p) \equiv (\al^{*}(x;p))^{\dia}.
\end{equation}
In particular, an \MSO-formula $\al(x,y)$ is safe for bisimulation iff 
the \MSO-formula $\al^{*}(x;p)$ is equivalent to a $\muML$-formula that is 
fully additive in $p$.

Thus, a syntactic characterization of the fully additive modal $\mu$-formulas
also yields a syntactic characterization of the safe fragment of monadic 
second-order logic:
Suppose that $F_{p} \sse \muML$ characterizes (modulo equivalence) the 
fragment of the modal $\mu$-calculus that characterizes full additivity 
in $p$, then the set $\{ (\mathit{ST}_{x}(\phi))_{*}(x,y) \mid \phi \in
F_{p} \}$ characterizes (modulo equivalence) the bisimulation-safe fragment 
of \MSO; here $\mathit{ST}_{x}: \muML \to \mathtt{MSO}$ denotes some standard
truth-preserving translation mapping $\muML$-formulas to \MSO-formulas with one
free individual variable $x$.
Such a characterization was first obtained 
by Hollenberg~\cite{holl:logi98}; we will come back to his results in
Remark~\ref{r:holl}.
\medskip

We now turn to our syntactic characterizations of the fully and completely 
additive modal fixpoint formulas.

\begin{defi}
\label{d:a-frag}
Given a set $P \sse \Prop$, we define the fragment $\muML^{F}_{P}$ by the 
following grammar:
\[
\phi \isbnf p \divbnf \bot
   \divbnf \phi\lor\phi \divbnf \phi\land\psi \divbnf
   \dia \phi 
   \divbnf \mu x. \phi',
\]
where $p \in P$, $\psi\in \muML(\Prop\setminus P)$ is a $P$-free formula
and $\phi' \in \muML^{F}_{P \cup \{ x \}}$.
Similarly, we define the fragment $\muML^{A}_{p}$ by induction in the following
way:
\[
\phi \isbnf p \divbnf \bot
\divbnf \psi 
\divbnf \phi\lor\phi \divbnf \phi\land\psi 
\divbnf \dia\phi
\divbnf \mu x. \phi,
\]
where $p \in P$, $\psi\in \muML(\Prop\setminus P)$ is a $P$-free formula
and $\phi' \in \muML^{A}_{P \cup \{ x \}}$.
In case $P$ is a singleton, say, $P = \{ p \}$, we will write 
$\muML^{F}_{p}$ and $\muML^{A}_{p}$ rather than, respectively,
$\muML^{F}_{\{p\}}$ and $\muML^{A}_{\{p\}}$.
\end{defi}

We let the grammars of Definition~\ref{d:a-frag} speak for itself.
Of the formulas from Example~\ref{ex:fmas}, $\phi_{0}$, $\phi_{2}$ and 
$\phi_{3}$ belong to $\muML^{F}_{p}$; these formulas also belong to 
$\muML^{A}_{p}$, as does $\phi_{1}$.
The difference between the fragments $\muML^{A}_{p}$ and $\muML^{F}_{p}$ is that 
$p$-free formulas belong to $\muML^{A}_{p}$ but not to 
$\muML^{F}_{p}$ (except the formula $\bot$).

\begin{thm}
\label{t:ca}
(i) There is an effective translation which, given a $\muML$-formula $\xi$,
computes a formula $\xi^{F} \in \muML^{F}_{p}$ such that 
\begin{equation}
\label{eq:t-a}
\xi \mbox{ is fully additive in } p \ouriff \xi \equiv \xi^{F},
\end{equation}
and it is decidable in elementary time whether a given formula $\xi$ is 
fully additive in $p$.

(ii) Similarly, there is an effective translation which, given a $\muML$-formula
$\xi$, computes a formula $\xi^{A} \in \muML^{A}_{p}$ such that 
\begin{equation}
\label{eq:t-p}
\xi \mbox{ is completely additive in } p \ouriff \xi \equiv \xi^{A},
\end{equation}
and it is decidable in elementary time whether a given formula $\xi$ is 
completely additive in $p$.
\end{thm}

In the sequel we will only prove the first part of Theorem~\ref{t:ca}, the proof
for complete additivity is a variant of this.
We first consider the easy direction of Theorem~\ref{t:ca}(i).

\begin{prop} 
\label{p:a1}
Every formula $\xi \in \muML^{F}_{p}$ is fully additive in $p$.
\end{prop}

\begin{proof}
Let $\xi$ be a formula in $\muML^{F}_{p}$, then clearly $\xi$ is monotone in $p$.
Fix a tree model $\bbS$ with root $r$. We have to show
\begin{equation}
\label{eq:a-path1}
\bbS,r \sat \xi \ouriff
\bbS[p\rst{u}],r \sat \xi, \mbox{ for some point } u \in V(p).
\end{equation}
The direction from right to left follows from the monotonicity of $\xi$ in
$p$. 
For the direction from left to right, suppose that $\bbS,r \sat \xi$.
We need to find a point $u \in V(p)$ such that $\bbS [p \rst{u}], r \sat \xi$.

Since $\xi$ is true at $r$ in $\bbS$, $\exists$ has a positional winning
strategy $f$ in the game $\EG_0=\EG(\xi,\bbS)@(\xi,s)$.
Similar to the proof of Proposition~\ref{p:b1}, we can prove the following 
claim.

\begin{claimfirstyv}
\label{cl:a0}
For every $k<\om$ there is at most one $f$-guided match $\Si = (\phi_{n},
s_{n})_{n\leq k}$ of length $k$, such that $p$ is active in every $\phi_{n}$.
\end{claimfirstyv}

In proving this claim, the difference with the single-branch case is that now,
the only $p$-active conjunctions are of the form $p \land \psi$, where $\psi$
is not $p$-active.
Hence, a partial match ending in a position with such a conjunction, will
have exactly one $p$-active continuation.

It follows from Claim~\ref{cl:a0} that there is a unique \emph{maximal} 
$p$-active match $\Si = (\phi_{n},s_{n})_{n < \kappa}$.
\item
\begin{claimyv}
\label{cl:a00}
$\Si$ is finite and its last position is the unique position in $\Si$ of the
form $(p,u)$.
\end{claimyv}

\begin{pfclaim}
It is easy to see that $\Si$ must be finite, since otherwise, being $f$-guided,
it should be won by $\eloi$, while the only $p$-active bound variables of $\xi$ 
are \emph{least} fixpoint variables (cf. the proof of Proposition~\ref{p:d1}
in the finite-depth case).
It should also be clear that $\Si$ can have at most one position of the form
$(p,s)$ (since at such a position the match will be over).

We may thus consider the final position $(\phi,r)$ of $\Si$.
It follows by maximality of $\Si$ that $\phi$ cannot be of the form $x$ (with
$x$ a bound variable of $\xi$), $\dia\phi'$, $\phi_{0} \land \phi_{1}$, or 
$\phi_{0} \lor \phi_{1}$ --- in the latter case, both disjuncts $\phi_{i}$ would
be $p$-active.
Hence the only possibility left is that $\phi = p$ indeed.
\end{pfclaim}

To finish the proof, let $u \in S$ be the state such that the pair $(p,u)$
is the final position of $\Si$.
It is easy to check that $(\xi,r)$ is a winning position for $\eloi$
in $\EG(\xi,\bbS[p\rst{u}])$ --- she may use the very same strategy $f$ as in 
$\EG(\xi,\bbS)$.
\end{proof}

For the proof of the hard direction of Theorem~\ref{t:ca}, we introduce the 
following class of automata.

\begin{defi}
\label{d:a-aut1} 
A bipartite modal automaton $\bbA = (A,B,\Th,\Om)$ belongs to the class 
$\Aut^{F}_{p}$ of \emph{finite-width automata} if the one-step language 
associated with $B$ is the language $\MLone(\Prop\setminus\{p\},B)$, and
the one-step language associated with $A$ is given by the following grammar:
\begin{equation} 
\label{eq:1st-A}
\al \isbnf
   p \divbnf \dia a \divbnf \bot
   \divbnf \be \land \al \divbnf \al \lor \al 
\end{equation}
where $a \in A$ and $\be \in \MLone(\Prop\setminus\{p\},B)$.
\end{defi}

\begin{prop}
\label{p:a-aut1}
Let $\bbA = (A, B, \Th, \Om)$ be a bipartite modal automaton in 
$\Aut^{F}_{p}$.
Then there is a translation $\tr_{\bbA}: A \to \muML$ such that
$\tr_{\bbA}(a) \in \muML^{F}_{p}$ for every state $a \in A$.
\end{prop}

\begin{proof}
Once more we will apply Proposition~\ref{p:tr}.
Since it is fairly obvious that $\bbA$ satisfies the conditions (ii) and (iii)
of mentioned proposition, and that $\muML^{F} \isdef \{ \muML^{F}_{P} \mid 
P \sse_{\om} \Propvar \}$ satisfies the properties (EP), (SP2) and (CA$_{\mu}$),
we only show that $\muML^{F}$ satisfies the first substitution property.

For this purpose we will prove, by induction on $\phi$, that $\phi[\psi/x]$ 
belongs to the fragment $\muML^{F}_{P}$ whenever $\phi \in \muML^{F}_{P \cup 
\{ x \}}$ and $\psi \in \muML^{F}_{P}$.
We leave the easy cases as exercises to the reader, and note that the case where
$\phi = \mu y.\phi'$ is dealt with exactly as in the proof of 
Proposition~\ref{p:w-aut1}.
We focus on the case where $\phi$ is of the form $\chi \land \phi'$ because
$\phi' \in \muML^{F}_{P \cup \{ x \}}$ and $\chi$ is $P \cup \{ x\}$-free.
It then follows that $\chi[\psi/x] = \chi$ is $P$-free, and by induction that 
$\phi'[\psi/x] \in \muML^{F}_{P}$.
But from this it is immediate that $\phi[\psi/x] = \chi[\psi/x]\land 
\phi'[\psi/x]$ belongs to the fragment $\muML^{F}_{P}$ indeed.
\end{proof}

As before, our main result is based on a transformation of an arbitrary
disjunctive automaton into an automaton in the class $\Aut^{F}_{p}$.

\begin{defi}
\label{d:a-aut2}
Let $(\cdot)^{F}: \DMLone(\Prop,A) \to \MLone(\Prop,A \uplus A^{\bot})$ be the 
one-step translation given by the following inductive definition:
\[\begin{array}{lll}
(\pi\ybullet\nb B)^{F} & \isdef & 
  \left\{\begin{array}{ll}
      \pi \land \nb B^{\bot} 
        & \text{if } p \in \pi
   \\ \bot
        & \text{if } p \not\in \pi \text{ and } B = \nada
   \\ {\displaystyle \bigvee}
        \big\{ \pi \land \dia b \land \nb B_{2}^{\bot} 
        \bigm| \{ b \} \cup B_{2} =  B \big\}
        & \text{if } p \not\in \pi \text{ and } B \neq \nada
   \end{array}\right.
\\ \bot^{F} & \isdef & \bot
\\ (\al \lor \be)^{F} & \isdef & \al^{F} \lor \be^{F}.
\end{array}\]

Let $\bbA = (A, \Th, \Om)$ be a disjunctive modal automaton which is positive in
$p$.
Without loss of generality we may assume that for every $a \in A$ and every 
disjunct $\pi\ybullet\nb B$ of $\Th(a)$, either $p$ or $\neg p$ (but not both) 
is a conjunct of $\pi$.
(If not, we may replace $\pi\ybullet\nb B$  with the formula 
$(\pi\land p)\ybullet\nb B \lor (\pi\land \neg p)\ybullet\nb B$.)
We define the automaton $\bbA^{F}$ as the structure $(A^{F},\Th^{F},\Om^{F})$,
where $A^{F} \isdef A \uplus A^{\bot}$, and the maps $\Th^{F}$ and $\Om^{F}$
are given by putting
\[\begin{array}{lll}
   \Th^{F}(a)        & \isdef & \Th(a)^{F} 
\\ \Th^{F}(a^{\bot}) & \isdef & \Th^{\bot}(a)
\end{array}
\hspace*{10mm}\text{ and } \hspace*{10mm}
\begin{array}{lll}
   \Om^{F}(a)        & \isdef & 1
\\ \Om^{F}(a^{\bot}) & \isdef & \Om(a),
\end{array}\]
for an arbitrary state $a \in A$.
\end{defi}

\begin{prop}
\label{p:a-aut2}
Let $\bbA\init{a}$ be an initialized disjunctive modal automaton which is 
positive in $p$.
Then its transformation $\bbA^{F}\init{\ai}$ belongs to the class 
$\IAut^{F}_{p}$.
\end{prop}

\begin{prop}
\label{p:a-aut3}
Let $\bbA\init{a}$ be an initialized disjunctive modal automaton which is 
positive in $p$.
If $\bbA\init{\ai}$ is fully additive in $p$, then $\bbA\init{\ai} \equiv 
\bbA^{F}\init{\ai}$.
\end{prop}

\begin{proof}
Let $\bbA = (A,\Th,\Om)$ be a disjunctive modal automaton, and assume that,
for some state $\ai \in A$, $\bbA\init{\ai}$ is fully additive in $p$.
In order to prove the equivalence of $\bbA\init{\ai}$ and $\bbA^{F}\init{\ai}$,
it suffices to take an arbitrary $\om$-unravelled Kripke tree $(\bbS,r)$ and
prove that 
\begin{equation}
\label{eq:a1}
\bbS,r \sat \bbA\init{\ai} \ouriff \bbS,r \sat \bbA^{F}\init{\ai}.
\end{equation}

Our proof of \eqref{eq:a1} is very similar to that of the corresponding 
statements in the proofs of Proposition~\ref{p:b-aut3} and 
Proposition~\ref{p:d-aut3}; for this reason we will be brief, often referring
for details to these earlier proofs.
\medskip

We first consider the direction from left to right of \eqref{eq:a1}.
Assume that $\bbS,r \sat \bbA\init{\ai}$, then it follows by the full additivity
of $\bbA\init{a}$ in $p$ that there is a point $u \in V(p)$ such that 
$\bbS[p\rst{u}],r \sat \bbA\init{\ai}$.
Let $\rho = (u_{n})_{n\leq m}$ be the unique path from $r = u_{0}$ to 
$u = u_{m}$, and define $U \isdef \{ t_{n} \mid n\leq m \}$ to be the set of 
points on this path.
By monotonicity of $\bbA^{F}$ it suffices to show that $\bbS[p\rst{u}],r \sat 
\bbA^{F}\init{\ai}$; that is, we need to supply $\eloi$ with a winning strategy
$h$ in the game $\AG^{F} \isdef \AG(\bbA^{F}, \bbS[p\rst{u}])@(\ai,r)$.
We will write $\si'_{V} \isdef \si_{\bbS[p\rst{u}]}$ and $\si' \isdef 
(\si'_{V},\si_{R})$ for the coalgebraic unfolding map of the Kripke structure
$\bbS[p\rst{u}]$.

In order to define the strategy $h$, we will make use of positional winning 
strategies $f$ and $g$ for $\eloi$ in the acceptance games 
$\AG(\bbA,\bbS[p\rst{u}])$ and $\AG(\bbA^{F},\bbS[p\rst{u}])$ itself,
respectively.
By Fact~\ref{f:daut} we may assume without loss generality that at any position 
$(a,s)$ that is winning for $\eloi$ in $\AG(\bbA,\bbS[p\rst{u}])$, the marking
picked by $f$ assigns a \emph{singleton} to each $t \in \si_{R}(s)$.

The condition that $\eloi$ will maintain when playing $\AG^{F}$ is the 
following.
\smallskip

\noindent($\dag_{F}$)\hspace{3mm}
\begin{minipage}{14cm}
With $\Si = (a_{n},s_{n})_{n\leq k}$ a partial match of 
$\AG^{F}$, one of the following holds:

($\dag_{F}^{1}$)
$(a_{k},s_{k}) \in A \times U$ and $\Si$ is an $f$-guided match of 
  $\AG(\bbA,\bbS[p\rst{u}])$,
\\($\dag_{F}^{2}$) $(a_{l},s_{l}) \in A^{\bot} \times S$
for some $l \leq k$ such that $\bbS[p\rst{u}],s_{l} \sat 
\bbA^{F}\init{a_{l}}$,
\\\hspace*{12mm} and $(a_{i},s_{i})_{l\leq i \leq k}$ is a $g$-guided 
$\AG(\bbA^{F}, \bbS[p\rst{u}])$
\end{minipage}

\begin{claimfirstyv}
\label{cl:a1}
Let $\Si$ be a partial match of $\AG^{F}$ satisfying 
($\dag_{F}$).
Then $\eloi$ has a legitimate move guaranteeing that, after any response move 
by $\abel$, ($\dag_{F}$) holds again.
\end{claimfirstyv}

\begin{pfclaim}
The proof of this claim is a subtle variation on that of the corresponding
claim in the single-branch case.
Let $\Si = (a_{n},s_{n})_{n\leq k}$ be as in the claim.
Leaving the easy case, where $\Si$ satisfies ($\dag_{F}^{2}$), as a exercise for 
the reader, we focus on the case where ($\dag_{F}^{1}$) holds.

Here the path $(s_{n})_{n\leq k} = (u_{n})_{n\leq k}$ is an initial segment of 
the branch $U$, and the final position $(a_{k},s_{k})$ of $\Si$ is a winning 
position for $\eloi$ in $\AG(\bbA,\bbS[p\rst{u}])$.
With $m: \si_{R}(s_{k}) \to \funP A$ being the marking given by $\eloi$'s
positional winning strategy $f$, we have $\si'(s_{k}), m \satone \Th(a_{k})$.
We now distinguish cases.

If $s_{k} = u$ is the last point of the branch $U$, then $\bbS[p\rst{u}],s_{k} 
\sat p$, and so we have $\si'(s_{k}), m \satone \pi\ybullet\nb B$ for some 
disjunct $\pi\ybullet\nb B$ of $\Th(a_{k})$ such that $p$ is a conjunct of $\pi$.
(Here we use the fact that for every disjunct $\pi\ybullet\nb B$ of $\Th(a)$, 
either $p$ or $\neg p$ is a conjunct of $\pi$).
We define the $\bbA^{F}$-marking $m^{F}$ on $\si_{R}(s_{k})$ by putting 
$m^{F}(t) \isdef m(t)^{\bot}$, and we leave it for the reader to verify that
$m^{F}$, as a move for $\eloi$ in $\AG(\bbA^{F},\bbS[p\rst{u}])$, satisfies the 
conditions of the claim.

If $s_{k} \neq u$, then $p$ is false at $s_{k}$.
Let $u_{k+1} \in \si_{R}(s_{k})$ be the (unique) successor of $s_{k} = u_{k}$
on the branch $U$.
By our assumption on $m$, there is a unique state $b\in A$ such that 
$m(u_{k+1}) = \{ b \}$.
Note that in this situation, we have that $\si(s_{k}), m \satone \pi\ybullet
\nb B$ for some disjunct $\pi\ybullet\nb B$ of $\Th(a_{k})$ such that 
$B \neq \nada$ and $\neg p$ is a conjunct of $\pi$.
From this we may infer that the translation $(\pi\ybullet\nb B)^{F}$ is not equal
to $\bot$, but of the form $\bv \big\{ \pi \land \dia b \land \nb B_{2}^{\bot} 
\bigm| \{ b \} \cup B_{2} =  B \big\}$.
In order to define a suitable marking $m^{F}$, we continue as in the 
single-branch case.
Let $u$ be the unique successor of $s_{k}$ in $U$, let $b$ be the unique state
in $A$ such that $m(u) = \{ b \}$, and let $\ul{u} \in \si_{R}\setminus U$ be
some sibling of $u$ such that $\bbS,u \bis \bbS, \ul{u}$.
Define
\[
m^{F}(t) \isdef
   \left\{\begin{array}{ll}
      m(t) \cup m(\ul{t})^{\bot} & \text{if } t = u
   \\ m(t)^{\bot}           & \text{if } t \neq u.
   \end{array}\right.
\]
The verification, that with this definition the marking $m^{F}$ meets all the 
specifications of the claim, is exactly as in the proof of 
Proposition~\ref{p:b-aut3}, and so we omit the details.
\end{pfclaim}

On the basis of Claim~\ref{cl:a1} we may define, in the by now familiar way,
a strategy $h$ for $\eloi$ which is winning for her in $\AG(\bbA^{F},
\bbS[p\rst{u}])$.
We omit the details.
\medskip

The opposite direction of \eqref{eq:a1} can be proved by a similar argument as
in the proof of Proposition~\ref{p:d-aut3}, so again we omit the details.
\end{proof}

\begin{proof}[{\rm\bf Proof of Theorem~\ref{t:ca}}]
As mentioned earlier on, we only cover part (i) of the theorem explicitly, part 
(ii) can be proved by a fairly obvious variation of this.

To define the required map $(\cdot)^{F}: \muML(\Prop) \to \muML^{F}_{p}(\Prop)$,
fix a $\muML$-formula $\xi$.
We define
\[
\xi^{F} \isdef \tr_{\bbD_{\xi}^{\mathit{MF}}}(d_{\xi}),
\]
where $\bbD_{\xi}\init{d_{\xi}}$ is an initialized disjunctive modal automaton 
that is equivalent to $\xi$, and $\tr_{\bbD_{\xi}^{\mathit{MF}}}$ is the 
translation associated with the automaton $\bbD_{\xi}^{\mathit{MF}}$.
We leave it for the reader to verify that this map satisfies the requirements 
stated by the theorem --- the proof follows the same lines as that of
Theorem~\ref{t:fw}.
\end{proof}

We finish this section with a series of remarks that provide some context to our
results.

\begin{rem}
\label{r:pdl}
There are interesting connections between the fragments $\muML^{F}_{P}$ and
$\muML^{A}_{P}$, and the language PDL of \emph{propositional dynamic 
logic}~\cite{hare:dyna00}.
Since PDL is by nature a \emph{poly-modal} language, to make our point we
momentarily switch to the poly-modal $\mu$-calculus.
Carreiro \& Venema~\cite{carr:pdl14} showed that PDL has the same expressive 
power as the fragment of $\muML$ in which the formula construction $\mu x. \phi$
is allowed only if $\phi$ is completely additive with respect to $x$.
More precisely, define the set $\mu_{A}\ML$ of formulas by the following 
grammar:
\[
\phi \isbnf q
\divbnf \neg \phi
\divbnf \phi\lor\phi 
\divbnf \mop{a}\phi
\divbnf \mu q. \phi',
\]
where $q$ is an arbitrary proposition letter, and $\phi'$ belongs to the 
fragment $\muML^{A}_{q} \cap \mu_{A}\ML$.
Then there are inductive, truth-preserving translations from PDL to 
$\mu_{A}\ML$ and vice versa~\cite{carr:pdl14}.
\end{rem}

\begin{rem}
\label{r:holl}
Hollenberg's characterization~\cite{holl:logi98} of the fully additive
fragment\footnote{%
   Note that Hollenberg's terminology clashes with ours: what he calls
   `completely additive' is what we call `fully additive'.
   }
of the modal $\mu$-calculus has a strong connection with propositional
dynamic logic as well.
He defines the sets of extended $\mu$-formulas $\phi$ and so-called 
$\mu$-programs $\pi$ by the following simultaneous induction (again we take a 
poly-modal perspective):
\begin{eqnarray*}
\phi &\isbnf& q 
   \divbnf \neg\phi \divbnf \phi\lor\phi 
   \divbnf \mop{\pi}\phi 
   \divbnf \mu q.\phi
\\ 
\pi  &\isbnf&  a 
   \divbnf \phi? \divbnf \pi + \pi 
   \divbnf \pi;\pi \divbnf \pi^{*},
\end{eqnarray*}
where $q$ is an arbitrary propositional variable, and $a$ is an atomic program; 
in $\mu q. \phi$, $q$ may only occur positively in $\phi$.
Hollenberg proves that a formula $\xi$ is fully additive in $p$ iff $\xi$ 
is equivalent to a formula of the form $\mop{\pi}p$, where $\pi$ is a $p$-free
$\mu$-program.

Comparing Hollenberg's result to ours, while his characterization is clearly 
well-suited to find the safe fragment of monadic second-order logic, our result 
has the advantage of directly providing a characterizing fragment inside the 
modal $\mu$-calculus.
But in any case, there are direct translations between our fragment and 
Hollenberg's.

From Hollenberg's fragment to ours, by a simultaneous induction on formulas 
and programs one may define a translation $(\cdot)^{\tau}$ mapping a formula 
$\phi$ in Hollenberg's language to a formula $\phi^{\tau}\in\muML$, and, for
each $\mu$-program $\pi$, a function $f_{\pi}: \muML \to \muML$,
in such a way that $f_{\pi}$ restricts to the fragment $\muML^{F}_{P}$ if $\pi$
is $P$-free.
Some key clauses in this definition are $(\mop{\pi}\phi)^{\tau} \isdef 
f_{\pi}(\phi^{\tau})$, $f_{\psi?}(\phi) \isdef \psi \land \phi$, and 
$f_{\pi^{*}}(\phi) \isdef \mu x. \phi \lor f_{\pi}(x)$, where $x$ is 
a fresh variable.

Conversely, by a straightforward formula induction one may provide, for each
formula $\xi \in \muML^{F}_{P}$, a collection $\{ \pi_{p} \mid p \in P\}$ of
$P$-free $\mu$-programs such that 
\[
\xi \equiv \bv \{ \mop{\pi_{p}}p \mid p \in P \}.
\]
The key induction step here is for a formula of the form $\mu x. \xi$, where
we may infer from the above equivalence that $\mu x. \xi \equiv \bv 
\{ \mop{\pi_{x}^{*};\pi_{p}} p \mid p \in P \setminus \{ x \} \}$.

We refrain from giving more details here, referring the interested reader to
section 3 of Carreiro \& Venema~\cite{carr:pdl14}, where very similar
translations between PDL and $\mu_{F}\ML$ are defined (cf.~Remark~\ref{r:pdl}),
or to Carreiro~\cite{carr:frag15}.
\end{rem}

\begin{rem}
\label{r:finadd}
As another variation of the properties of \emph{full} and \emph{complete}
additivity, a formula 
$\xi \in \muML$ is {\em finitely additive} in $p \in \Prop$ if, for every 
Kripke model $\bbS$,
\[
\xi_{p}^{\bbS}\Big(\bigcup \mathcal{X} \Big) = 
\bigcup \Big\{ \xi_{p}^{\bbS}(X) \mid X \in \mathcal{X} \Big\},
\]
for any \emph{finite} collection $\mathcal{X}$ of subsets of $S$.
This condition can be equivalently expressed by requiring that the formula $\xi$ 
is both \emph{additive} ($\xi(p \lor p') \equiv \xi(p) \lor \xi(p')$) and
\emph{normal} ($\xi(\bot) \equiv \bot$) in $p$.
In the case of basic modal logic, the two properties can be shown to be 
equivalent (through a straightforward argument based on finite trees of depth 
not exceeding the modal depth of $\xi$), but this is not so in the case of 
the modal $\mu$-calculus.
For instance, consider the formula $\nu y. \mu x. (p \land \dia y) \lor \dia x$, 
expressing the existence of an infinite path, starting at the current state, 
where $p$ holds infinitely often.
It is easy to see that this formula is finitely but neither fully nor completely
additive in $p$.

We leave it as an open problem to characterize the finitely $p$-additive 
fragment of $\muML$.
\end{rem}

\section{Conclusions}
\label{sec:conc}

We finish the paper with drawing some conclusions, listing some issues for 
discussion, and suggesting some questions for further research.
\medskip

This paper contributes to the theory of the modal $\mu$-calculus by proving
some model-theoretic results.
For a number of semantic properties pertaining to formulas of the modal 
$\mu$-calculus, we provided a corresponding syntactic fragment, showing that
a $\mu$-formula $\xi$ has the given property iff it is equivalent to a formula 
$\xi'$ in the corresponding fragment.
Since this formula $\xi'$ will always be effectively obtainable from $\xi$, as
a corollary, for each of the properties under discussion, we prove that it is
decidable in elementary time whether a given $\mu$-calculus formula
has the property or not.

The properties that we study have in common that they all concern the dependence 
of the truth of the formula at stake, on a single proposition letter $p$.
In each case the semantic condition on $\xi$ will be that $\xi$, if true 
at a certain state in a certain model, will remain true if we restrict the set 
of states where $p$ holds, to a special subset of the state space.
Important examples include the properties of full additivity and continuity, 
where the special subsets are the singletons and the finite sets, respectively.

Our proofs for these characterization results are fairly uniform in nature,
employing the well-known correspondence between formulas of the modal 
$\mu$-calculus, and \emph{modal automata}.
In fact, the effectively defined maps on formulas are induced by rather simple 
transformations on modal automata, based on composing a bipartite automaton 
$\bbA'$ from an arbitrary (disjunctive) automaton $\bbA$, where the final part
of $\bbA$ consists of the automaton $\bbA^{\bot}$ and its initial part of
another modification of $\bbA$.
This modification is always obtained by applying a straightforward one-step 
translation to the transition map of $\bbA$, by redefining its priority map,
or by a combination of these operations.

\subsection*{Discussion}

\begin{enumerate}
\item
As mentioned in the introduction, pure logic-based proofs for our results are 
possible in almost all cases --- the exception being the single-branch property
where we only have automata-theoretic proofs.
In fact, logic-based proofs were given in the dissertation of the first
author~\cite{font:moda10}.
The main advantage of the automata-theoretic approach is that it allows for 
transparant and uniform proofs based on simple transformations of automata.

In any case, the difference between the two approaches should not be 
exaggerated.
Recall that the particular shape of our automata is logic-based: the transition 
map of our structures uses so-called \emph{one-step formulas}, and many of our 
proofs are based on semantic properties of and syntactic manipulations on these
very simple modal formulas.
In some sense then, our paper is also a contribution to the model theory of 
modal automata.

\item
As mentioned in the introduction we have not undertaken an in-depth study of 
the computational \emph{complexity} of the various problems of which we  
established the decidability.
It should be clear that the algorithms that we have presented here are not
optimal.
In particular, in order to find out whether a formula $\xi$ has, say, the 
finite width property, it is not needed to compute its translation $\xi^{W}$:
it suffices to check whether the initialized automata $\bbD\init{d_{\xi}}$ and
$\bbD_{\xi}^{W}\init{d_{\xi}}$ are equivalent.
To obtain a good upper bound here, one should know the exact size and weight 
of the automaton $\bbD_{\xi}$ in terms of the size $\sz{\xi}$ of $\xi$. 
Fact~\ref{f:fma-aut}(ii) gives a doubly exponential weight for $\bbD_{\xi}$,
but we conjecture that a tighter bound is possible. 
We leave this, and other complexity-theoretic matters, as questions for 
further research.

\item
There are some variations of our results that are not hard to prove.
To start with, all characterization results (and their proofs) can be easily
restricted to the setting of basic (i.e., fixpoint-free) modal logic.
For instance, if we define the fragment $\ML^{C}_{p}$ by the following 
grammar:
\[
\phi \isbnf p \divbnf \psi 
   \divbnf \phi\lor\phi \divbnf \phi\land\phi 
   \divbnf \dia \phi 
\]
where $\psi\in \muML(\Prop\setminus \{ p \})$ is a $p$-free formula, then we 
can show that our map $(\cdot)^{C}$ maps formulas in $\ML$ to $\ML^{C}_{p}$.
As a result we find that a basic modal logic formula $\xi$ is continuous in
$p$ iff $\xi \equiv \xi^{C}$, so that $\ML^{C}_{p}$ characterizes 
continuity-in-$p$ for basic modal logic.

\item
Recall that in the presentation of the language $\muML$ (as in 
Definition~\ref{d:syn}), the standard restriction on the occurrence of the 
least fixpoint operator $\mu x$ is that it can be applied only to formulas
that are positive in $x$, i.e., belong to the language $\muML^{M}_{x}$.
We get interesting logics by restricting the application of $\mu$-operators
even further.
This applies in particular to the fragments $\muML^{D}$, $\muML^{C}$ and 
$\muML^{A}$ discussed in this paper.
For $Q \in \{ D, C, A \}$, let $\mu_{Q}\ML$ be the version of the modal 
$\mu$-calculus of which the formulas are given by the following grammar:
\begin{equation}
\label{eq:mu-syn1a}
\phi \isbnf p 
   \divbnf \neg\phi \divbnf \phi\lor\phi 
   \divbnf \dia\phi 
   \divbnf \mu x.\phi',
\end{equation}
where $p$ is a propositional variable, and the formation of the formula 
$\mu x.\phi$ is subject to the constraint that the formula $\phi'$ belongs to 
the fragment $\muML^{Q}_{x}$.
We already saw in Remark~\ref{r:pdl} that the language $\mu_{A}\ML$ is 
effectively equivalent to PDL.
It is not hard to prove that the logic $\mu_{D}\ML$ is effectively equivalent
to the \emph{alternation-free} fragment of the modal $\mu$-calculus, whereas
it seems that the logic $\mu_{C}\ML$ has not been used or studied much (although
it was mentioned under the name `$\om$-$\mu$-calculus' by van 
Benthem~\cite{bent:moda06}, and it is related, and perhaps equivalent in 
expressive power, to the logic $\mathtt{CPDL}$ of \emph{concurrent propositional 
dynamic logic}, cf.~Carreiro~\cite[section 3.2]{carr:frag15} for more 
information).

These logics become particularly interesting in the light of the 
Janin-Walukiewicz Theorem~\cite{jani:expr96}.
Recall that this result states that the modal $\mu$-calculus is the 
bisimulation-invariant fragment of monadic second-order logic (\MSO), in brief: 
$\muML \equiv \MSO/{\bis}$.
For each of the logics $\mu_{Q}\ML$, with $Q \in \{ D, C, A \}$ we can prove 
the following version of this result:
\[
\mu_{Q}\ML \equiv \MSO_{Q}/{\bis},
\]
where $\MSO_{Q}$ is a variant of $\MSO$ where we quantify over a restricted
collection $\funP_{Q}(S)$ of subsets of the model $\bbS$.
More specifically, 
$\funP_{D}(S)$ consists of the so-called \emph{noetherian} sets of a Kripke 
model~\cite{facc:char13},
$\funP_{C}(S)$ is the collection of \emph{finite} subsets of 
$S$~\cite{carr:weak14} (so that $\MSO_{D}$ is \emph{weak} monadic second-order
logic), and 
$\funP_{A}(S)$ is the set of so-called \emph{generalized finite chains} 
in $\bbS$~\cite{carr:pdl15}.

\item
In section~\ref{sec:prel} we proved a strengthened version of the Lyndon 
Theorem for the modal $\mu$-calculus proved by D'Agostino and 
Hollenberg~\cite{dago:logi00}.
In the same vein as the other results in this paper, we can also strengthen
their \emph{{\L}os-Tarski Theorem}.

We say that a formula $\xi \in \muML$ is \emph{preserved under substructures} 
if $\bbS, s \sat \xi$ implies $\bbS',s \sat \xi$, whenever $\bbS'$ is a
substructure of $\bbS$ (in the standard model-theoretic sense).
D'Agostino and Hollenberg proved that a $\mu$-formula $\xi$ is preserved under
substructures iff it is equivalent to a \emph{universal} formula, that is, a
formula in the $\dia$-free fragment $\muML^{U}$ given by the following grammar:
\[
\phi \isbnf 
   q \divbnf \neg q 
   \divbnf \phi\lor\phi \divbnf \phi\land\phi 
   \divbnf \Box\phi 
   \divbnf \mu x.\phi \divbnf \nu x.\phi.
\]
We can reprove this result by our means; given a disjunctive
automaton $\bbD = (D,\Th,\Om)$, define the automaton $\bbD^{U} = (D,\Th^{U},
\Om)$ where $\Th^{U}$ is given by the one-step translation based on the  
clause $(\pi\bullet\nb B)^{U} \isdef \pi \land \Box\bv B$.
One may then show that an initialized disjunctive automaton $\bbD\init{d}$
is preserved under taking substructures iff $\bbD\init{d} \equiv 
\bbD^{U}\init{d}$.
From this the result of D'Agostino and Hollenberg easily follows, and as a 
bonus we find that it is decidable in elementary time whether a given formula
$\xi \in \muML$ has this property.
\end{enumerate}

\subsection*{Questions}

Finally, we mention some open problems for further research.

\begin{enumerate}
\item
It would be interesting to find out the \emph{exact} complexity of 
the problems discussed in this paper.
This would include establishing suitable \emph{lower bounds}.

\item
As mentioned in the section on continuity, it would be good to know whether the
$\mu$-calculus formulas that are \emph{constructive} in $p$ admit a good
syntactic characterization.
In particular, we would like to clarify the connection between the notions of 
continuity and constructivity.
Can we find, for any formula $\xi \in \muML$ which is constructive in $p$, a
continuous formula $\psi$ such that $\mu p. \xi \equiv \mu p.\psi$?

\item
Similarly, we would be curious to see a syntactic characterization of the 
finitely $p$-additive fragment of the modal $\mu$-calculus
(cf.~Remark~\ref{r:finadd}).

\item
While, as already mentioned, some variations of our results are easy to obtain,
there are some \emph{interesting variations} of the problems considered here 
as well.
For instance, it is not so clear how to adapt our characterisation results to
other fixpoint logics like PDL or CTL. 
A second direction to take here would be to look for \emph{coalgebraic
generalisations} of our results. 
In recent years it has been shown that many results on the modal $\mu$-calculus,
including the link with automata theory, can be generalized to the far
wider setting of coalgebraic modal logic~\cite{kupk:coal08,cirs:expt09,%
font:auto10,enqv:comp16}.

\item
As a variation of Theorem~\ref{t:cont}, Gouveia and Santocanale~\cite{gouv:alep17} 
recently gave a characterization of the set of $\muML$-formulas that are 
$\aleph_{1}$-continuous in a fixed proposition letter $p$.
It would be interesting to relate their approach to ours.

\item
Not directly related to the results in this paper, but in our opinion one of 
the most interesting open model-theoretic problems concerning the modal 
$\mu$-calculus is whether $\muML$ admits a natural abstract characterization in 
the form of a \emph{Lindstr\"om theorem}.
Going back to de Rijke~\cite{rijk:lind95}, there are various Lindstr\"{o}m-type
characterizations of basic modal logic (see for 
instance~\cite{bent:lind07,otto:lind08,kurz:coal10,enqv:gene13})
but to the best of our knowledge no abstract characterizations of fixpoint
logics have been established yet.
\end{enumerate}

\section*{Acknowledgment}
\noindent 
The authors are deeply indebted to the two referees for many helpful comments on
an earlier version of this paper.

\bibliographystyle{plain}

\end{document}